\documentclass[11pt,reqno,english,american]{amsart}
\PassOptionsToPackage{natbib=true}{biblatex}
\usepackage{lmodern}
\usepackage[T1]{fontenc}
\usepackage[utf8]{inputenc}
\usepackage{babel}
\usepackage{amsbsy}
\usepackage{amstext}
\usepackage{amssymb}
\usepackage{graphicx}
\usepackage[letterpaper]{geometry}
\usepackage{setspace}
\usepackage[]
 {hyperref}

\makeatletter

\providecommand{\tabularnewline}{\\}

\usepackage[T1]{fontenc}
\usepackage{lmodern}

\usepackage{amsthm}
\theoremstyle{plain}
	\newtheorem{thm}{Theorem}[section]
	\newtheorem{lem}[thm]{Lemma}
	\newtheorem{prop}[thm]{Proposition}

	\newtheorem{assu}{Assumption}[section]

\theoremstyle{definition}

\newtheorem{eg}[thm]{Example}

\theoremstyle{remark}
\newtheorem{rmk}[thm]{Remark}

\numberwithin{equation}{section}

\newcommand{\indep}{\mathop{\perp \! \! \! \perp}}
\renewcommand{\overline}{\bar}
\renewcommand{\epsilon}{\varepsilon}

\AddToHook{package/after/biblatex}{

\renewbibmacro*{name:andothers}{
  \ifboolexpr{
    test {\ifnumequal{\value{listcount}}{\value{liststop}}}
    and
    test \ifmorenames
  }
    {\ifnumgreater{\value{liststop}}{1}
       {\finalandcomma}
       {}%
     \andothersdelim\bibstring[\emph]{andothers}}
    {}}
  
  \DeclareFieldFormat{pages}{#1} 
  \renewbibmacro{in:}{\ifentrytype{article}{}{\printtext{\bibstring{in}\intitlepunct}}} 
}

\usepackage{hyperref}
\hypersetup{pdfusetitle,
            colorlinks,
            pagebackref,
            bookmarksdepth=3,
            citecolor=[RGB]{46,126,42},
            linkcolor=[RGB]{128,0,6},
            urlcolor=[RGB]{138,0,135}}
\usepackage{bookmark}

\usepackage[hyphenbreaks]{breakurl}

\makeatother

\usepackage[style=authoryear,sorting=nyt,dashed=false,maxbibnames=999,maxcitenames=2, uniquename=init,uniquelist=false,giveninits=true,date=year]{biblatex}
\begin{document}
\title[Optimal treatment assignment rules under capacity constraints]{Optimal treatment assignment rules \\
 under capacity constraints}
\date{September 10, 2025 \,(First version: June 13, 2025)}
\author[K.~Sunada]{Keita Sunada}
\address{Department of Economics, University of Rochester, Rochester, NY 14627,
USA.}
\email{\href{mailto:ksunada@ur.rochester.edu}{ksunada@ur.rochester.edu}}
\author[K.~Izumi]{Kohei Izumi}
\address{Department of Economics, University of Rochester, Rochester, NY 14627,
USA.}
\email{\href{mailto:kizumi3@ur.rochester.edu}{kizumi3@ur.rochester.edu}}
\thanks{We are particularly grateful to Nese Yildiz for her guidance and support.
We thank Bin Chen, Kazuma Inagaki, Yukun Ma, Tai Otsu, Yuya Sasaki,
and Takeaki Sunada for their helpful comments and suggestions.}
\begin{abstract}
We study treatment assignment problems under capacity constraints,
where a planner aims to maximize social welfare by assigning treatments
based on observable covariates. Such constraints, common when treatments
are costly or limited in supply, introduce nontrivial challenges for
deriving optimal statistical assignment rules because the planner
needs to coordinate treatment assignment probabilities across the
entire covariate distribution. To address these challenges, we reformulate
the planner's constrained maximization problem as an \emph{optimal
transport} problem, which makes the problem effectively unconstrained.
We then establish local asymptotic optimality results of assignment
rules using a limits of experiments framework. Finally, we illustrate
our method with a voucher assignment problem for private secondary
school attendance using data from \citet{angrist2006long}.

\medskip{}
\end{abstract}

\maketitle

\section{Introduction}

When a social planner allocates goods such as vaccines or school vouchers,
two key challenges often arise. First, the planner typically lacks
knowledge of the true treatment effects. Second, even when treatment
is clearly beneficial, capacity constraints---such as limited budgets
or supplies---may prevent treating everyone. While many studies have
addressed the first challenge, the second has received comparatively
less attention. This paper studies a treatment assignment problem
under such capacity constraints using Le Cam's limits of experiments
framework.

We consider a planner who aims to maximize social welfare by designing
treatment rules based on observable covariates subject to a capacity
constraint. In unconstrained settings, the asymptotically optimal
assignment rule can be determined pointwise: one simply treats individuals
for whom the estimated (conditional) treatment effect is positive
\citep{hirano2009asymptotics}. However, under a capacity constraint,
such pointwise rules may violate the overall treatment quota. A natural
alternative is a plug-in rule that ranks individuals by their estimated
effects and assigns treatment in descending order until the quota
is met. But is this rule still optimal in the presence of constraints?
A key challenge in answering this question stems from the need to
coordinate treatment probabilities across the entire covariate distribution
$F_{X}$, which makes it difficult to obtain an asymptotic representation
of a treatment assignment rule.

To address this, we reformulate the planner's constrained maximization
problem as an \emph{optimal transport} problem. To illustrate, consider
a binary treatment setting where a fixed fraction $p$ of the population
can be treated. If the constraint binds exactly, then the distribution
of treatment assignments must be a Bernoulli distribution with success
probability $p$. The planner’s problem can then be framed as transporting
the mass of $F_{X}$ into the Bernoulli distribution $F_{T}$, in
a way that maximizes the social welfare. Since solutions to this optimal
transport problem automatically satisfy the capacity constraints,
this reformulation simplifies the analysis. This approach is also
computationally attractive, as one can use existing software for optimal
transport. In particular, when both $F_{X}$ and $F_{T}$ are discrete
(or discretized), which is often the case in practice, the optimal
transport problem reduces to a simple linear programming problem. 

Based on this reformulation, we analyze the optimality of two canonical
assignment rules: the plug-in rule and the Bayesian rule. The plug-in
rule is known to be average optimal under point-identified models
when the planner's utility function is smooth \citep{hirano2009asymptotics}.
The Bayesian rule is known to be average optimal under partially identified
models when the planner's utility function is directionally differentiable
\citep{christensen2023optimal}.\footnote{Formally, \citet{christensen2023optimal} adopt an optimality criterion
that is weaker than average optimality. Under additional assumptions
their results imply average optimality as well.} We show that both rules remain average optimal under capacity constraints,
provided that the planner's utility function is smooth.

A meaningful distinction between the two rules arises when the planner's
utility function is only directionally differentiable. This scenario
is not merely a theoretical curiosity but arises in important practical
settings. For instance, it occurs when the (conditional) potential
outcome distributions may differ slightly between the target and training
populations, and the planner adopts a utility function that is robust
to such distributional shifts \citep{adjaho2023externallyvalidpolicychoice}.
In this setting, we show that the plug-in rule is generally not optimal,
whereas the Bayesian rule remains optimal. This result echoes the
findings of \citet{christensen2023optimal}, where the lack of point
identification similarly precludes full differentiability and renders
the plug-in rule suboptimal. 

Our simulation study supports these theoretical results. Specifically:
(i) the Bayesian rule achieves a substantially lower risk than the
plug-in rule in small samples for both the smooth and the directionally
differentiable utility functions, (ii) the two rules behave similarly
under the smooth utility function in larger samples, and (iii) the
Bayesian rule continues to outperform the plug-in rule under the directionally
differentiable utility function in larger samples.

We illustrate our methods using data from \citet{angrist2006long},
who study the impact of receiving a randomly assigned voucher (which
allows students to attend private secondary schools) on educational
attainment seven years later. We hypothetically treat the marginal
distribution of covariates (age and sex) in the observed sample as
that of the target population, and compute both the plug-in rule and
the Bayesian rule. The two rules produce identical allocations under
the smooth utility function, but they produce different allocations
under the directionally differentiable utility function.

\subsection{Related literature}

This study builds on the literature on the statistical treatment assignment
problems in econometrics, where the pioneering works include \citet{manski2004statistical}
and \citet{dehejia2005program}. Within this expanding literature,
our main contribution is to provide a decision-theoretic optimality
result under capacity constraints on available treatments. Previous
studies have established the decision-theoretic optimality of treatment
rules in several settings including: (i) point-identified smooth (semi-)parametric
models under local asymptotics \citep{hirano2009asymptotics,masten2023minimax},
(ii) partially-identified smooth (semi-)parametric models under local
asymptotics \citep{christensen2023optimal,kido2023locally,xu2024jmp},
(iii) point-identified models with finite samples \citep{stoye2009minimax,stoye2012minimax,tetenov2012statistical,guggenberger2024minimax,kitagawa2024treatmentchoicenonlinearregret,chen2025note},
(iv) partially-identified models with finite samples \citep{manski2007minimax,stoye2012minimax,yata2023optimal,ishihara2021evidence,fernandez2024robust,oleaDecisionTheoryTreatment2024}.
However, none of these studies consider the capacity constraints in
the way we do. To the best of our knowledge, this is the first study
to establish an optimality result under such constraints within the
framework of \citet{hirano2009asymptotics}, extended to cover non-binary
treatments---both discrete and continuous.

Besides \citet{hirano2009asymptotics}, the most closely related paper
is \citet{christensen2023optimal}. They extend \citet{hirano2009asymptotics}
to allow for partially identified parameters and (non-randomized)
discrete actions, and show the asymptotic optimality of the Bayesian
rule. Our setting differs in that the action space is the set of couplings
of $F_{X}$ and $F_{T}$ (see notation below) equipped with the Wasserstein
distance, which naturally allows for randomized assignment rules.
This non-standard action space requires nontrivial extensions to analyze
the asymptotic properties of the Bayesian rule, which we address by
drawing on tools from optimal transport. In contrast to their framework,
we focus on point-identified parameters.\footnote{\citet{xu2024jmp} further extends the framework of \citet{christensen2023optimal}
to continuous decision problems using an expansion-based approach.}

There also exist studies that incorporate exogenously given constraints
into the treatment assignment problem. \citet{bhattacharya2012inferring}
impose the capacity constraints in the way we do, but focus on the
estimation and inference of a nonparametric plug-in rule. Other papers
adopting the \emph{empirical welfare maximization} approach allow
for various types of constraints, including the capacity constraints
\citep{kitagawa2018should,athey2021policy,mbakop2021model,sun2024empiricalwelfaremaximizationconstraints}.
\citet{kitagawa2018should} show the optimality of their proposed
rule in terms of the welfare convergence rate, which measures how
quickly the average welfare achieved by the proposed rule converges
to the maximum welfare under the true data generating process. 

Some recent works utilize tools from optimal transport theory in the
literature of treatment assignment problems. \citet{kidoDistributionallyRobustPolicy2022}
and \citet{adjaho2023externallyvalidpolicychoice} study the external
validity of treatment choices by measuring the difference in potential
outcome distributions between the training and target populations
using the Wasserstein distance. \citet{hazard2025Who} formulate a
learning problem of optimal matching policies in a two-sided market
as an empirical optimal transport problem, and derive a welfare regret
bound for their estimated policy. 

Our work also relates to the growing field of statistical methods
for optimal transport problems \citep{chewi2024statisticaloptimaltransport},
as we study the local asymptotic properties of transport maps of an
optimal transport problem where the cost function is indexed by parameters
that can be efficiently estimated.

\subsection{Structure of the paper}

The remainder of the paper is organized as follows. Section \ref{sec:Setting}
formulates the planner's problem and introduces the data generating
process. Section \ref{sec:Optimal-decisions} introduces the decision
theoretic framework and define the plug-in rule and the Bayesian rule.
Then the optimality results are stated. Section \ref{sec:Simulation}
provides a simulation study to evaluate the finite sample performance
of rules. Section \ref{sec:Empirical-application} illustrates our
methods using the data from \citet{angrist2006long}. Finally, Section
\ref{sec:Conclusion} concludes. All of the proofs are relegated to
Appendix.

\subsection{Notation}

A function $f:\Theta\subset\mathbb{R}^{k}\to\mathbb{R}$ is \emph{(Hadamard)
directionally differentiable }at $\theta_{0}$ if there is a continuous
function $\dot{f}_{\theta_{0}}:\mathbb{R}^{k}\to\mathbb{R}$ such
that 
\[
\lim_{n\to\infty}\left|\frac{f(\theta_{0}+t_{n}h_{n})-f(\theta_{0})}{t_{n}}-\dot{f}_{\theta_{0}}(h)\right|=0
\]
for all sequences $\left\{ t_{n}\right\} \subset\mathbb{R}_{+}$ and
$\left\{ h_{n}\right\} \subset\mathbb{R}^{k}$ such that $t_{n}\downarrow0$,
$h_{n}\to h\in\mathbb{R}^{k}$ as $n\to\infty$ and $\theta_{0}+t_{n}h_{n}\in\Theta$
for all $n$. It is worth noting that this requires $\dot{f}_{\theta_{0}}$
need to be continuous, but not to be linear.

Let $P$ and $Q$ be Borel probability measures on $\mathcal{A}$
and $\mathcal{B}$, respectively. A joint distribution $\mu$ on $\mathcal{A}\times\mathcal{B}$
is called a \emph{coupling }of $P$ and $Q$ if its marginals are
$P$ and $Q$; that is, $\mu(A\times\mathcal{B})=P(A)$ and $\mu(\mathcal{A}\times B)=Q(B)$
for any measurable sets $A$ and $B$. 

\section{Setting\label{sec:Setting}}

The setting of this paper closely follows that of \citet{hirano2009asymptotics}.
We consider a social planner who assigns a treatment $T$ to individuals
based on their observable covariates $X$. Let $F_{X}$ denote the
marginal distribution of $X$ in the target population, with support
$\mathcal{X}$. We assume that $F_{X}$ is known to the planner. The
planner can fractionally (probabilistically) assign treatment $T=t$
to an individual with covariate $X=x$.

Let $Y(t)$ denote potential outcomes under treatment $T=t$. In contrast
to \citet{hirano2009asymptotics}, we distinguish between the conditional
potential outcome distribution in the target population and in the
training population. We denote the conditional distribution of $Y(t)$
in the training population as $F_{t}(\cdot|x,\theta)$, where $F_{t}(\cdot|x,\theta)$
belongs to families of distributions indexed by a parameter $\theta\in\Theta\subset\mathbb{R}^{k}$.
The planner must learn $\theta$ from the available data from experimental
or observational studies. 

\subsection{The planner's preferences}

The planner's utility for assigning treatment $T=t$ to an individual
with covariate $X=x$ depends on the conditional distribution $F_{t}(\cdot|x,\theta)$
via a functional $w$. For the shorthand notation, we write 
\[
w(\theta,x,t):=w(F_{t}(\cdot|x,\theta)).
\]

We consider two scenarios. First, $w(\theta,x,t)$ is fully differentiable
in $\theta$. Second, $w(\theta,x,t)$ is only directionally differentiable
in $\theta$. Two examples corresponding to each scenario are given
as follows.

\begin{eg} \label{eg:bayesian} When the planner is interested in
the (conditional) mean outcome, then a natural choice of utility function
is 
\[
w(\theta,x,t)=\int y\mathrm{d}F_{t}(y|x,\theta).
\]
This choice is standard in the literature and appropriate especially
when the target and training populations are assumed to have the same
conditional potential outcome distribution. $\blacksquare$\end{eg}

\begin{eg} \label{eg:maxmin} When the conditional potential outcome
distributions may differ slightly across the target and training populations,
the planner may wish to adopt a utility function that is robust to
distributional shifts. Following \citet{adjaho2023externallyvalidpolicychoice},
we formalize such a utility function.\footnote{For alternative approaches to robust welfare, see \citet{si2020distributionally},
\citet{kidoDistributionallyRobustPolicy2022}, and \citet{qi2023robustness}. } We first define an $\varepsilon$-neighborhood of $F_{t}(\cdot|x,\theta)$
as
\[
\mathcal{N}_{\varepsilon}:=\left\{ G_{t}(\cdot|x):\tilde{d}_{W}(G_{t}(\cdot|x),F_{t}(\cdot|x,\theta))\le\varepsilon\right\} ,
\]
where $\varepsilon>0$ is a measure of neighborhood size, and $\tilde{d}_{W}$
is the Wasserstein distance of order 1.\footnote{Formally, 
\[
\tilde{d}_{W}(F_{t}(\cdot|x),G_{t}(\cdot|x,\theta)):=\inf_{\pi\in\Pi(F_{t},G_{t})}\int\left|y-\tilde{y}\right|\mathrm{d}\pi(y,\tilde{y}),
\]
where $\Pi(F_{t},G_{t})$ denotes all couplings of $F_{t}$ and $G_{t}$.
Note that $\tilde{d}_{W}$ differs from $d_{W}$ defined in (\ref{eq:d_W}).} Given $\lambda\in[0,1]$, the planner's utility for assigning treatment
$T=t$ for individuals with covariate $X=x$ is then defined as
\[
w(\theta,x,t)=\lambda\int y\mathrm{d}F_{t}(y|x,\theta)+(1-\lambda)\inf_{G_{t}(\cdot|x)\in\mathcal{N}_{\varepsilon}}\int y\mathrm{d}G_{t}(y|x).
\]
This formulation corresponds to maxmin preferences of an ambiguity-averse
decision maker \citep{gilboa1989maxmin}. In the second term, since
the true target distributions are unknown, the planner computes the
conditional mean under the worst-case distribution within the neighborhood
of $F_{t}(\cdot|x,\theta)$, treating it as a fixed reference prior
on $Y(t)$ given $x$.\footnote{Several recent studies incorporate non-Bayesian preferences that arise
naturally in settings with set-identifiable parameters \citep{giacominiRobustBayesianInference2021Econometrica,fernandez2024robust,christensen2023optimal}.
\citet{banerjee2020theory} adopt maxmin preferences to study of the
optimal experimental design problems.}

By \citet[Remark 2.2]{adjaho2023externallyvalidpolicychoice}, this
utility function can be rewritten as 
\begin{equation}
w(\theta,x,t)=\lambda\int y\mathrm{d}F_{t}(y|x,\theta)+(1-\lambda)\max\left\{ \int y\mathrm{d}F_{t}(y|x,\theta)-\varepsilon,y_{\ell}(t)\right\} ,\label{eq:maxmin2}
\end{equation}
where $y_{\ell}(t)$ is the possible minimum values in the support
of $Y(t)$. From this expression, one finds that (\ref{eq:maxmin2})
is only directionally differentiable, and the directional derivative
of the second term of (\ref{eq:maxmin2}) for direction $h$ at $\theta$
is given by $(1-\lambda)$ times
\[
\begin{cases}
\left(\frac{\partial}{\partial\theta}\int y\mathrm{d}F_{t}(y|x,\theta)\right)^{\top}h & \text{if }\int y\mathrm{d}F_{t}(y|x,\theta)-\varepsilon>y_{\ell}(t),\\
\max\left\{ \left(\frac{\partial}{\partial\theta}\int y\mathrm{d}F_{t}(y|x,\theta)\right)^{\top}h,0\right\}  & \text{if }\int y\mathrm{d}F_{t}(y|x,\theta)-\varepsilon=y_{\ell}(t),\\
0 & \text{if }\int y\mathrm{d}F_{t}(y|x,\theta)-\varepsilon<y_{\ell}(t).
\end{cases}
\]
We note that the non-linearity of this directional derivative corresponds
to the failure of the full differentiability of (\ref{eq:maxmin2})
at $\theta$. $\blacksquare$ \end{eg}

\subsection{The planner's problem as optimal transport}

We consider a setting where the planner faces the capacity constraints
on the available treatments. To illustrate, consider a simple binary
treatment case, $T\in\mathcal{T}:=\left\{ 0,1\right\} $, where a
fraction $p$ of the target population is to be treated. We assume
that the capacity constraint binds exactly. Let $\mu(t|x)$ denote
the conditional probability of assigning treatment $T=t$ to individuals
with $X=x$. Then, under the capacity constraint, the planner's problem
can be written as
\begin{align}
 & \max_{\mu(\cdot|\cdot)}\int_{\mathcal{X}}\left\{ w(\theta,x,1)\mu(1|x)+w(\theta,x,0)\mu(0|x)\right\} \mathrm{d}F_{X}(x),\label{eq:constrained_problem}\\
 & \text{s.t. }\int\mu(1|x)\mathrm{d}F_{X}(x)=p.\nonumber 
\end{align}

We now show how to convert this constrained optimization problem into
an unconstrained one. Observe that the distribution of treatment assignments
must be a Bernoulli distribution with the success probability $p$.
With this in mind, the planner’s problem can be seen as an optimal
transport problem: the planner transports the mass of $F_{X}$ into
$F_{T}$, the Bernoulli distribution, in a way that maximizes the
social welfare. 

Formally, let $\mathcal{M}_{a}$ be the set of all couplings of $F_{X}$
and $F_{T}$. Let $d$ be a distance function on $\mathcal{X}\times\mathcal{T}$,
and define the Wasserstein distance of order 1 as
\begin{equation}
d_{W}(\mu,\nu):=\inf_{\gamma\in\Gamma(\mu,\nu)}\int d((x,t),(x^{\prime},t^{\prime}))\mathrm{d}\gamma((x,t),(x^{\prime},t^{\prime})),\label{eq:d_W}
\end{equation}
where $\Gamma(\mu,\nu)$ is the set of couplings whose marginals are
$\mu$ and $\nu$. We focus on couplings that have a finite first
moment: 
\[
\mathcal{M}:=\left\{ \mu\in\mathcal{M}_{a}:\int d((x_{0},t_{0}),(x,t))\mathrm{d}\mu<+\infty\right\} ,
\]
for some arbitrary $(x_{0},t_{0})\in\mathcal{X}\times\mathcal{T}$.
Then $(\mathcal{M},d_{W})$ becomes a metric space, and $d_{W}$ is
finite on $\mathcal{M}$ \citep[Theorem 6.9]{villani2009optimal}.
Using this setup, the original constrained problem (\ref{eq:constrained_problem})
is equivalent to the following:
\begin{equation}
\max_{\mu\in\mathcal{M}}W(\theta,\mu),
\end{equation}
where 
\begin{align*}
W(\theta,\mu) & :=\int_{\mathcal{X}\times\mathcal{T}}w(\theta,x,t)\mathrm{d}\mu=\int_{\mathcal{X}}\int_{\mathcal{T}}w(\theta,x,t)\mathrm{d}\mu(t|x)\mathrm{d}F_{X}(x).
\end{align*}
Hence, the action space of the planner is the space of couplings $(\mathcal{M},d_{W})$.
We remark that $\mathrm{d}\mu(t|x)$ becomes a conditional probability
measure when $F_{T}$ is continuous.

There are three important remarks regarding this optimal transport
formulation. First, the capacity constraint is automatically satisfied
by any coupling in $\mathcal{M}$, making the problem effectively
unconstrained. Second, this reformulation is computationally attractive,
as one can use an existing software for optimal transport. In particular,
when both $F_{X}$ and $F_{T}$ are discrete---which is often the
case in practice---the optimal transport problem reduces to a simple
linear programming problem. Finally, this approach can easily accommodate
non-binary treatment settings. We assume that $T$ follows a distribution
$F_{T}$, determined by the capacity constraints, with support $\mathcal{T}$.

\subsection{The data generating process}

When the true (finite-dimensional) parameter $\theta_{0}$ is known
to the planner, the optimal rules can be obtained by solving
\[
\max_{\mu\in\mathcal{M}}W(\theta_{0},\mu).
\]
However, since the planner does not know $\theta_{0}$ in practice,
she must select a rule in a data-driven manner. For this purpose,
data $Z^{n}=(Z_{1},\dots,Z_{n})$, which are informative about $\theta$
(and hence about $F_{t}(\cdot|x,\theta)$) are available from a training
population. We assume that the data $Z^{n}$ are i.i.d. with $Z_{i}\sim P_{\theta}$
on some space $\mathcal{Z}$ equipped with the Borel $\sigma$-algebra
$\mathcal{B}(\mathcal{Z})$. We let $P_{\theta}^{n}$ denote the joint
probability measure of $Z^{n}$. In Appendix \ref{sec:Semipara},
we consider an extension in which the sampling distribution of data
may depend on (possibly infinite-dimensional) nuisance parameters,
as in a GMM model. 

\begin{eg}\label{eg:ABK} \citet{angrist2006long} study the medium-term
effects of the PACES program in Columbia.\footnote{PACES stands for Programa de Ampliación de Cobertura de la Educación
Secundaria.} Specifically, they investigate the impact of receiving a randomly
assigned voucher (which allows attendance at private secondary schools)
on educational attainment seven years later, using the administrative
data. 

In one of their main specifications, they consider the following linear
model:
\[
Y_{i}=X_{i}^{\top}\beta+\alpha T_{i}+u_{i},
\]
where $Y_{i}$ denotes test scores, $X_{i}$ includes observable covariates
(sex and age), $T_{i}$ is an indicator for the treatment status,
and $u_{i}$ is an error term. However, not all individuals in the
sample took the exam. Because the voucher recipients were more likely
to take the exam than non-recipients, a selection issue arises. To
address this, \citet{angrist2006long} construct a modified test score
variable by censoring the observed scores at a specific quantile of
the test score distribution. Let $R_{i}$ be an indicator for exam
registration and $\tau>0$ be the censoring threshold. Then the censored
test score is defined as 
\[
Y_{i}(\tau):=\boldsymbol{1}\left\{ R_{i}Y_{i}\ge\tau\right\} Y_{i}+\boldsymbol{1}\left\{ R_{i}Y_{i}<\tau\right\} \tau.
\]
Under the assumptions that (i) $u_{i}$ is normally distributed with
mean zero, and (ii) any untested student would have scored below the
threshold $\tau$ had they taken the exam, the parameters can be consistently
estimated using a Tobit-type maximum likelihood estimator. In this
context, the parameter is $\theta=(\alpha,\beta,\sigma)$, and the
observed data is $Z^{n}=\left\{ \left(T_{i},X_{i},Y_{i}(\tau)\right):i=1,\dots,n\right\} $,
where $\sigma$ is the standard deviation of the error term $u_{i}$.
$\blacksquare$ \end{eg}

Following \citet{hirano2009asymptotics} and among others, we use
a local asymptotic framework where we perturb the data-generating
process around the true one. Let $\Theta$ be an open subset of $\mathbb{R}^{k}$
and suppose that $\theta_{0}$ is the true parameter. Let $\theta_{nh}:=\theta_{0}+h/\sqrt{n}$.
We assume that the sequence of experiments $\mathcal{E}_{n}=\left\{ P_{\theta}^{n}:\theta\in\Theta\right\} $
satisfies \emph{differentiability in quadratic mean} (DQM) at $\theta_{0}$:
there exists a function $s:\mathcal{Z}^{n}\to\mathbb{R}^{k}$ such
that 
\begin{equation}
\int\left[\mathrm{d}P_{\theta_{0}+h}^{1/2}(z)-\mathrm{d}P_{\theta_{0}}^{1/2}(z)-\frac{1}{2}h^{\prime}s(z)\mathrm{d}P_{\theta_{0}}^{1/2}(z)\right]^{2}=o(\left\lVert h\right\rVert ^{2})\quad\text{as }h\to0,
\end{equation}
where $s$ is the score function associated with $\mathcal{E}_{1}$.
Let $I_{0}=\mathbb{E}_{P_{\theta_{0}}^{n}}[ss^{\prime}]$.

The planner's \emph{statistical treatment assignment rule }(or just
\emph{rule}) $\mathcal{\mu}:\mathcal{Z}^{n}\to\mathcal{M}$ maps realizations
of data into the coupling. Let 
\[
A_{0}:=\arg\max_{\mu\in\mathcal{M}}W(\theta_{0},\mu)
\]
 be the set of couplings that maximize the welfare at $\theta_{0}$.
We define the class of sequences of rules by 
\begin{equation}
\mathcal{D}:=\left\{ \left\{ \mu_{n}\right\} :\mu_{n}(Z^{n})\stackrel{h}{\rightsquigarrow}Q_{\theta_{0},h}\text{ and }\sqrt{n}P_{\theta_{nh}}^{n}\left(\mu_{n}(Z^{n})\notin A_{0}\right)\to0\quad\forall h\in\mathbb{R}^{k},\forall\theta_{0}\in\Theta\right\} ,
\end{equation}
where $\stackrel{h}{\rightsquigarrow}$ denotes convergence in distribution
along $P_{\theta_{nh}}^{n}$ with $Z^{n}\sim P_{\theta_{nh}}^{n}$
for each $n$, and $Q_{\theta_{0},h}$ is a (possibly degenerate)
probability measure on $\mathcal{M}$. For technical reasons, we restrict
our analysis to rules satisfying $\sqrt{n}P_{\theta_{nh}}^{n}\left(\mu_{n}(Z^{n})\notin A_{0}\right)\to0$,
a condition also imposed by \citet{christensen2023optimal}. This
condition ensures that the treatment rule maximizes the welfare at
the true parameter $\theta_{0}$ with high probability in a neighborhood
of $\theta_{0}$, and that the probability of selecting a suboptimal
coupling (i.e., $\mu_{n}\not\in A_{0}$) vanishes sufficiently fast. 

Since $(\mathcal{M},d_{W})$ is compact (and thus complete and separable)
by \citet[Remark 6.19]{villani2009optimal}, we obtain the following
asymptotic representation theorem by \citet[Theorem 3.1]{van1991asymptotic}. 

\begin{prop}Let $\left\{ \mu_{n}\right\} \in\mathcal{D}$ satisfy
$\mu_{n}\overset{h}{\rightsquigarrow}Q_{\theta_{0},h}$ for all $h\in\mathbb{R}^{k}$
and $\theta_{0}\in\Theta$. Under the DQM condition, there exists
$\mu_{\infty}:\mathbb{R}^{k}\times[0,1]\to\mathcal{M}$ such that
for every $h\in\mathbb{R}^{k}$, $\mathcal{L}_{h}\left(\mu_{\infty}(\Delta,U)\right)=Q_{\theta_{0},h}$,
$\mathcal{L}_{h}\left(\Delta\right)=N(h,I_{0}^{-1})$, and $\mathcal{L}_{h}(U)=\mathrm{Unif}[0,1]$
with $U\indep\Delta$. \label{prop:ART} \end{prop}

This proposition states that any sequence $\left\{ \mu_{n}\right\} $
in $\mathcal{D}$ can be matched by some treatment rule $\mu_{\infty}$
in a \emph{limit experiment} where we observe a single draw $\Delta$
from a shifted normal distribution and an independent uniform random
variable $U$. This representation is useful for analyzing the asymptotic
properties of rules, as the limit experiment is more analytically
tractable than the original experiments $\mathcal{E}_{n}$. 

\begin{rmk} \citet{hirano2009asymptotics} study the optimality of
treatment assignment rules after conditioning on a fixed covariate
value $X=x$. In our notation, this is equivalent to finding an optimal
conditional probability $\mu(\cdot|x)$ for each $x$. Accordingly,
they derive the asymptotic representation of $\mu(\cdot|x)$ as a
function of $Z^{n}$ by applying a version of the representation theorem
(see their Proposition 3.1). However, because it is difficult to accommodate
the capacity constraints within this framework, we instead apply the
representation theorem to couplings $\mu\in\mathcal{M}$. Note that
in our setup, the map $Z^{n}\mapsto\mu(Z^{n})$ takes values in $\mathcal{M}$,
rather than the unit interval. \end{rmk}

\section{Optimal decisions\label{sec:Optimal-decisions}}

\subsection{Decision theoretic framework and rules\label{subsec:framework_rules}}

We begin by introducing a decision theoretic framework to evaluate
the performance of a sequence of rules $\left\{ \mu_{n}\right\} \in\mathcal{D}$.
Let 
\[
W_{\mathcal{M}}^{*}(\theta):=\max_{\mu^{\prime}\in\mathcal{M}}W(\theta,\mu^{\prime})
\]
 denote the maximum attainable welfare at $\theta$. Following the
literature, we define the \emph{welfare regret} $W_{\mathcal{M}}^{*}(\theta)-W(\theta,\mu)$
as the loss incurred from choosing a coupling $\mu\in\mathcal{M}$.
Accordingly, the \emph{risk} associated with the map $Z^{n}\mapsto\mu(Z^{n})\in\mathcal{M}$
at $\theta$ is given by
\begin{equation}
R(\mu,\theta):=\mathbb{E}_{P_{\theta}^{n}}\left[W_{\mathcal{M}}^{*}(\theta)-W(\theta,\mu(Z^{n}))\right],
\end{equation}
where the expectation is taken with respect to the sampling distribution
$P_{\theta}^{n}$ of $Z^{n}$. The planner's objective is to minimize
the risk by constructing data-driven rules $\left\{ \mu_{n}\right\} \in\mathcal{D}$.

Let $\pi$ be any prior density function on $\Theta$ that is continuous
and positive at $\theta_{0}$. A sequence of rules $\left\{ \mu_{n}^{*}\right\} \in\mathcal{D}$
is said to be \emph{average optimal} if $\left\{ \mu_{n}^{*}\right\} $
attains the infimum of the asymptotic risk function: 
\begin{equation}
\inf_{\left\{ \mu_{n}\right\} \in\mathcal{D}}\liminf_{n\to\infty}\int\sqrt{n}R(\mu_{n},\theta_{nh})\pi(\theta_{nh})\mathrm{d}h.
\end{equation}

Our main goal is to construct a sequence of rules that is average
optimal. A natural candidate is the plug-in rule, which has been shown
to be average optimal by \citet{hirano2009asymptotics} in settings
without capacity constraints. To formalize this rule, let $\hat{\theta}_{n}$
be a \emph{best regular estimator} of $\theta_{0}$ such that 
\[
\sqrt{n}(\hat{\theta}_{n}-\theta_{nh})\stackrel{h}{\rightsquigarrow}N(0,I_{0}^{-1})\quad\text{as }n\to\infty.
\]
The maximum likelihood estimator or the Bayesian posterior mean estimator
are typical examples of best regular estimators. The \emph{plug-in
rule }is the sequence $\{\mu_{n}^{P}(Z^{n})\}$, where for each $n$,
\[
\mu_{n}^{P}(Z^{n})\in\arg\max_{\mu\in\mathcal{M}}W(\hat{\theta}_{n},\mu).
\]

Another rule we examine is the \emph{Bayesian rule} $\{\mu_{n}^{B}(Z^{n})\}$,
defined for each $n$ by 
\begin{align*}
\mu_{n}^{B}(Z^{n}) & \in\arg\max_{\mu\in\mathcal{M}}\int_{\Theta}W(\theta,\mu)\pi_{n}(\theta|Z^{n})\mathrm{d}\theta,
\end{align*}
where $\pi_{n}(\theta|Z^{n})$ is the posterior density obtained from
a strictly positive, continuous prior density $\pi$ on $\Theta$.
\citet{christensen2023optimal} show the average optimality of such
Bayesian rules in discrete choice problems when (i) the model is partially
identified, and (ii) decision rules are not fractional. 

\subsection{\label{subsec:An-optimality-result}An optimality result for the
Bayesian rule}

We provide an optimality result for a Bayesian rule under the directional
differentiability of $w$. We discuss asymptotic properties of the
plug-in rule in Subsection \ref{subsec:plugin-asym}. We impose the
following assumptions for results in this subsection. As in \citet{clarke2002information}
and \citet{christensen2023optimal}, we say that a family $\mathcal{P}$
is \emph{locally quadratic }if for any $\theta_{0}\in\Theta$, $D_{\mathrm{KL}}(p_{\theta}\mathrel{\Vert}p_{\theta^{\prime}})\le2(\theta-\theta^{\prime})^{\top}I_{0}(\theta-\theta^{\prime})$
holds for any $\theta$ and $\theta^{\prime}$ belonging to a neighborhood
of $\theta_{0}$, where $D_{\mathrm{KL}}(p_{\theta}\mathrel{\Vert}p_{\theta^{\prime}})$
is the Kullback-Leibler divergence with respect to a common dominating
measure $\nu$. Also, we say $\mathcal{P}$ is \emph{sound} if weak
convergence of $P_{\theta}$ to $P_{\theta^{\prime}}$ is equivalent
to the convergence of $\theta$ to $\theta^{\prime}$ for probability
measures $P_{\theta},P_{\theta^{\prime}}$ and parameters $\theta,\theta^{\prime}\in\Theta$.

\begin{assu}\label{assu:cms25:assu2} (i) $\Theta$ is open. (ii)
$\mathcal{P}$ is DQM at any $\theta_{0}$. (iii) $I_{0}$ is finite
and nonsingular at any $\theta_{0}$. (iv) $\mathcal{P}$ is locally
quadratic. (v) $\mathcal{P}$ is sound. \end{assu}

The first three conditions are standard assumptions in local asymptotic
frameworks. Conditions (iv) and (v) ensure that Schwartz’s theorem---which
originally establishes posterior consistency in a space of density
functions (see, for example, \citet{ghosh2003Bayesian} and \citet{ghosal2017fundamentals})---can
be applied in a parametrized setting, as in \citet{clarke2002information}.
We impose these assumptions to show that the Bayesian rule $\mu_{n}^{B}$
satisfies $\sqrt{n}P_{\theta_{nh}}^{n}\left(\mu_{n}^{B}(Z^{n})\notin A_{0}\right)\to0$.
The same conditions are also imposed by \citet{christensen2023optimal}. 

\begin{assu}\label{assu:villani_thm6-9} $\mathcal{X}\times\mathcal{T}$
is compact in the product metric space where the distance function
$d$ is equipped. \end{assu}

For example, this condition is satisfied if $\mathcal{X}$ is a compact
metric space, and $\mathcal{T}$ is a finite discrete space.

\begin{assu}\label{assu:w} (i) $w(\theta,x,t)$ is bounded continuous
on $\Theta$ for any $(x,t)\in\mathcal{X}\times\mathcal{T}$. (ii)
$w(\theta,x,t)$ is continuous on $\mathcal{X}\times\mathcal{T}$
uniformly over $\Theta$. \end{assu}

Note that discrete covariates are compatible with condition (ii),
since the metric $d$ can be defined to incorporate the discrete metric.

\begin{assu}\label{assu:w-dot} (i) $w(\theta,\cdot)$ is directionally
differentiable (as a function on $\mathcal{X}\times\mathcal{T}$)
at any $\theta\in\Theta$ with derivative $\dot{w}_{\theta}$.\footnote{That is, for any $r_{n}\downarrow0$ and $h_{n}\to h$, 
\[
\max_{(x,t)\in\mathcal{X}\times\mathcal{T}}\left|\frac{w(\theta_{0}+r_{n}h_{n},x,t)-w(\theta_{0},x,t)}{r_{n}}-\dot{w}_{\theta_{0}}(x,t;h)\right|\to0.
\]
} (ii) $\dot{w}_{\theta_{0}}(x,t;h)$ is continuous on $\mathcal{X}\times\mathcal{T}$
for any $h$. (iii) $\dot{w}_{\theta_{0}}(\cdot;h)$ is uniformly
dominated by a function $K(h)$ that grows at most subpolynomially
of order $p$; i.e., $\max_{(x,t)\in\mathcal{X}\times\mathcal{T}}\left|\dot{w}_{\theta_{0}}(x,t;h)\right|\le K(h)\le1+\left\lVert h\right\rVert ^{p}$.

\end{assu}

Condition (i) imposes a uniform version of directional differentiability,
rather than requiring it only pointwise in $(x,t)$. For condition
(iii), a similar polynomial growth condition appears in the study
of Bayes estimators \citep[Section 10.3]{van2000asymptotic}. Choosing
$p=1$ is sufficient for (\ref{eq:maxmin2}) provided $\max_{(x,t)\in\mathcal{X}\times\mathcal{T}}\left\lVert \frac{\partial}{\partial\theta}\int y\mathrm{d}F_{t}(y|x,\theta_{0})\right\rVert <\infty$.

\begin{assu}\label{assu:prior} (i) The prior (Lebesgue) density
function $\pi$ is positive, continuous, and bounded on $\Theta$
. (ii) $\int\left\lVert \theta\right\rVert ^{p}\pi(\theta)\mathrm{d}\theta<\infty$.
\end{assu}

The order $p$ used in condition (ii) must align with the order in
Assumption \ref{assu:w-dot} (iii). 

\begin{assu}\label{assu:CMS25:lem:8} There exists $K$ such that
for all $(\mu,\nu)\in A_{0}\times(\mathcal{M}\setminus A_{0})$, $W(\theta_{0},\mu)>K\ge W(\theta_{0},\nu).$
 \end{assu}

Note that $W(\theta_{0},\mu)$ is constant over $A_{0}$. This condition
requires that the value of $W(\theta_{0},\cdot)$ is uniformly separated
between $A_{0}$ and $\mathcal{M}\setminus A_{0}$. The requirement
arises because $\mathcal{M}$ is infinite; it is unnecessary when
the action space is finite. To see an implication from this condition,
note that the correspondence $A(\theta):=\arg\max_{\mu\in\mathcal{M}}W(\theta,\mu)$
is upper hemicontinuous at $\theta_{0}$ by the theorem of maximum
of Berge. From this observation, one can show that Assumption \ref{assu:CMS25:lem:8}
implies that for sufficiently small $\varepsilon>0$ we have that
$A(\theta)=A_{0}$ for all $\theta\in N_{\varepsilon}(\theta_{0})$,
which means that $A(\theta)$ is invariant around the neighborhood
of $\theta_{0}$. 

We note that the sequence $\{\mu_{n}^{B}(Z^{n})\}$ of Bayesian rules
may not be uniquely determined as our framework allows for multiple
maximizers of the objective function. This non-uniqueness complicates
the analysis since we cannot directly apply the argmax theorem, which
is often used to study the asymptotic behavior of general argmax-functionals. 

To deal with this, we utilize a penalized version of the Bayesian
rule. Let $\nu\in\mathcal{M}$ be any fixed reference measure and
$H:\mathcal{M}\to\mathbb{R}_{+}$ be a functional given by $\mu\mapsto\left(d_{W}\left(\mu,\nu\right)\right)^{2}$,
which will serve as a penalty function of a maximization problem.
For example, we can let $\nu$ be the product measure of $F_{X}$
and $F_{T}$. The functional $H$ has following properties:

\begin{prop}\label{prop:penalty} $H$ is a nonnegative, continuous,
strictly convex, and bounded functional on $(\mathcal{M},d_{W})$.
\end{prop}

\begin{proof} See Appendix \ref{sec:penalty}. \end{proof}

Then we define the penalized Bayesian rule by 
\[
\mu_{n,\varepsilon}^{B}(z):=\arg\max_{\mu\in\mathcal{M}}\int\sqrt{n}W(\theta,\mu)\pi_{n}(\theta|z)\mathrm{d}\theta-\varepsilon H(\mu),
\]
for $\varepsilon>0$. Note that $\mu_{n,\varepsilon}^{B}$ becomes
the unique maximizer of this penalized problem by the strict convexity
of $H$. We then obtain a useful result on the penalized rules $\{\mu_{n,\varepsilon}^{B}\}$
by following the arguments of \citet{nutzIntroductionEntropicOptimal2022}.\footnote{\citet{nutzIntroductionEntropicOptimal2022} provides corresponding
results by choosing $H$ as the Kullback-Leibler (KL) information
criterion between $\mu\in\mathcal{M}$ and any reference measure in
$\mathcal{M}$. KL is nonnegative and strictly convex in $\mu$, but
not continuous and bounded. Here we impose stronger requirements for
the penalty function $H$, which is needed to handle weak convergence
of functionals on $\mathcal{M}$ to study the asymptotic properties
of rules. Accordingly, the mode of convergence of $\mu_{n,\varepsilon}^{B}(z)$
is modified to weak convergence from convergence in total variation,
see \citet[Theorem 5.5]{nutzIntroductionEntropicOptimal2022}.} 

\begin{prop}\label{prop:nutz} Let $\mathcal{M}_{opt}(z):=\arg\max_{\mu\in\mathcal{M}}\int\sqrt{n}W(\theta,\mu)\pi_{n}(\theta|z)\mathrm{d}\theta$.
For each $z\in\mathcal{Z}^{n}$, there exists a unique $\mu_{n}^{B}(z)\in\mathcal{M}$
such that (i) $\mu_{n,\varepsilon}^{B}(z)$ converges weakly to $\mu_{n}^{B}(z)$
as $\varepsilon\downarrow0$, (ii) $\mu_{n}^{B}(z)\in\mathcal{M}_{opt}(z)$,
and (iii) $\mu_{n}^{B}(z)=\arg\min_{\mu\in\mathcal{M}_{opt}(z)}H(\mu)$.
\end{prop}

\begin{proof} See Appendix \ref{sec:Proof_Lemma_Nutz}. \end{proof}

Thus, we can construct a unique $\{\mu_{n}^{B}(z)\}$ where $\mu_{n}^{B}(z)$
minimizes the penalty function $H$ over $\mathcal{M}_{opt}(z)$.
The following result is stated in terms of $\{\mu_{n}^{B}(Z^{n})\}$
defined in this way. 

\begin{thm}\label{thm:optimality} Under Assumptions \ref{assu:cms25:assu2}--\ref{assu:CMS25:lem:8},
$\{\mu_{n}^{B}\}\in\mathcal{D}$ is average optimal. \end{thm}

\begin{proof} See Appendix \ref{sec:Proof-of-Theorem}. \end{proof}

Our proof strategy follows the approaches of \citet{hirano2009asymptotics},
\citet{christensen2023optimal}, and \citet{xu2024jmp}, but requires
suitable extensions since our action space $(\mathcal{M},d_{W})$
is more complicated than theirs. This gives rise to technical challenges
specific to our framework, which we address by drawing on tools from
optimal transport.

\begin{rmk} Formally, the construction of $\{\mu_{n}^{B}\}$ depends
on the choice of a prior distribution. However, any prior satisfying
Assumption \ref{assu:prior} leads to the same conclusion.\end{rmk}

\begin{rmk} $\{\mu_{n}^{B}(Z^{n})\}$ remains average optimal if
the directional differentiability in Assumption \ref{assu:w-dot}
is strengthened to the full differentiability. \end{rmk}

\begin{rmk}\label{rmk: no_first-order_ties} \citet{christensen2023optimal}
impose a condition called \emph{no first-order ties}, which requires
the uniqueness of the minimizer of the loss function in the limit
experiment. This condition addresses an indeterminacy: their treatment
rule must be matched to one of the minimizers, but it is not uniquely
determined without this condition. By contrast, our construction of
the rule $\{\mu_{n}^{B}(Z^{n})\}$ allows us to avoid imposing this
condition, since $\mu_{n}^{B}(z)$ is matched with $\mu_{\infty}^{*}$
in the limit experiment where $\mu_{\infty}^{*}$ uniquely minimizes
the penalty function $H$ over the set of maximizers of (\ref{eq:max_optimal})
defined below.\footnote{\citet{xu2024jmp} assumes uniqueness of the rule and therefore does
not encounter the issue of non-uniqueness we addressed here.} \end{rmk}

\subsection{\label{subsec:plugin-asym}Asymptotic behavior of the plug-in rules}

We will explore the asymptotic behavior of the plug-in rules to compare
with the Bayesian rule. 

\subsubsection{When $w$ is directionally differentiable}

As a part of the proof of Theorem \ref{thm:optimality}, we show that
any average optimal rule $\{\mu_{n}(Z^{n})\}\in\mathcal{D}$ will
be matched by the rule in the limit experiment $\mu_{\infty}$ where
$\mu_{\infty}$ satisfies
\begin{equation}
\mu_{\infty}(\Delta)\in\arg\max_{\mu\in A_{0}}\int\int\dot{w}_{\theta_{0}}(x,t;s)\mathrm{d}N(\Delta,I_{0}^{-1})(s)\mathrm{d}\mu,\label{eq:max_optimal}
\end{equation}
with $\Delta\sim N(h,I_{0}^{-1})$ (see Lemma \ref{lem:lower_bound}).
The Bayesian rule satisfies this condition. We want to check whether
the plug-in rule $\mu_{n}^{P}(Z^{n})$ satisfies this. 

One can show that $\{\mu_{n}^{P}\}\in\mathcal{D}$, which implies
$\mu_{n}^{P}(Z^{n})\in A_{0}$ with probability approaching to one
along $P_{\theta_{nh}}^{n}$. Thus, for sufficiently large $n$, $\mu_{n}^{P}(Z^{n})$
equivalently solves
\begin{align*}
\arg\max_{\mu\in A_{0}}\sqrt{n}W(\hat{\theta}_{n},\mu) & =\arg\max_{\mu\in A_{0}}\int\sqrt{n}w(\hat{\theta}_{n},x,t)\mathrm{d}\mu\\
 & =\arg\max_{\mu\in A_{0}}\sqrt{n}\left[\int w(\hat{\theta}_{n},x,t)\mathrm{d}\mu-\int w(\theta_{0},x,t)\mathrm{d}\mu\right]
\end{align*}
where the second equality follows because the value of $W(\theta_{0},\mu)$
is constant across $\mu\in A_{0}$. We will see that the maximization
problem that plug-in rules solve weakly converges to a different maximization
problem from the one that the matched rules of optimal rules must
solve; i.e., (\ref{eq:max_optimal}). We actually claim that 
\begin{equation}
\max_{\mu\in A_{0}}\sqrt{n}\left[\int w(\hat{\theta}_{n},x,t)\mathrm{d}\mu-\int w(\theta_{0},x,t)\mathrm{d}\mu\right]\stackrel{h}{\rightsquigarrow}\max_{\mu\in A_{0}}\int\dot{w}_{\theta_{0}}(x,t;\Delta)\mathrm{d}\mu\quad\text{as }n\to\infty,\label{eq:max_plugin}
\end{equation}
where $\Delta\sim N(h,I_{0}^{-1})$. 

To see this, let $B_{n}(\mu):=\sqrt{n}\left[\int w(\hat{\theta}_{n},x,t)\mathrm{d}\mu-\int w(\theta_{0},x,t)\mathrm{d}\mu\right]$
and $B_{\infty}(\mu):=\int\dot{w}_{\theta_{0}}(x,t;\Delta)\mathrm{d}\mu$.
We impose a high-level condition that the process $\left\{ B_{n}(\mu):\mu\in A_{0}\right\} $
is asymptotically tight. Also note that $\sqrt{n}(\hat{\theta}-\theta_{0})=I_{0}^{-1}S_{n}+o_{P_{\theta_{0}}^{n}}(1)$
with $S_{n}\stackrel{0}{\rightsquigarrow}N(0,I_{0})$ as $n\to\infty$
by the best regularity of $\hat{\theta}_{n}$. Combining Le Cam's
third lemma and the delta method for the directionally differentiable
functions \citep[Theorem 2.1]{fang2019inference} yields
\[
\sqrt{n}\left[\int w(\hat{\theta}_{n},x,t)\mathrm{d}\mu-\int w(\theta_{0},x,t)\mathrm{d}\mu\right]\stackrel{h}{\rightsquigarrow}\int\dot{w}_{\theta_{0}}(x,t;\Delta)\mathrm{d}\mu\quad\text{as }n\to\infty,
\]
where $\Delta\sim N(h,I_{0}^{-1})$. By the asymptotic tightness of\emph{
}$\left\{ B_{n}(\mu):\mu\in A_{0}\right\} $, we can extend this result
to convergence in distribution of the process 
\[
B_{n}\stackrel{h}{\rightsquigarrow}B_{\infty}\quad\text{as }n\to\infty\text{ on }\ell^{\infty}(A_{0}),
\]
by \citet[Theorem 1.5.4]{vdvW1996}. Then applying the continuous
mapping theorem yields
\begin{align*}
\max_{\mu\in A_{0}}B_{n}(\mu) & \stackrel{h}{\rightsquigarrow}\max_{\mu\in A_{0}}B_{\infty}(\mu)\quad\text{as }n\to\infty,
\end{align*}
which completes the argument.

It is evident that the solutions of RHS of (\ref{eq:max_plugin})
need not to solve (\ref{eq:max_optimal}), the maximization problem
in the limit experiment that the matched rules of optimal rules must
solve. Thus the plug-in rules might not be average optimal in general
when $w$ is directionally differentiable. 

\begin{rmk} To investigate the asymptotic properties of the plug-in
rule formally, we need to handle the non-uniqueness issue. This can
be done by the penalization used for the construction of the Bayesian
rule. \end{rmk}

\subsubsection{When $w$ is fully differentiable}

If we strengthen the directional differentiability in Assumption \ref{assu:w-dot}
to the full differentiability, then the plug-in rules become average
optimal. To see this, notice that $\dot{w}_{\theta_{0}}(x,t;s)=\dot{w}_{\theta_{0}}(x,t)^{\top}s$
for some $\dot{w}_{\theta_{0}}(x,t)\in\mathbb{R}^{k}$ from the linearity
of the directional derivative. Then the maximization problem (\ref{eq:max_optimal})
is rewritten as 
\[
\max_{\mu\in A_{0}}\int\dot{w}_{\theta_{0}}(x,t)^{\top}\Delta\mathrm{d}\mu,
\]
which is the same as (\ref{eq:max_plugin}), the maximization problem
that the matched rules of the plug-in rules solve in the limit experiment. 

This pattern is consistent with findings from the existing literature.
As discussed earlier, \citet{hirano2009asymptotics} show that the
plug-in rule is average optimal in point-identified models when the
utility function is fully differentiable. In contrast, \citet{christensen2023optimal}
demonstrate that the Bayesian rule is average optimal in partially
identified models when the utility function is only directionally
differentiable, and the plug-in rule fails to be optimal unless the
full differentiability holds. In partially identified settings, directional
differentiability is a natural and often unavoidable assumption, as
full differentiability typically does not hold.

\section{Simulation \label{sec:Simulation}}

We conduct a simulation study to evaluate the performance of the Bayesian
rule and the plug-in rule under the following conditions: (i) the
welfare function is either smooth or only directionally differentiable,
and (ii) the sample size is relatively small ($n=200$) and large
($n=500$).

We closely follows the data generating process described in Example
\ref{eg:ABK}. For the training population, the latent variable is
generated by 
\[
Y_{i}^{*}=X_{i}^{\top}\beta+\alpha T_{i}+u_{i},
\]
where $X_{i}\in\mathbb{R}^{2}$ denotes the observable covariates
and $T_{i}$ is the binary treatment that is randomly assigned. The
first coordinate of $X_{i}$, interpreted as age, follows a truncated
normal distribution with mean 4, standard deviation of 2, and is bounded
on $[1,10]$. The second coordinate, interpreted as sex, is a binary
variable assigned with equal probability. The observed outcome is
\[
Y_{i}=\max\{0,Y_{i}^{*}\}.
\]
We set $\beta_{0}=(-2,-3)$, $\alpha_{0}=4$, and $u_{i}\sim N(0,\sigma_{0}^{2})$
with $\sigma_{0}=10$. The observed data is an i.i.d. sample $Z^{n}=\{(Y_{i},X_{i},T_{i})\}_{i=1}^{n}$.
The parameters $\theta_{0}=(\beta_{0},\alpha_{0},\sigma_{0})$ can
be estimated by the maximum likelihood using $Z^{n}$. 

In this Tobit model, the conditional mean of the potential outcomes
in the training population is given by 
\[
w(\theta,x,t)=(x^{\top}\beta+\alpha t)-(x^{\top}\beta+\alpha t)\Phi\left(\frac{-x^{\top}\beta-\alpha t}{\sigma}\right)+\sigma\phi\left(\frac{-x^{\top}\beta-\alpha t}{\sigma}\right),
\]
where $\Phi$ and $\phi$ denote the standard normal cdf and pdf,
respectively. Following Examples \ref{eg:bayesian} and \ref{eg:maxmin},
we specify the planner's utility function as 
\[
w_{R}(\theta,x,t,\varepsilon,\lambda)=\lambda w(\theta,x,t)+(1-\lambda)\max\left\{ w(\theta,x,t)-\varepsilon,0\right\} ,
\]
where $\lambda\in[0,1]$ and $\varepsilon>0$. We note that $w_{R}$
is differentiable when $\lambda=1$, but only directionally differentiable
otherwise. In what follows, we focus on the cases $\lambda=0,1$ and
$\varepsilon=0.8$. 

Figure \ref{fig:contrast} plots the welfare contrasts $w_{R}(\theta,x,1,\varepsilon,\lambda)-w_{R}(\theta,x,0,\varepsilon,\lambda)$
for both males and females at $\theta_{0}$.
\begin{figure}
\caption{Welfare contrasts under smooth and directionally differentiable welfare
at $\theta_{0}$}
\label{fig:contrast}

\includegraphics[scale=0.35]{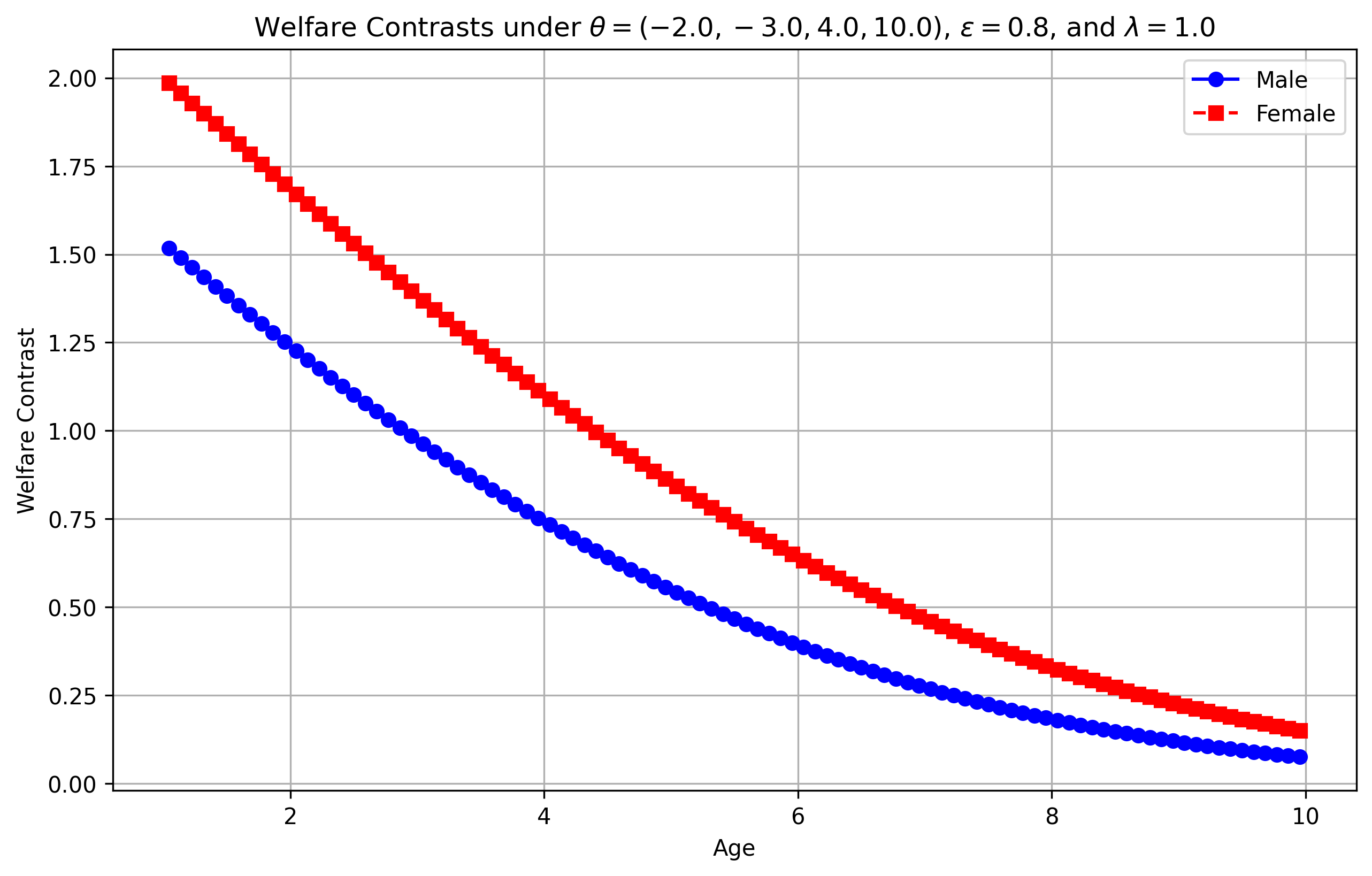}\includegraphics[scale=0.35]{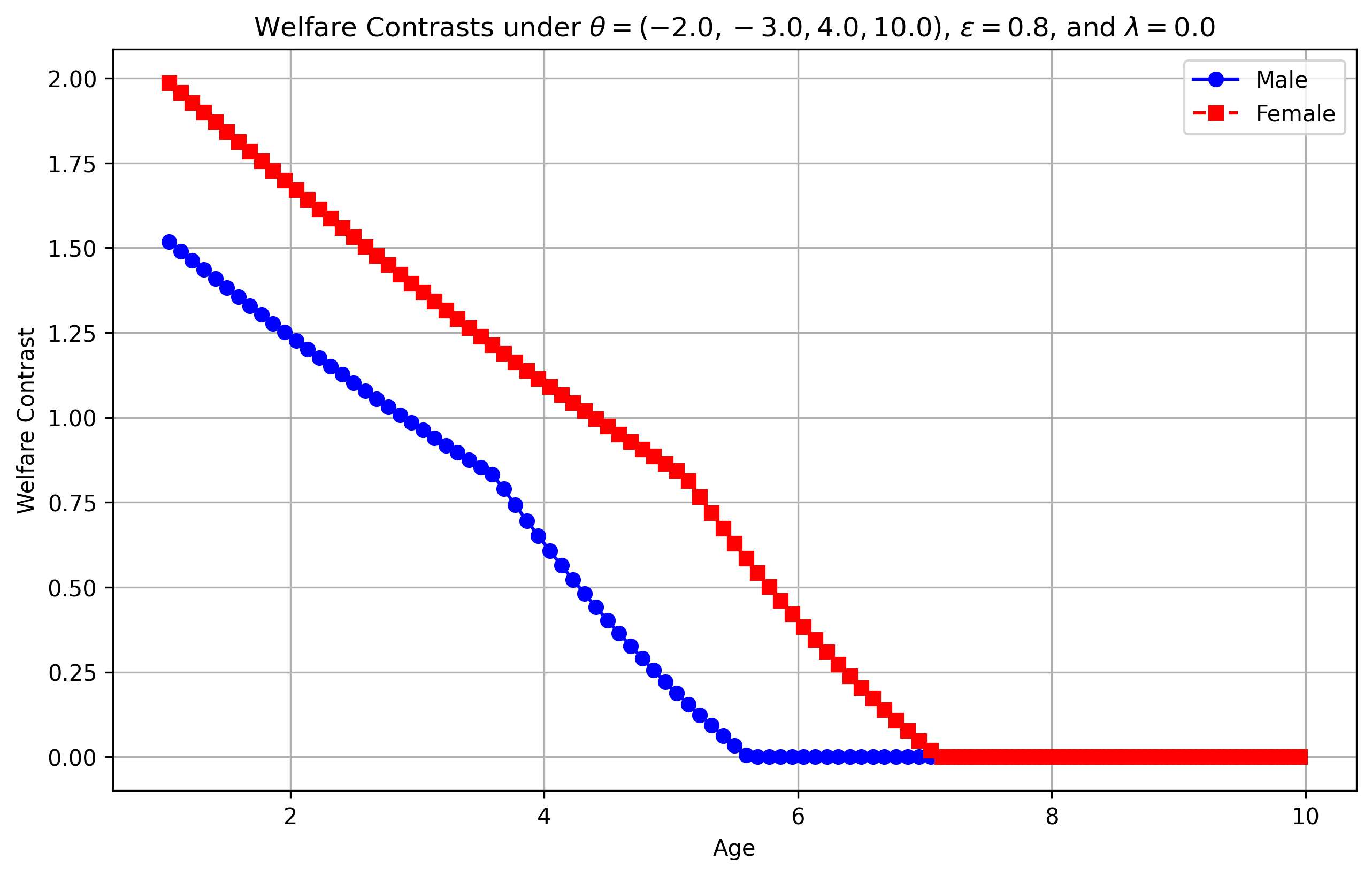}
\end{figure}
 Under $\lambda=0$, kinks appear when $w(\theta,x,t)-\varepsilon<0$.
For a given covariate $x$, the contrast is zero if both $w(\theta,x,1)<\varepsilon$
and $w(\theta,x,0)<\varepsilon$; that is, individuals with sufficiently
low welfare are regarded as deriving no benefit from treatment. 

The upper panels of Figure \ref{fig:assignment_n=00003D200} show
the oracle (infeasible optimal) rule at $\theta_{0}$. The rule assigns
treatment to females younger than approximately 6.5 and to males younger
than approximately 5. Notably, this oracle rule remains unchanged
across $\lambda=0$ and $\lambda=1$. It also remains the same at
$\theta_{0}+h/\sqrt{n}$, for the range of local deviation parameters
$h$ specified below.

We assume that the true distribution of covariates $X$ in the training
population is known and that the target population shares the same
distribution. Specifically, we define $F_{X}$ as the joint distribution
of (a) the truncated normal distribution for the age variable and
(b) the binary distribution for the sex variable. For computational
purposes, we discretize $F_{X}$ into 99 bins, each corresponding
to a distinct combination of age and sex. Each bin is assigned a probability
mass according to $F_{X}$, representing the proportion of individuals
falling into that bin.

Suppose that the planner has resources to allocate to 75\% of the
target population. Let
\[
W(\theta,\mu)=\int w_{R}(\theta,x,t,\varepsilon,\lambda)\mathrm{d}\mu(x,t).
\]

\subsection{Average optimality}

We evaluate performance under a sequence of perturbed DGPs:
\[
\theta_{nh}=\theta_{0}+h/\sqrt{n},\quad\text{for }h\in H:=\{-2,-1.6,\dots,2\},
\]
where $h/\sqrt{n}$ is added to $\theta_{0}$ element-wisely. Our
goal is to compare the average risk 
\[
\int R(\mu_{n}^{Q},\theta_{nh})\mathrm{d}h=\int\mathbb{E}_{P_{\theta_{nh}}^{n}}\left[W_{\mathcal{M}}^{*}(\theta_{nh})-W(\theta_{nh},\mu_{n}(Z^{n}))\right]\mathrm{d}h
\]
for $Q=P,B$. The simulation proceeds as follows:
\begin{enumerate}
\item For each $h$, draw $J$ independent samples of data $\{Z^{n,j}\}_{j=1}^{J}$
from $P_{\theta_{nh}}^{n}$ where $Z^{n,j}=\{Z_{i}^{j}\}_{i=1}^{n}$. 
\item For each $j$: 
\begin{enumerate}
\item Obtain the MLE estimates of parameter $\hat{\theta}_{nh}^{j}$ and
Fisher information matrix $\hat{I}_{nh}^{j}$.
\item Compute the plug-in rule $\mu_{nh}^{P,j}$ by
\[
\mu_{nh}^{P,j}\in\arg\max_{\mu\in\mathcal{M}}W(\hat{\theta}_{nh}^{j},\mu).
\]
 
\item Draw $L$ samples $\{\theta_{\ell}\}_{\ell=1}^{L}$ from $N(\hat{\theta}_{nh}^{j},(n\hat{I}_{nh}^{j})^{-1})$,
which can be interpreted as the quasi-posterior using the quasi-likelihood
$N(\hat{\theta}_{nh}^{j},(n\hat{I}_{nh}^{j})^{-1})$ with the uniform
prior \citep{kim2002limited,christensen2023optimal}. Then compute
the Bayesian rule $\mu_{nh}^{B,j}$ by
\[
\mu_{nh}^{B,j}\in\arg\max_{\mu\in\mathcal{M}}\frac{1}{L}\sum_{\ell=1}^{L}W(\theta_{\ell},\mu).
\]
\end{enumerate}
\item Compute the oracle welfare $W_{\mathcal{M}}^{*}(\theta_{nh})=\max_{\mu\in\mathcal{M}}W(\theta_{nh},\mu)$,
and estimate $R(\mu_{n}^{Q}(Z^{n}),\theta_{nh})$ by 
\[
R(Q,h):=\frac{1}{J}\sum_{j=1}^{J}\left[W_{\mathcal{M}}^{*}(\theta_{nh})-W(\theta_{nh},\mu_{nh}^{Q,j})\right]
\]
for $Q=P,B$. Store $R(Q,h)$ for each $h$. 
\item Taking the average of $R(Q,h)$ over $h$ gives an estimate of $\int R(\mu_{n}^{Q},\theta_{nh})\mathrm{d}h$
for $Q=P,B$. 
\end{enumerate}
We use \textsf{POT}, an open-source Python library developed by \citet{JMLR:v22:20-451},
to compute the plug-in rule and the Bayesian rule. 

\subsection{Results}

We study the cases of $n=200,500$, $J=2000$, and $L=2000$ for both
$\lambda=0$ and $\lambda=1$. We first report the estimated risks,
followed by comparisons of the resulting treatment allocations.

Figure \ref{fig:risk_n=00003D200} shows the results for $n=200$.
While our theory predicts the plug-in and Bayesian rules are asymptotically
optimal under smooth welfare ($\lambda=1$), the simulation shows
that the Bayesian rule performs better in small samples. We also observe
that the Bayesian rule outperforms the plug-in rule under directionally
differentiable welfare ($\lambda=0$).
\begin{figure}
\caption{Comparisons of estimated risks: $n=200$ (left: smooth welfare, right:
directionally differentiable welfare)}
\label{fig:risk_n=00003D200}

\includegraphics[scale=0.45]{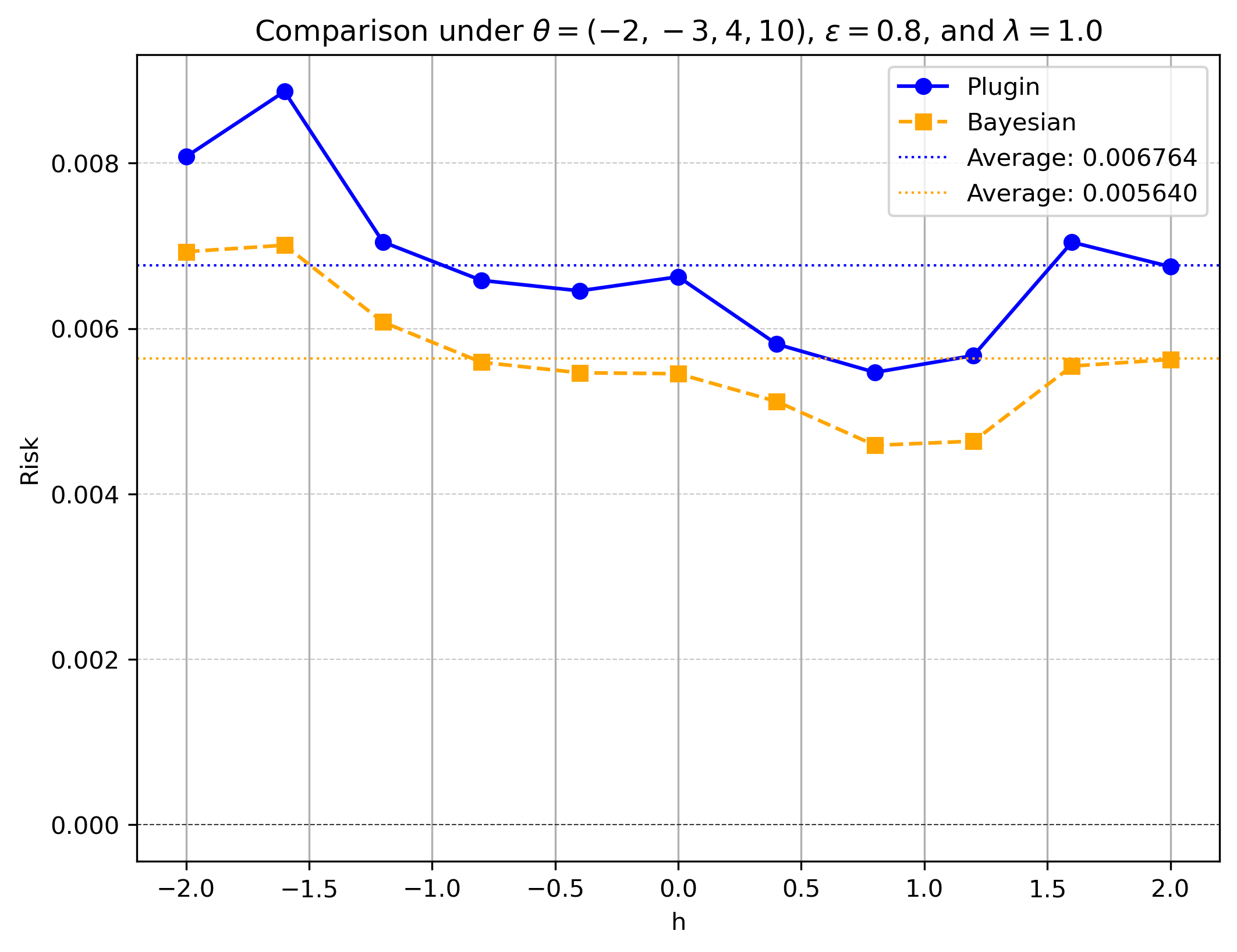}\includegraphics[scale=0.45]{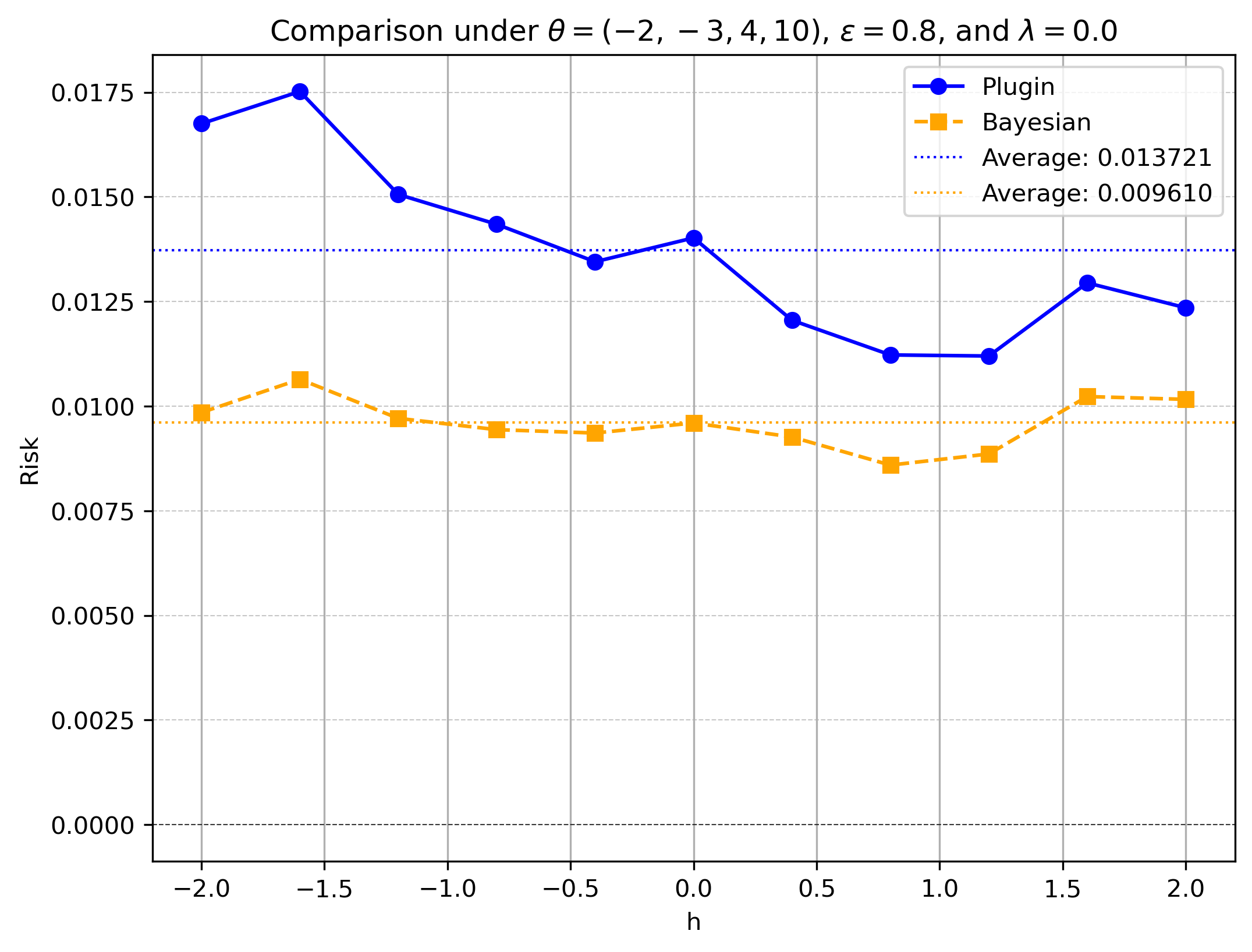}
\end{figure}

Figure \ref{fig:risk_n=00003D500} shows the results for $n=500$.
The Bayesian rule still performs slightly better when $\lambda=1$,
but the overall risk levels are substantially reduced, and the performance
gap between the two rules narrows. This indicates that both rules
are approaching optimality as the sample size increases from 200 to
500. When $\lambda=0$, the Bayesian rule continues to outperform
the plug-in rule, which is consistent with our theoretical predictions:
under the directionally differentiable welfare the Bayesian rule is
optimal, but the plug-in rule may not be. Notably, the Bayesian rule
performs particularly well when the values of $h$ are negative. In
these cases, the welfare contrasts become smaller, making the assignment
problem more challenging. This highlights the robustness of the Bayesian
rule to local perturbations that make treatment decisions harder.
\begin{figure}
\caption{Comparisons of estimated risks: $n=500$ (left: smooth welfare, right:
directionally differentiable welfare)}
\label{fig:risk_n=00003D500}

\includegraphics[scale=0.45]{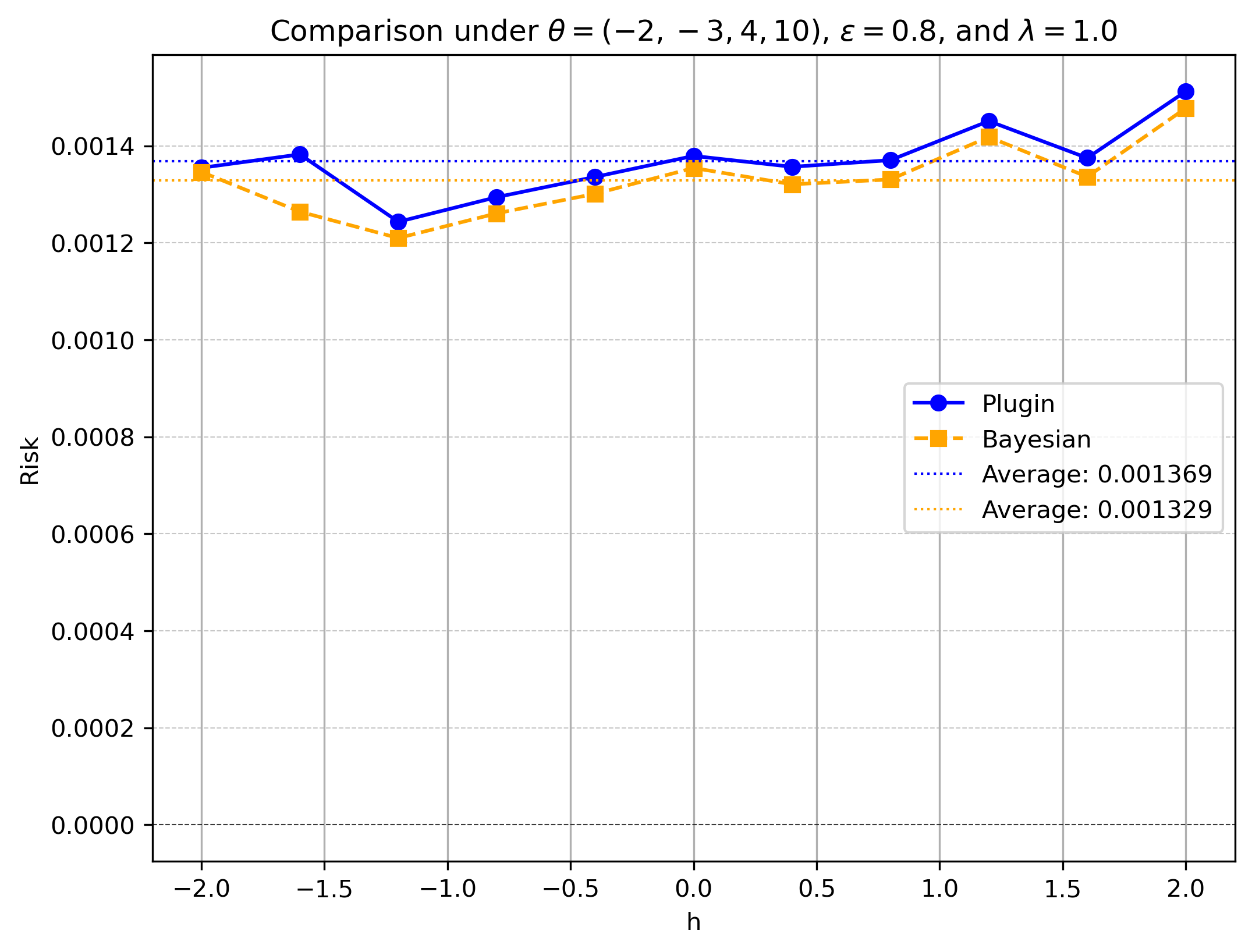}\includegraphics[scale=0.45]{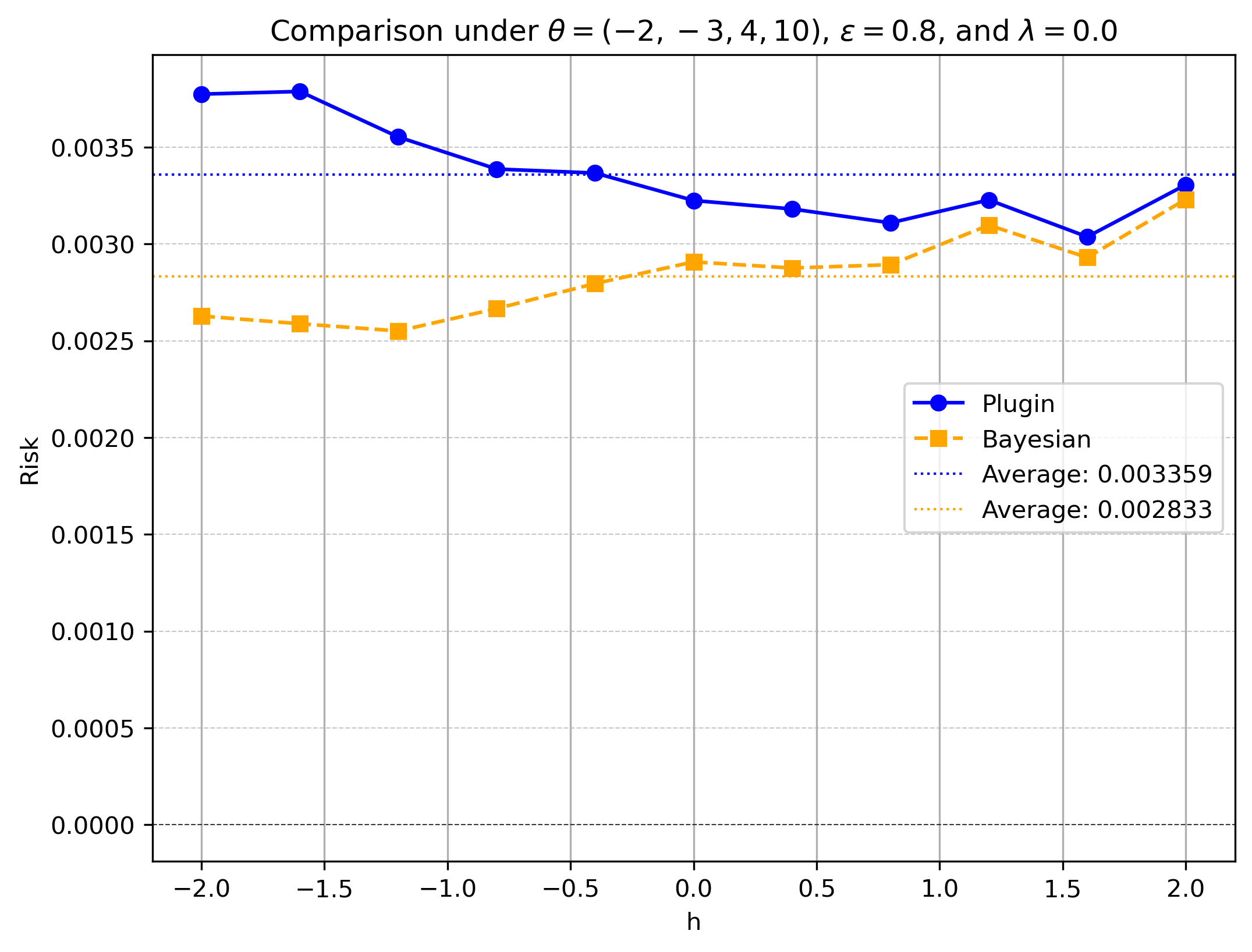}
\end{figure}

To gain further insight into the behavior of the two rules, we visualize
the average allocations, $J^{-1}\sum_{j}\mu_{nh}^{Q,j}$, for $Q=P,E$,
under $\theta_{0}$, $n=200$, and $\lambda=0$. Figure \ref{fig:assignment_n=00003D200}
shows that the Bayesian rule deviates from the oracle rule only near
the decision boundary, while the plug-in rule exhibits substantial
deviations even away from it. This relative stability of the Bayesian
rule contributes to a sizable risk reduction. A similar, albeit weaker,
pattern is observed for $n=500$. 
\begin{figure}
\caption{Comparison of (average) treatment assignment under directionally differentiable
welfare}
\medskip{}

{\small\textbf{Note}}{\small : The upper panels show the oracle rule,
the middle show the Bayesian rule, and the lower show the plug-in
rule under $\theta_{0}$ and $n=200$. }\label{fig:assignment_n=00003D200}

\includegraphics[scale=0.55]{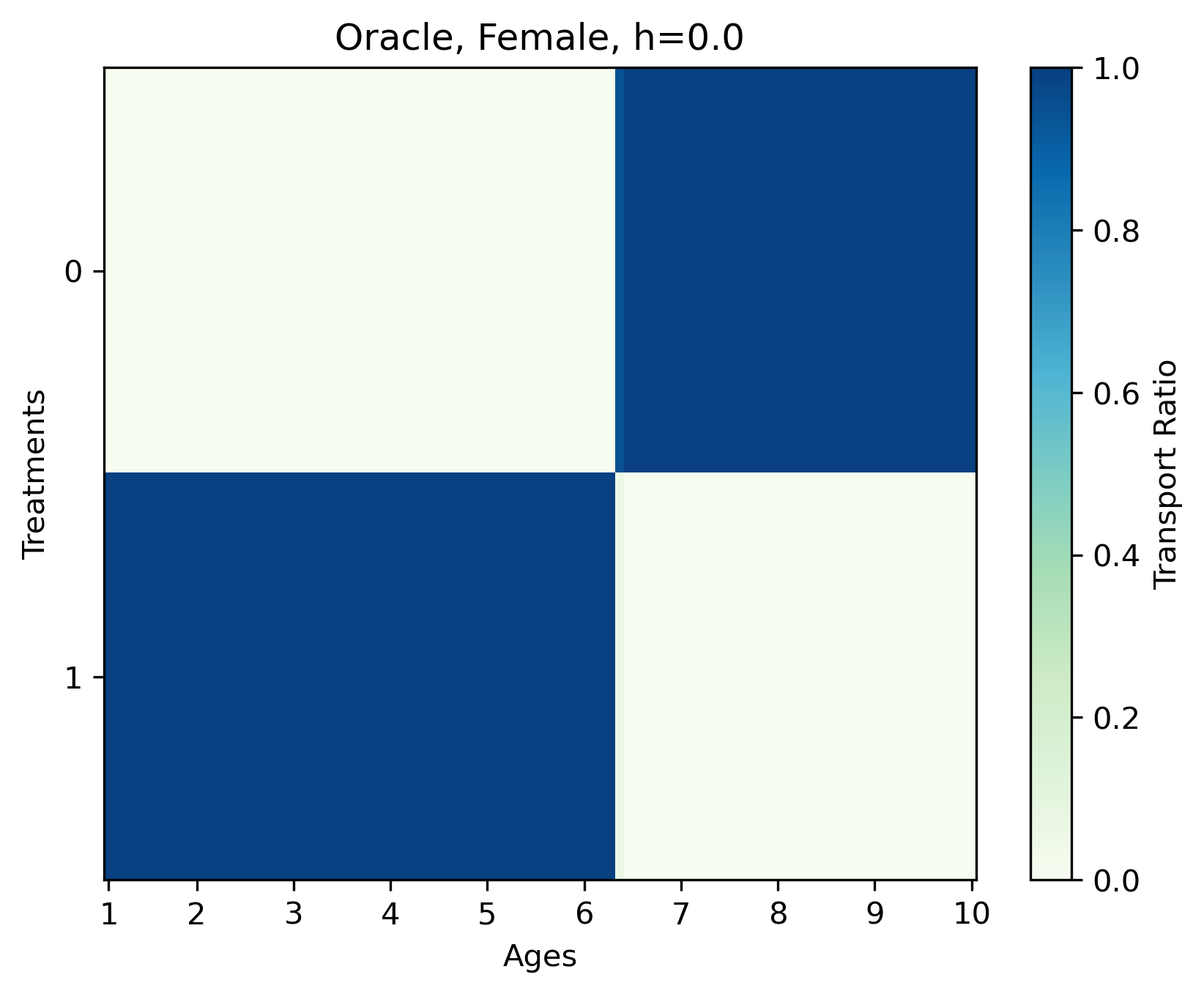}\includegraphics[scale=0.55]{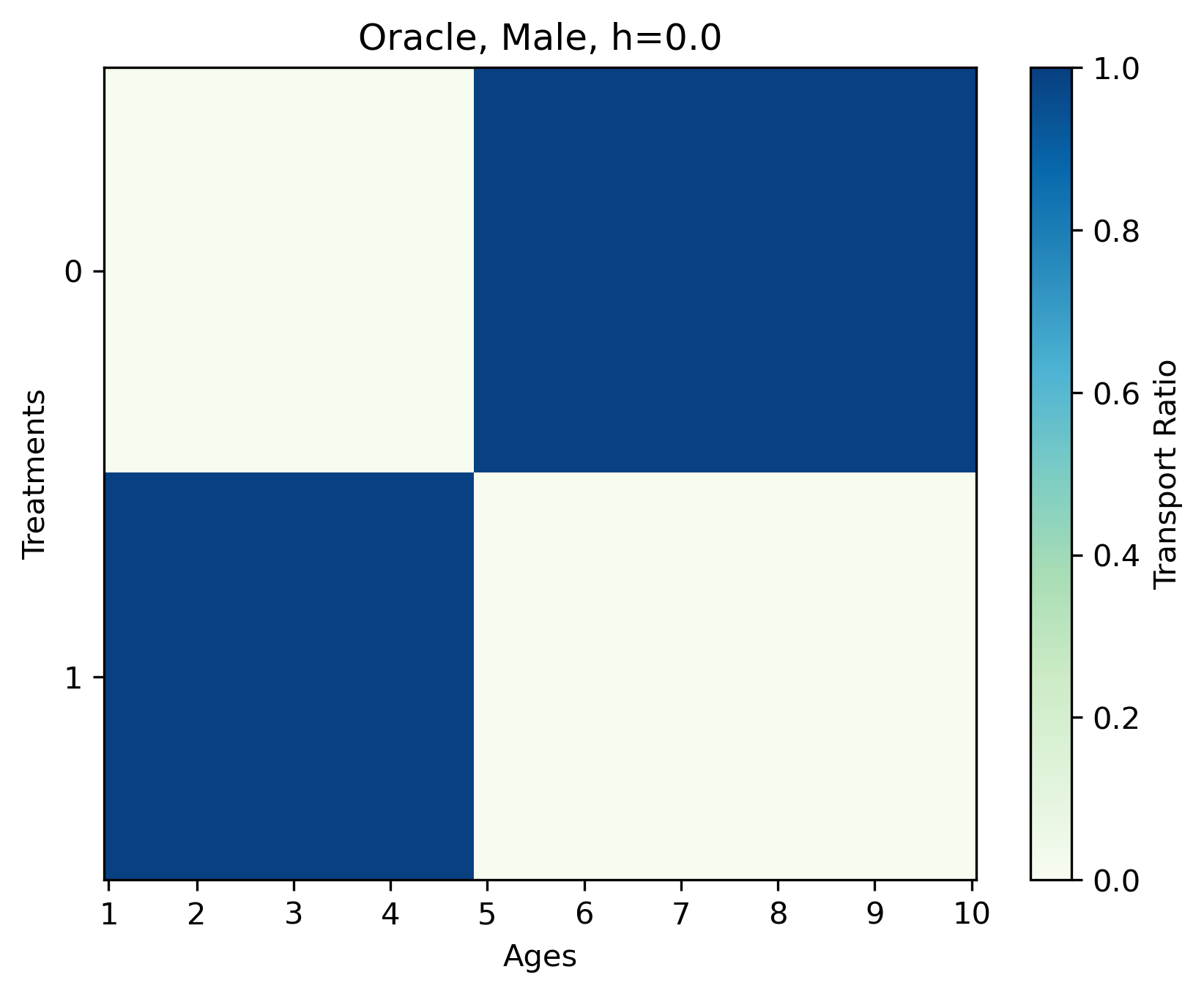}\linebreak{}
\includegraphics[scale=0.55]{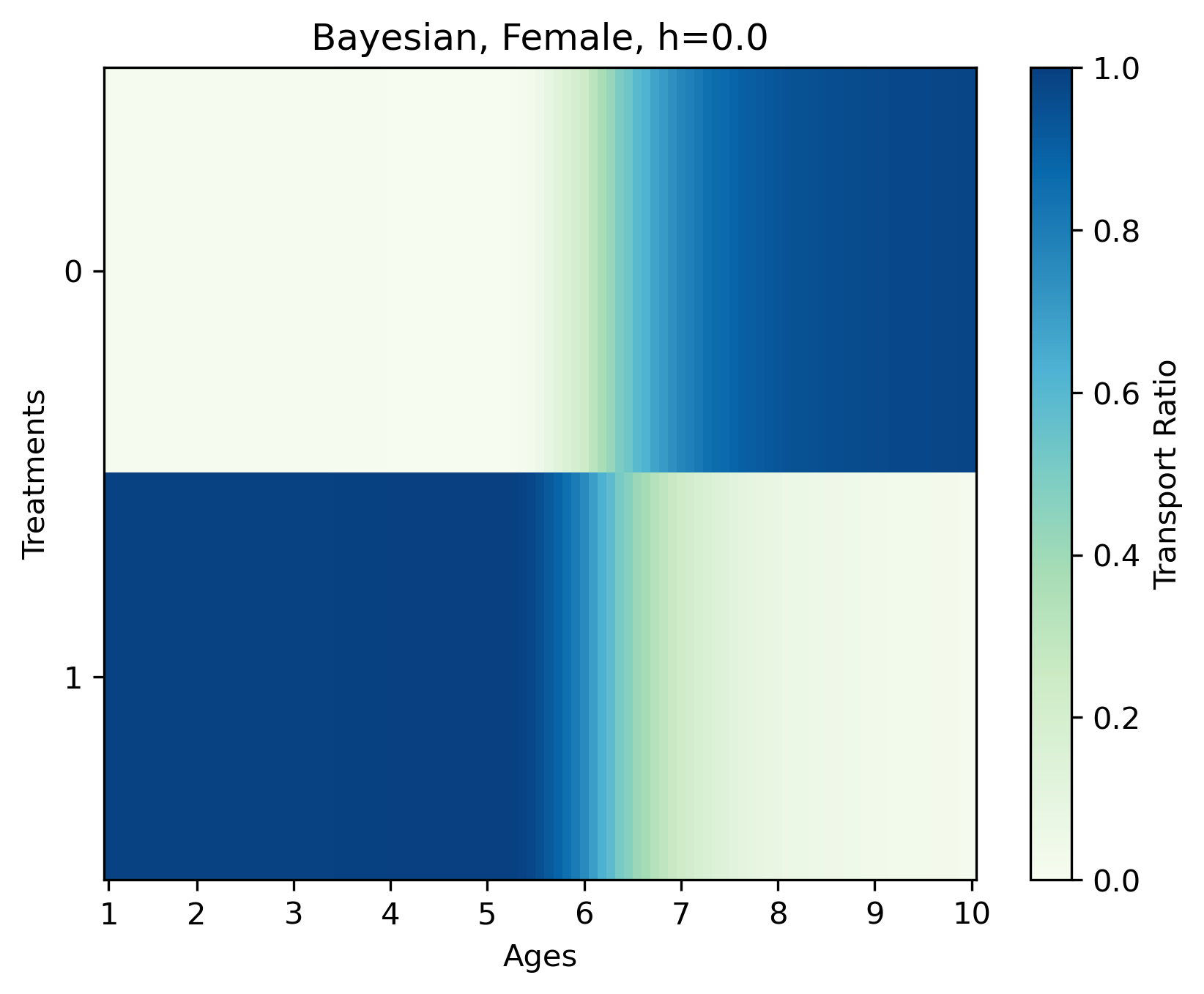}\includegraphics[scale=0.55]{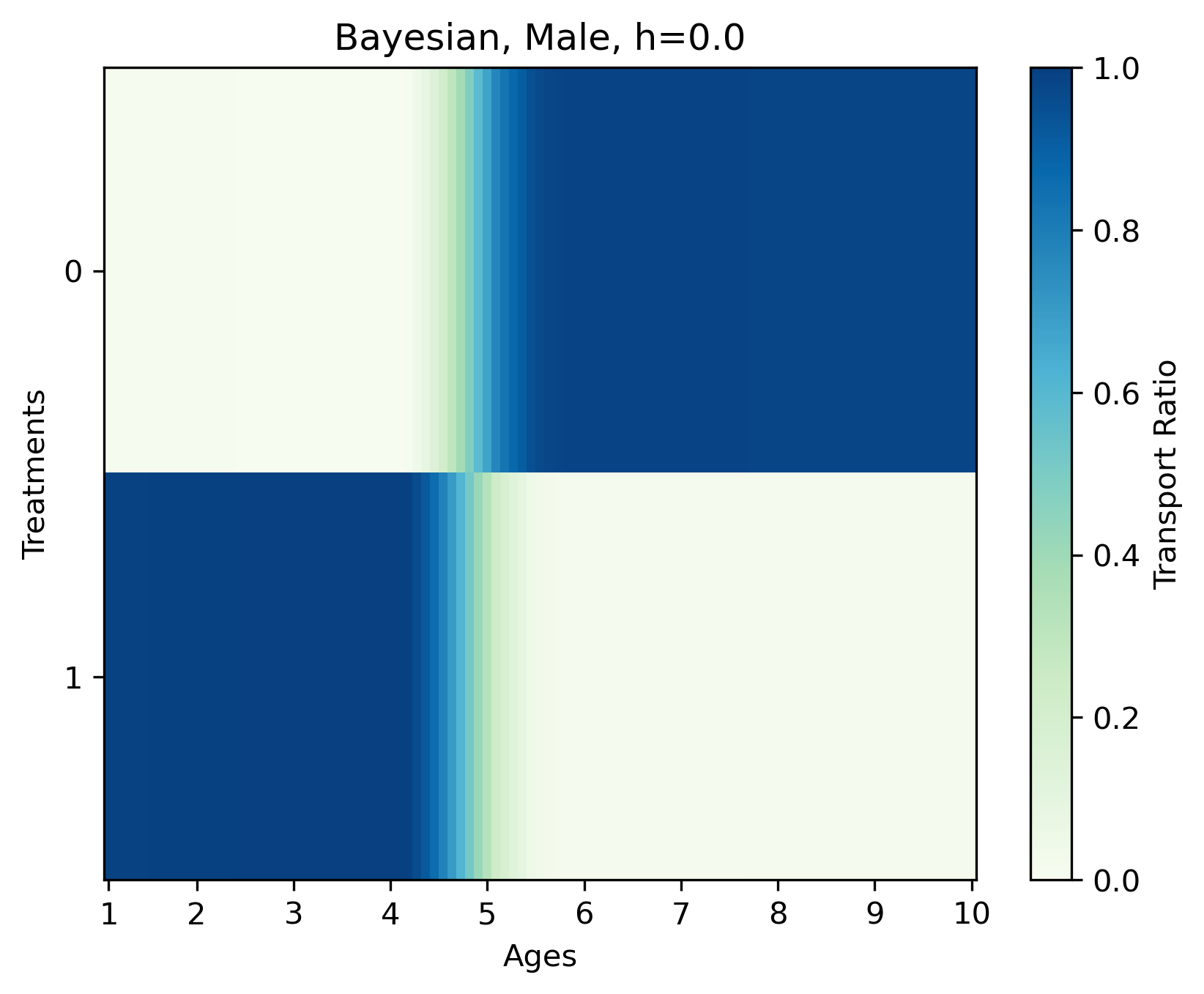}\linebreak{}
\includegraphics[scale=0.55]{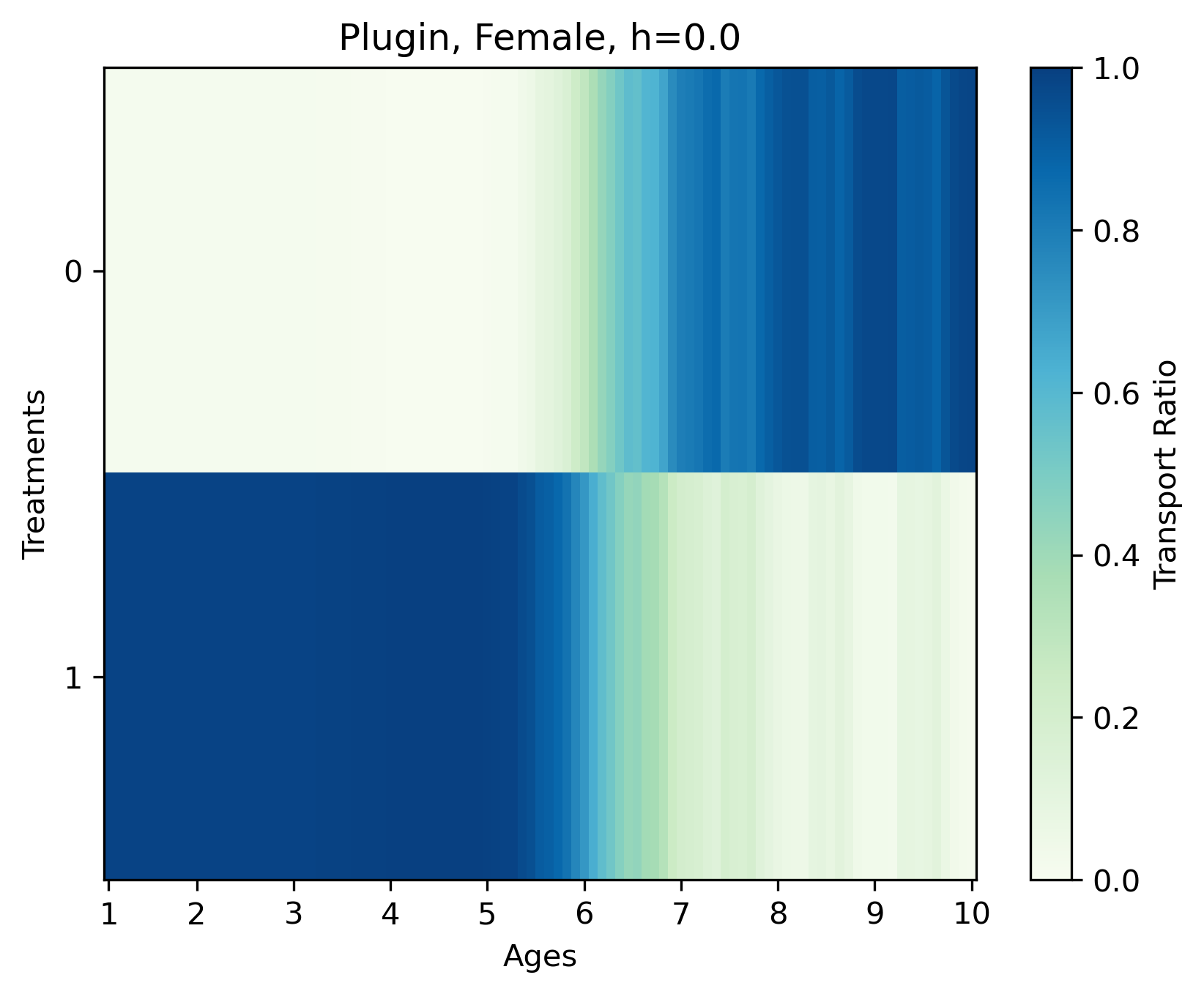}\includegraphics[scale=0.55]{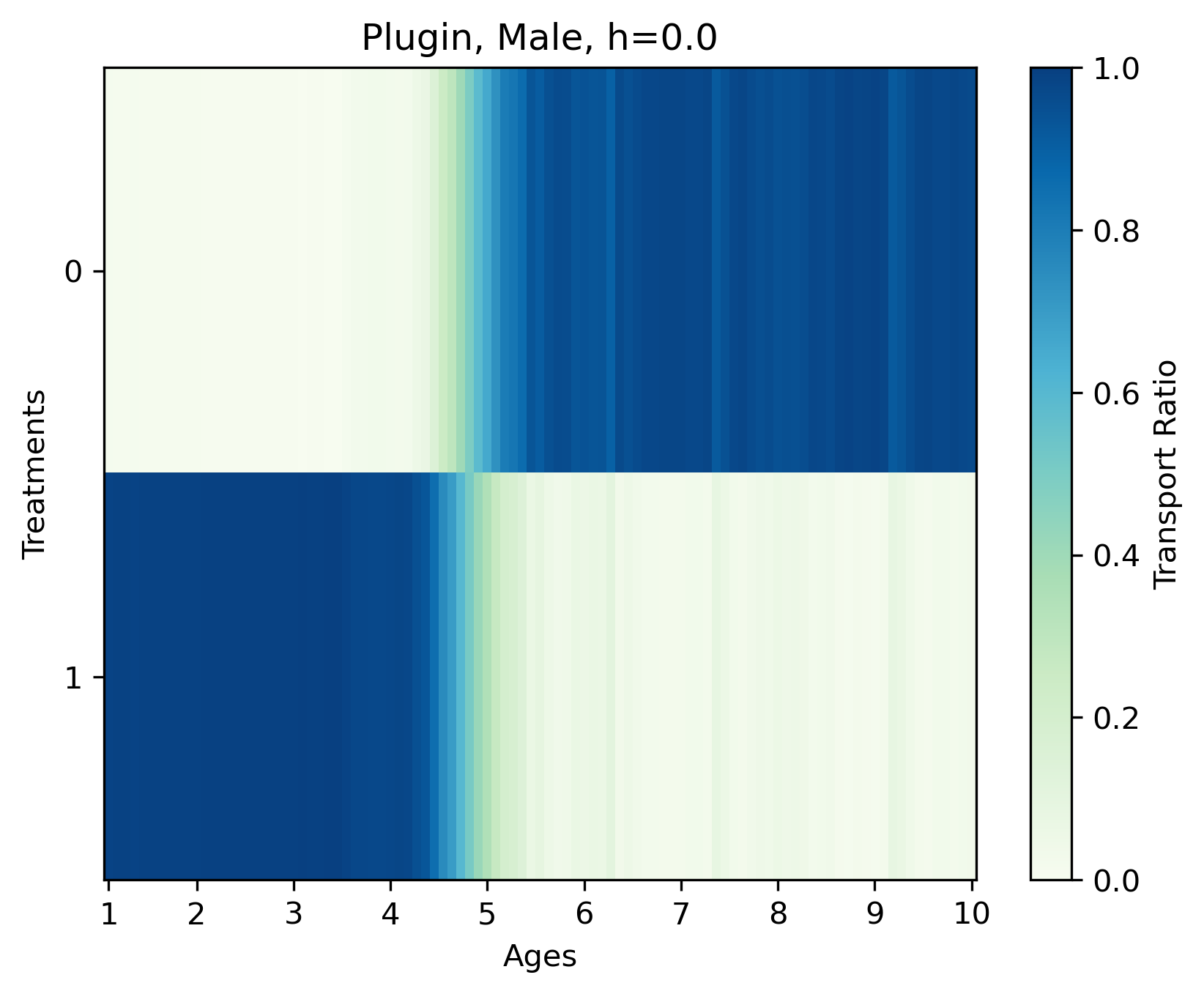}
\end{figure}

\begin{rmk} Under the current simulation setup, when the total amount
of resources is sufficiently small, the kink points of the directionally
differentiable welfare do not affect the assignment decision, as the
available resources are exhausted before the assignment rule encounters
the kink points. In such cases, the behavior of the two rules under
directionally differentiable welfare resembles their behavior under
the smooth welfare. We set the resource level to 75\% to allow interaction
between the decision boundary under the oracle rule and the kink points.
\end{rmk}

\section{Empirical application\label{sec:Empirical-application}}

We illustrate our methods using data from \citet{angrist2006long}
continued from Example \ref{eg:ABK}.\footnote{For the replication dataset of the original article, see \citet{angristReplicationDataLongterm2019Am.Econ.Assoc.}.}
As outcome variables, \citet{angrist2006long} use the test scores
in language and math. Since the estimation results are similar between
these two, we use math scores as the outcome variable for illustration.
We focus on the case where the observed test scores are censored at
the tenth percentiles of the test score distribution among test-takers
(denoted by $\tau$), in line with the original article, to address
selection issues. In addition to the test scores, we observe treatment
status, as well as age and sex as covariates. The sample includes
3,541 individuals overall, with 1,788 girls and 1,753 boys. Ages range
from 10 to 17 with mean 12.7 and standard deviation 1.3. The maximum
likelihood estimates are summarized in Table \ref{table:MLE}. 
\begin{table}
\caption{Maximum likelihood estimates for math scores}
\label{table:MLE}

\begin{tabular}{cccccc}
\hline 
Variables & voucher & age & gender & const & $\sigma^{2}$\tabularnewline
\hline 
Estimates & 2.06 & -5.5 & -0.72 & 102.77 & 104.31\tabularnewline
\hline 
Standard error & 0.46 & 0.24 & 0.44 & 2.87 & 5.17\tabularnewline
\hline 
\end{tabular}
\end{table}

We then hypothetically treat the marginal distribution of the covariates
in the observed sample as that of the target population, and compute
both the plug-in and the Bayesian rules as described in the previous
section. In this example, the planner's utility function is given
by: 
\begin{align*}
w_{R}(\theta,x,t,\varepsilon,\lambda) & =\lambda w(\theta,x,t)+(1-\lambda)\max\left\{ w(\theta,x,t)-\varepsilon,\tau\right\} ,
\end{align*}
where 
\[
w(\theta,x,t)=(x^{\top}\beta+\alpha t)+(\tau-x^{\top}\beta-\alpha t)\Phi\left(\frac{\tau-x^{\top}\beta-\alpha t}{\sigma}\right)+\sigma\phi\left(\frac{\tau-x^{\top}\beta-\alpha t}{\sigma}\right).
\]
Note that this is slightly different from the utility function in
the previous section as the outcome variable is censored at $\tau\neq0$.
In what follows, we focus on $\varepsilon=3.5$ and $\lambda=0,1$.
We consider the case where we can assign vouchers for 50\% of the
target population.

Figure \ref{fig:bayesian} shows the allocations under smooth welfare
($\lambda=1$). As Table \ref{table:MLE} shows, age has a negative
effect on outcomes. Accordingly, the plug-in rule allocates vouchers
to younger individuals. Since the effect of sex is slightly negative,
the plug-in rule prioritizes females over males, resulting in the
allocation where vouchers are fully allocated to females aged 10-12,
while not fully allocated to males at age 12 as the resource is exhausted
due to the capacity constraints. In this setting, the Bayesian rule
yields exactly the same allocation, which is natural since both rules
are optimal under smooth welfare. This also aligns with the simulation
result in the previous section: both rules perform similarly when
the sample size is large enough. 
\begin{figure}
\caption{Voucher allocations under smooth welfare}
\medskip{}

{\small\textbf{Note}}{\small : The upper panels show the plug-in rule,
and the lower panels show the Bayesian rule. The color intensity represents
the density of each cell in $F_{X}$, with darker shades indicating
higher density.}{\small\par}

\label{fig:bayesian} 

\includegraphics[scale=0.55]{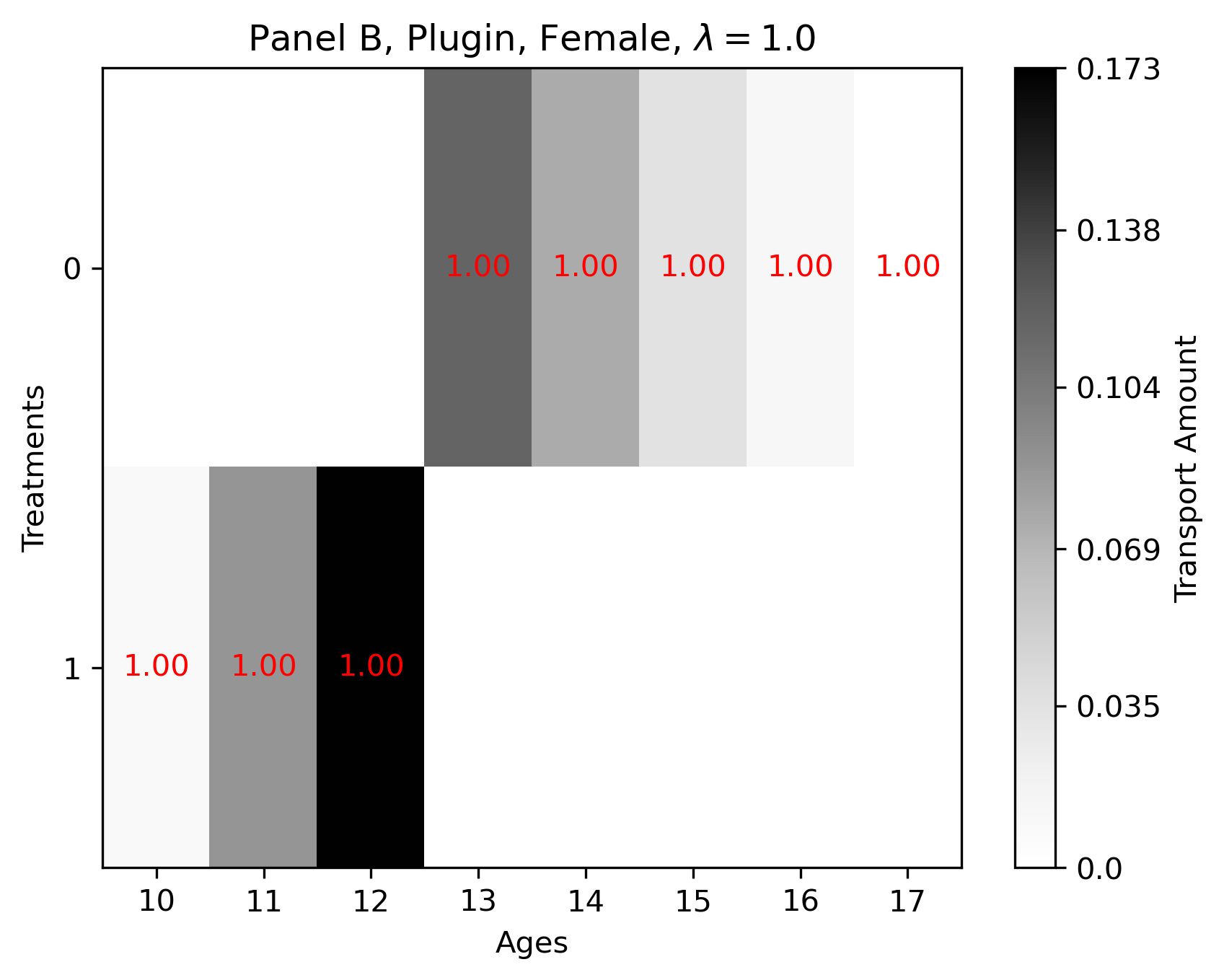}\includegraphics[scale=0.55]{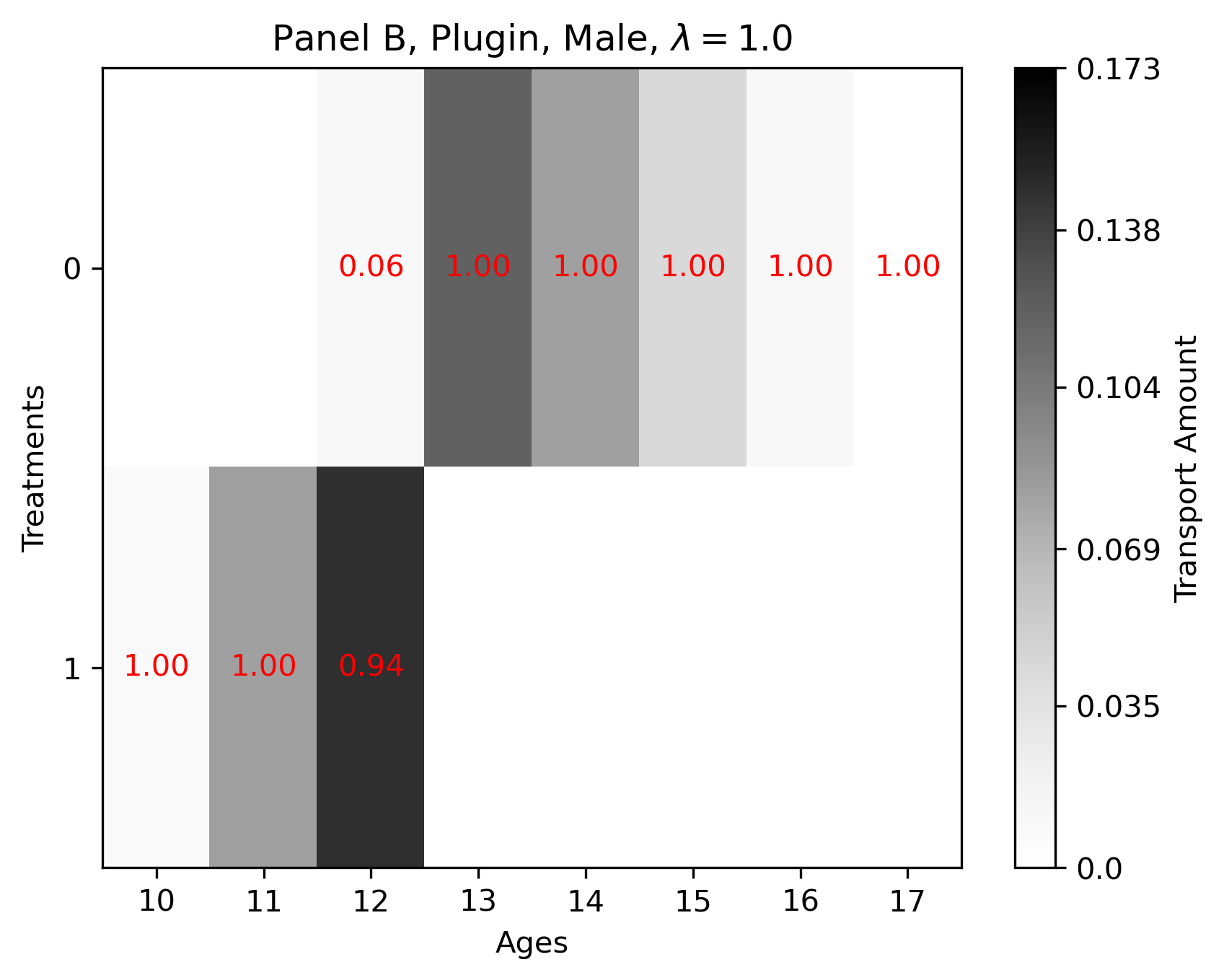}\linebreak{}

\includegraphics[scale=0.55]{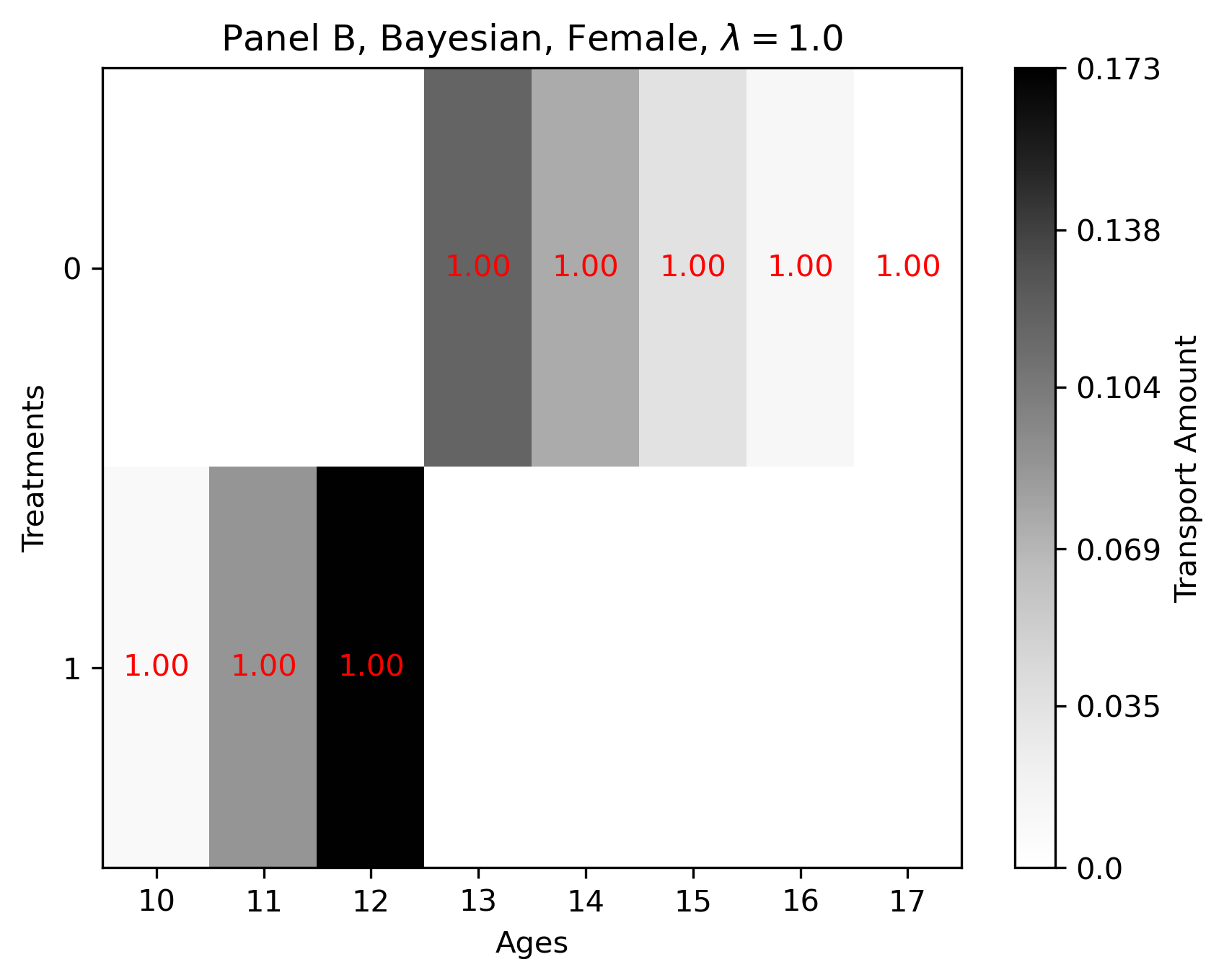}\includegraphics[scale=0.55]{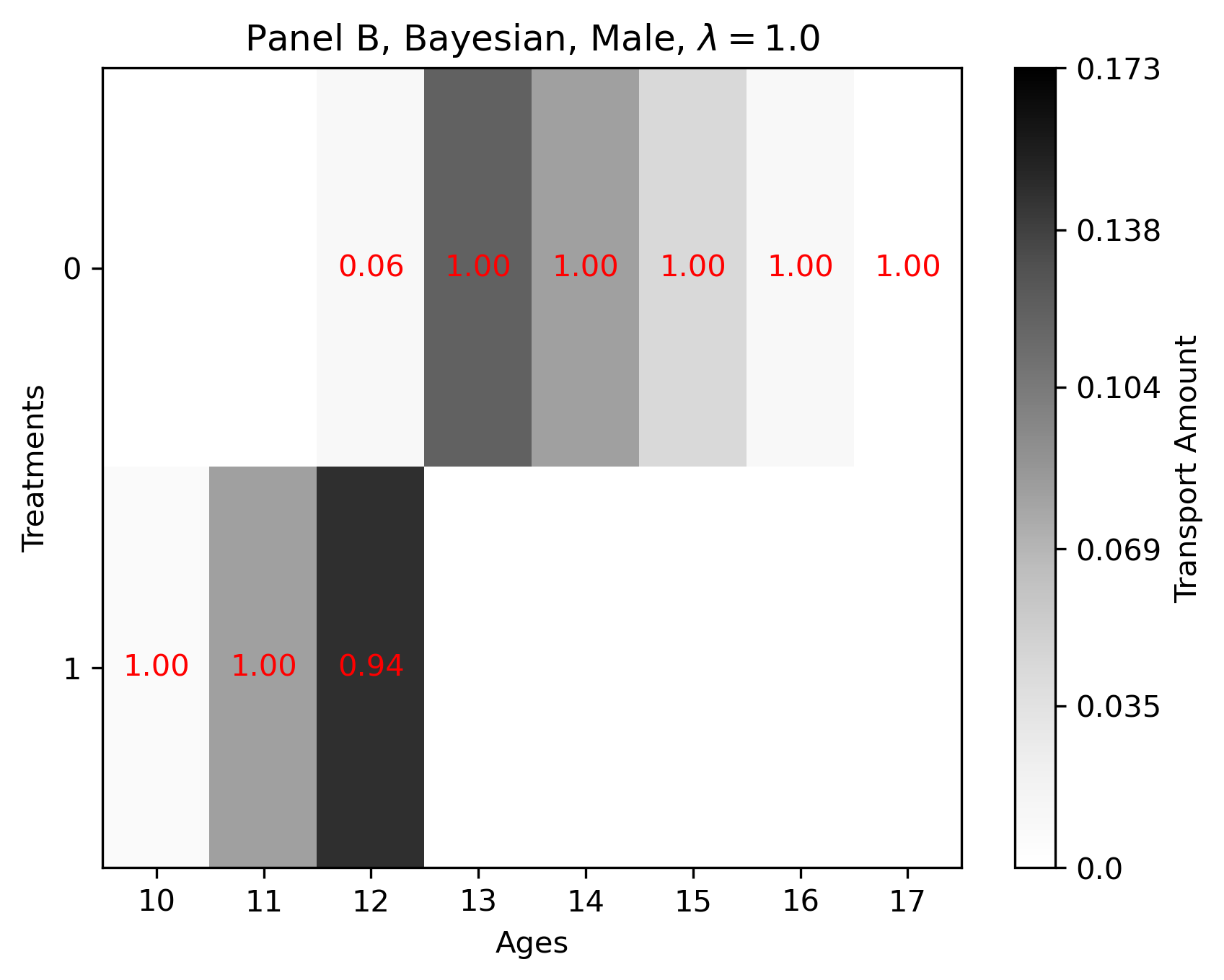}
\end{figure}

Next, Figure \ref{fig:ambig} shows the allocations under directionally
differentiable welfare ($\lambda=0$). For the plug-in rule, the value
of $w_{R}$ is censored by $\tau$ at age 13 for females and at 12
for males in this setting. As in the previous case, vouchers are fully
allocated to females aged 10--12 and males aged 10--11. However,
the remaining vouchers are randomly assigned, as the value of $w_{R}$
is just equal to $\tau$ for the rest. For the Bayesian rule, after
integrating with respect to the posterior distribution, the value
of $w_{R}$ is censored at $\tau$ at age 13 for both females and
males. This leads to allocation to males at age 12 until the resource
is exhausted, resulting in the same allocation as seen in Figure \ref{fig:bayesian}.
This illustrates that the plug-in and the Bayesian rules could generate
different allocations under $\lambda=0$. 
\begin{figure}
\caption{Voucher allocations under directionally differentiable welfare}
\medskip{}

{\small\textbf{Note}}{\small : The upper panels show the plug-in rule,
and the lower panels show the Bayesian rule. The color intensity represents
the density of each cell in $F_{X}$, with darker shades indicating
higher density.}\label{fig:ambig}

\includegraphics[scale=0.55]{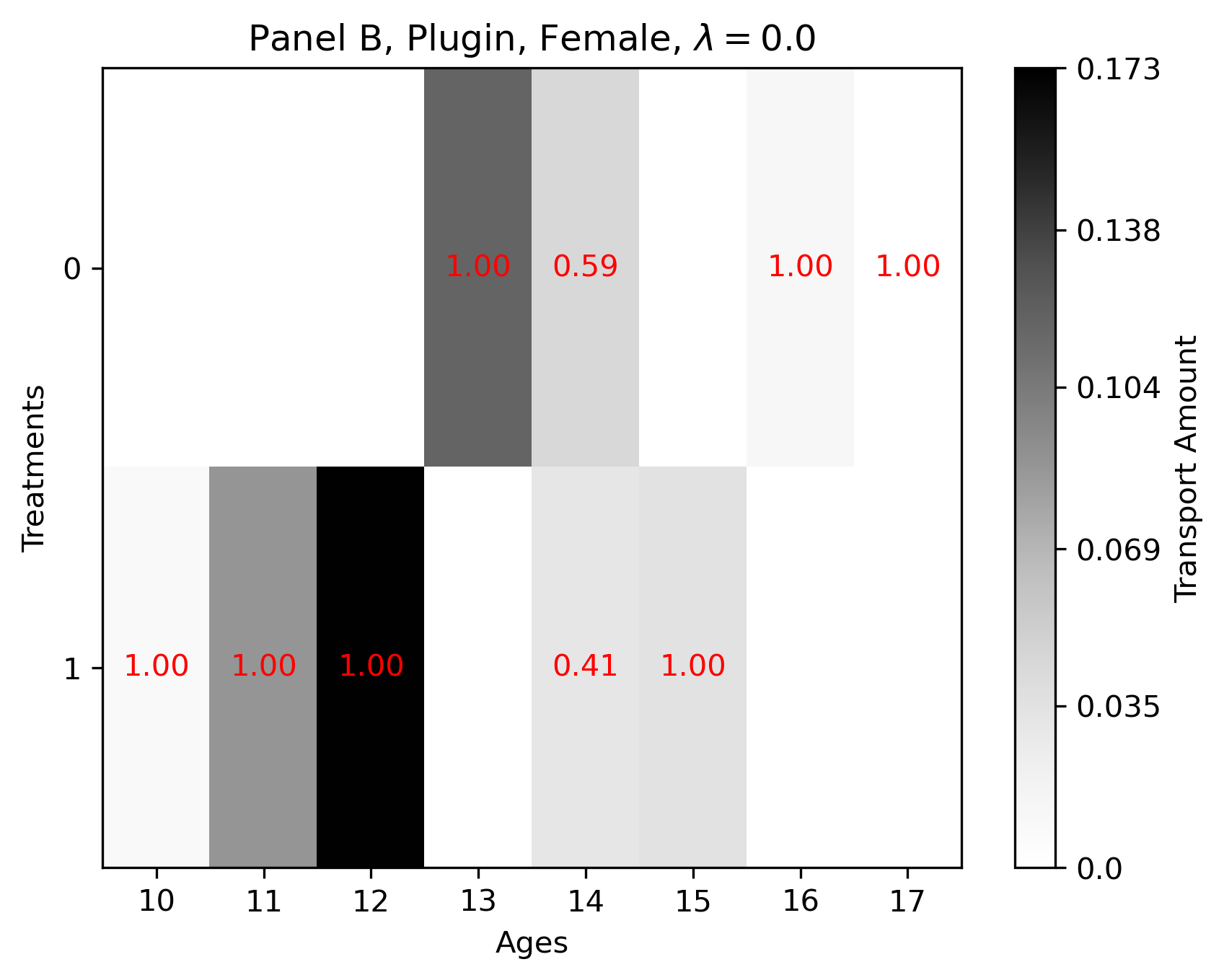}\includegraphics[scale=0.55]{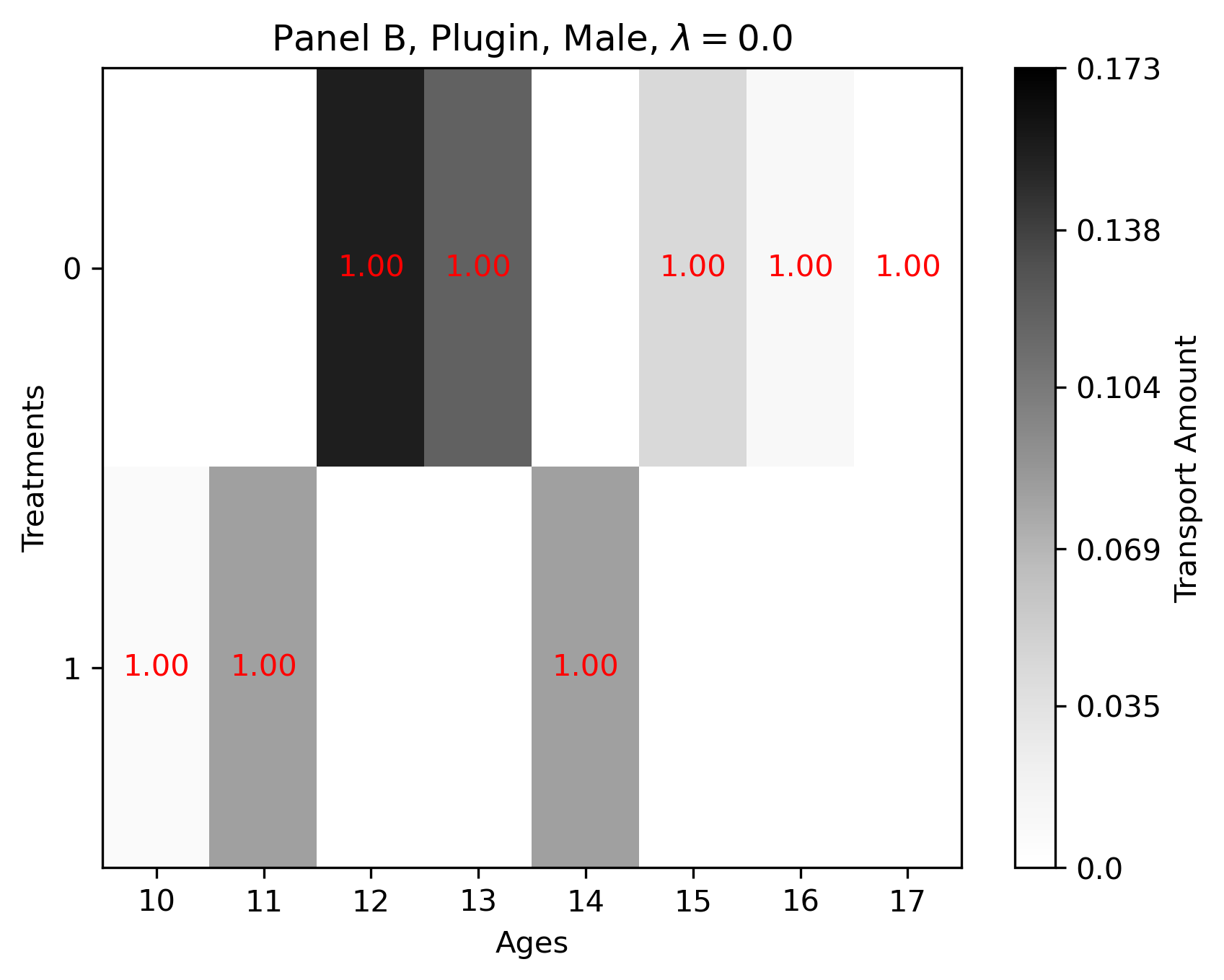}\linebreak{}

\includegraphics[scale=0.55]{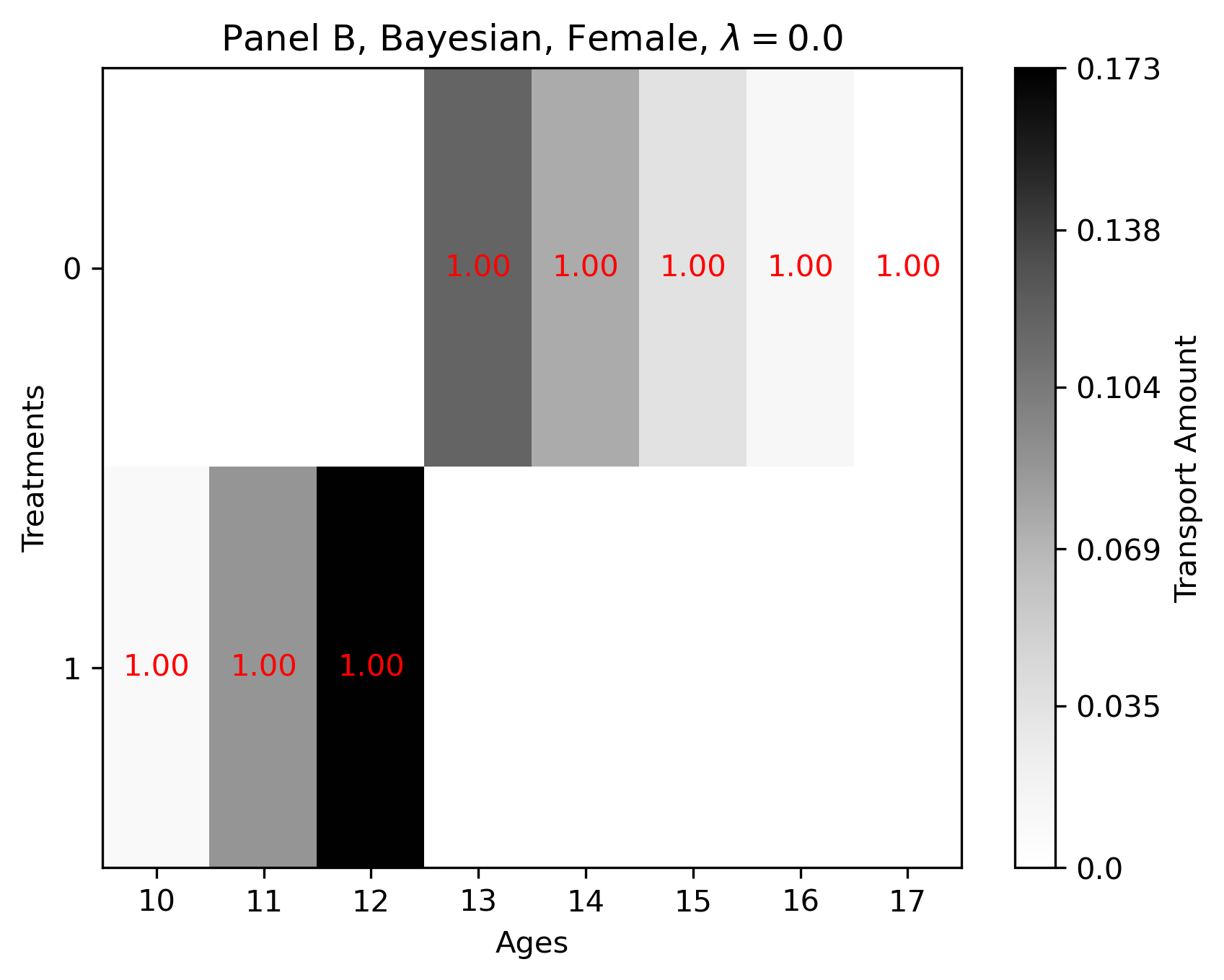}\includegraphics[scale=0.55]{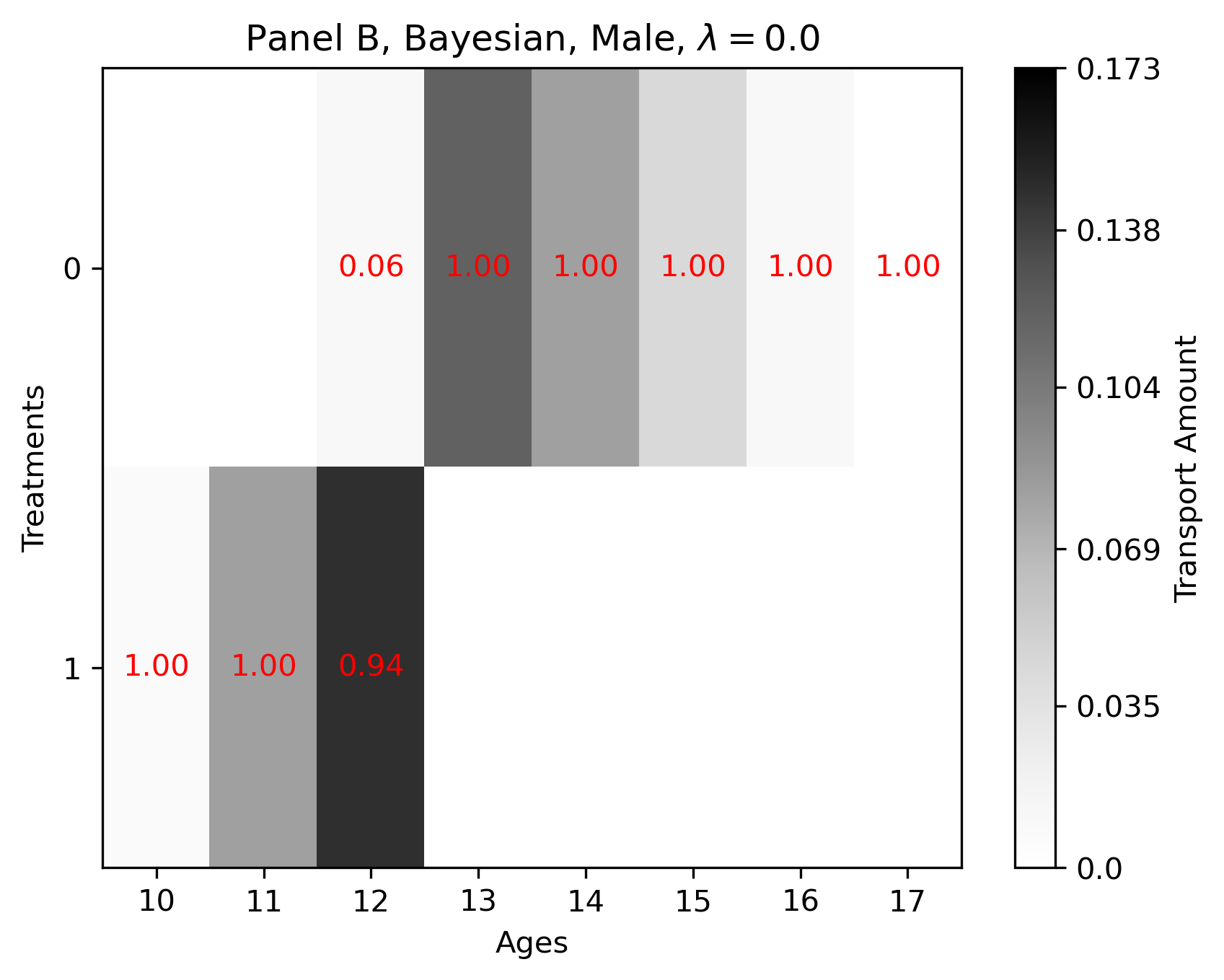}
\end{figure}

\section{Conclusion\label{sec:Conclusion}}

We studied the decision-theoretic optimality of treatment assignment
rules under capacity constraints on available treatments. Since such
constraints complicate the analysis of optimal rules, we transformed
the planner's constrained maximization problem into the unconstrained
one using tools from optimal transport theory. This reformulation
allows us to search for optimal rules in terms of couplings that automatically
satisfy the capacity constraints. We investigated two rules previously
studied in the literature---the plug-in rule and the Bayesian rule.
Both are average optimal when the planner's utility function is smooth;
however, the plug-in rule may no longer be optimal when the planner's
utility function is only directionally differentiable. A simulation
study supports our theoretical predictions. We demonstrated our methods
with a voucher assignment problems for private secondary school attendance
using data from \citet{angrist2006long}.

While we focused on average optimality as the optimality criterion,
asymptotic minimax optimality is also a widely used benchmark in local
asymptotics frameworks. \citet{kido2023locally} provides an asymptotic
minimax optimality result when the ATE is partially identified and
there are no constraints on available treatments. In that setting,
the plug-in rule becomes optimal only when the oracle rule is (Hadamard)
differentiable. Given that \citet{hirano2009asymptotics} show the
minimax optimality of the plug-in rule under the full differentiability
conditions, we conjecture that the modes of differentiability of the
planner's utility function $w$ plays a key role in the minimax optimality
of the plug-in rule in our setting. 

\newpage{}

\appendix

\section{\label{sec:Proof-of-Theorem}Proof of Theorem \ref{thm:optimality}}

\subsection{Preliminaries}

Our proof of Theorem \ref{thm:optimality} proceeds as follows. Lemma
\ref{lem:lower_bound} provides sufficient conditions for a rule $\left\{ \mu_{n}\right\} $
in $\mathcal{D}$ to be average optimal. Lemma \ref{lem:The-ex-post-Bayes}
shows that the Bayesian rule $\{\mu_{n}^{B}\}$ satisfies these conditions.
Lemmas \ref{lem:conclusion_from_conditionK}--\ref{lem:epb_Qn} are
used to establish Lemma \ref{lem:The-ex-post-Bayes}. Additional auxiliary
lemmas are relegated to Appendix \ref{sec:auxiliary-lemmas}. 

In what follows, the expectation $\mathbb{E}_{P_{\theta_{nh}}^{n}}$
can be understood as the outer expectation when $\mu_{n}$ is not
measurable. Also, $P_{h}$ denotes the joint law of $\Delta\sim N(h,I_{0}^{-1})$
and $U\sim\mathrm{Unif}[0,1]$ in the limit experiment. 

The following is a known result and can be found at \citet[Theorem 10.8]{van2000asymptotic}. 

\begin{prop}\label{prop:subpoly-vdv} Suppose that model is DQM at
$\theta_{0}$. Let $C_{n}$ be the ball of radius $M_{n}$ for a given,
arbitrary sequence $M_{n}\to\infty$. Further, suppose $\int\left\lVert \theta\right\rVert ^{p}\mathrm{d}\pi(\theta)<\infty$.
Then, for every measurable function $f$ that grows subpolynomially
of order $p$, 
\[
\int f(h)\boldsymbol{1}_{C_{n}^{c}}(h)\pi(\theta_{nh}|Z^{n})\mathrm{d}h=o_{P_{\theta_{0}}^{n}}(1)\quad\text{as }n\to\infty.
\]
\end{prop}

\subsection{Proof}

\selectlanguage{english}%
Let $A_{0}:=\arg\max_{\mu\in\mathcal{M}}W(\theta_{0},\mu)$ be the
set of couplings that maximize the welfare at $\theta_{0}$. Note
that $A_{0}$ is compact by \citet[Corollary 5.21]{villani2009optimal}.
Define the loss function in the limit experiment by 
\[
L_{\infty}(\mu,\Delta):=\int\left[\dot{W}_{\mathcal{M},0}^{*}[s]-\int\dot{w}_{\theta_{0}}(x,t;s)\mathrm{d}\mu\right]\mathrm{d}N(\Delta,I_{0}^{-1})(s),
\]
where $\dot{W}_{\mathcal{M},0}^{*}$ is Hadamard directional derivative
of $W_{\mathcal{M}}^{*}:\Theta\to\mathbb{R}$ at $\theta_{0}$, which
is defined in Lemma \ref{lem:HDD}. 

\begin{lem}\label{lem:lower_bound} Let $\{\mu_{n}\}\in\mathcal{D}$
be any sequence of decision rules that will be matched by $\mu_{\infty}$
in the limit experiment. It holds (i)
\[
\liminf_{n\to\infty}\int\sqrt{n}R(\mu_{n},\theta_{nh})\pi(\theta_{nh})\mathrm{d}h\ge\pi(\theta_{0})\int_{0}^{1}\int L_{\infty}(\mu_{\infty}(\Delta,u),\Delta)\mathrm{d}\Delta\mathrm{d}u,
\]
(ii) 
\[
\limsup_{n\to\infty}\int\sqrt{n}R(\mu_{n},\theta_{nh})\pi(\theta_{nh})\mathrm{d}h\le\pi(\theta_{0})\int_{0}^{1}\int L_{\infty}(\mu_{\infty}(\Delta,u),\Delta)\mathrm{d}\Delta\mathrm{d}u,
\]
and (iii) $\left\{ \mu_{n}\right\} $ is average optimal if the matched
rule $\mu_{\infty}$ of $\mu_{n}$ satisfies $\mu_{\infty}(\Delta)\in\arg\min_{\mu\in A_{0}}L_{\infty}(\mu,\Delta)$
where $\Delta\sim N(h,I_{0}^{-1})$. \end{lem}

\begin{proof} (i). By Fatou's lemma, 
\begin{align}
 & \liminf_{n\to\infty}\int R(\mu_{n},\theta_{nh})\pi(\theta_{nh})\mathrm{d}h\nonumber \\
 & \ge\pi(\theta_{0})\int\left\{ \liminf_{n\to\infty}\mathbb{E}_{P_{\theta_{nh}}^{n}}\sqrt{n}\left[W_{\mathcal{M}}^{*}(\theta_{nh})-W(\theta_{nh},\mu_{n}(Z^{n}))\right]\right\} \mathrm{d}h.\label{eq:Fatou}
\end{align}
For any $h$, 
\begin{align}
 & \mathbb{E}_{P_{\theta_{nh}}^{n}}\sqrt{n}\left[W_{\mathcal{M}}^{*}(\theta_{nh})-W(\theta_{nh},\mu_{n}(Z^{n}))\right]\label{eq:decomp}\\
 & =\int\sqrt{n}\left(\max_{\mu\in\mathcal{M}}\int w(\theta_{nh},x,t)\mathrm{d}\mu-\max_{\mu\in\mathcal{M}}\int w(\theta_{0},x,t)\mathrm{d}\mu\right)\mathrm{d}P_{\theta_{nh}}^{n}(z)\nonumber \\
 & \quad+\int\sqrt{n}\left(\max_{\mu\in\mathcal{M}}\int w(\theta_{0},x,t)\mathrm{d}\mu-\int w(\theta_{0},x,t)\mathrm{d}\mu_{n}(z)\right)\mathrm{d}P_{\theta_{nh}}^{n}(z)\nonumber \\
 & \quad-\int\sqrt{n}\left(\int w(\theta_{nh},x,t)\mathrm{d}\mu_{n}(z)-\int w(\theta_{0},x,t)\mathrm{d}\mu_{n}(z)\right)\mathrm{d}P_{\theta_{nh}}^{n}(z).\nonumber 
\end{align}
For the first term of the RHS of (\ref{eq:decomp}),
\begin{align*}
 & \int\sqrt{n}\left(\max_{\mu\in\mathcal{M}}\int w(\theta_{nh},x,t)\mathrm{d}\mu-\max_{\mu\in\mathcal{M}}\int w(\theta_{0},x,t)\mathrm{d}\mu\right)\mathrm{d}P_{\theta_{nh}}^{n}(z)\\
 & =\sqrt{n}\left(\max_{\mu\in\mathcal{M}}\int w(\theta_{nh},x,t)\mathrm{d}\mu-\max_{\mu\in\mathcal{M}}\int w(\theta_{0},x,t)\mathrm{d}\mu\right)\to\dot{W}_{\mathcal{M},0}^{*}[h],
\end{align*}
where the convergence follows from Lemma \ref{lem:HDD}. For the second
term of the RHS of (\ref{eq:decomp}), it is clear that
\[
\int\sqrt{n}\left(\max_{\mu\in\mathcal{M}}\int w(\theta_{0},x,t)\mathrm{d}\mu-\int w(\theta_{0},x,t)\mathrm{d}\mu_{n}(z)\right)\mathrm{d}P_{\theta_{nh}}^{n}(z)\ge0.
\]
For the third term of the RHS of (\ref{eq:decomp}),
\begin{align*}
 & \int\sqrt{n}\left(\int w(\theta_{nh},x,t)\mathrm{d}\mu_{n}(z)-\int w(\theta_{0},x,t)\mathrm{d}\mu_{n}(z)\right)\mathrm{d}P_{\theta_{nh}}^{n}(z)\\
 & =\int\int\dot{w}_{\theta_{0}}(x,t;h)\mathrm{d}\mu_{n}(z)\mathrm{d}P_{\theta_{nh}}^{n}(z)\\
 & \quad+\int\int\left\{ \text{\ensuremath{\sqrt{n}\left[w(\theta_{nh},x,t)-w(\theta_{0},x,t)\right]}}-\dot{w}_{\theta_{0}}(x,t;h)\right\} \mathrm{d}\mu_{n}(z)\mathrm{d}P_{\theta_{nh}}^{n}(z)
\end{align*}
where $\dot{w}_{\theta_{0}}(x,t;h)$ is the directional derivative
of $w(\theta,x,t)$ at $\theta_{0}$. By Assumption \ref{assu:w-dot}
(i),
\begin{align*}
 & \left|\int\int\left\{ \text{\ensuremath{\sqrt{n}\left[w(\theta_{nh},x,t)-w(\theta_{0},x,t)\right]}}-\dot{w}_{\theta_{0}}(x,t;h)\right\} \mathrm{d}\mu_{n}(z)\mathrm{d}P_{\theta_{nh}}^{n}(z)\right|\\
\le & \max_{(x,t)\in\mathcal{X}\times\mathcal{T}}\left|\text{\ensuremath{\sqrt{n}\left[w(\theta_{nh},x,t)-w(\theta_{0},x,t)\right]}}-\dot{w}_{\theta_{0}}(x,t;h)\right|\to0.
\end{align*}
Therefore, 
\begin{align*}
 & \liminf_{n\to\infty}\mathbb{E}_{P_{\theta_{nh}}^{n}}\sqrt{n}\left[W_{\mathcal{M}}^{*}(\theta_{nh})-W(\theta_{nh},\mu_{n}(Z^{n}))\right]\\
 & \ge\dot{W}_{\mathcal{M},0}^{*}[h]-\limsup_{n\to\infty}\int\int\dot{w}_{\theta_{0}}(x,t;h)\mathrm{d}\mu_{n}(z)\mathrm{d}P_{\theta_{nh}}^{n}(z)\\
 & \equiv\dot{W}_{\mathcal{M},0}^{*}[h]-\limsup_{n\to\infty}\int\varphi(\mu_{n}(z))\mathrm{d}P_{\theta_{nh}}^{n}(z)\\
 & \ge\dot{W}_{\mathcal{M},0}^{*}[h]-\mathbb{E}_{(\Delta,U)\sim P_{h}}\varphi(\mu_{\infty}(\Delta,U))
\end{align*}
from the portmanteau theorem in metric spaces since the map $\varphi:\mu\mapsto\int\dot{w}_{\theta_{0}}(x,t;h)\mathrm{d}\mu$
is bounded continuous on $\mathcal{M}$. Thus, the RHS of (\ref{eq:Fatou})
is bounded below by 
\begin{align*}
 & \pi(\theta_{0})\int\liminf_{n\to\infty}\mathbb{E}_{P_{\theta_{nh}}^{n}}\sqrt{n}\left[W_{\mathcal{M}}^{*}(\theta_{nh})-W(\theta_{nh},\mu_{n}(Z^{n}))\right]\mathrm{d}h\\
 & \ge\pi(\theta_{0})\int\mathbb{E}_{P_{h}}\left[\dot{W}_{\mathcal{M},0}^{*}[h]-\varphi(\mu_{\infty})\right]\mathrm{d}h\\
 & =\pi(\theta_{0})\int\int_{0}^{1}\int\left[\dot{W}_{\mathcal{M},0}^{*}[h]-\varphi(\mu_{\infty}(\Delta,u))\right]\mathrm{d}N(h,I_{0}^{-1})(\Delta)\mathrm{d}u\mathrm{d}h\\
 & =\pi(\theta_{0})\int_{0}^{1}\int\int\left[\dot{W}_{\mathcal{M},0}^{*}[s]-\int\dot{w}_{\theta_{0}}(x,t;s)\mathrm{d}\mu_{\infty}(\Delta,u)\right]\mathrm{d}N(\Delta,I_{0}^{-1})(s)\mathrm{d}\Delta\mathrm{d}u,
\end{align*}
where the last equality follows by Tonelli's theorem since the integrand
is nonnegative. By the definition of $L_{\infty}$, the last display
is equal to
\[
\int_{0}^{1}\int L_{\infty}(\mu_{\infty}(\Delta,u),\Delta)\mathrm{d}\Delta\mathrm{d}u.
\]
It should be noted that $L_{\infty}(\mu_{\infty}(\Delta,u),\Delta)$
can depend on $u$ only through $\mu_{\infty}$. 

(ii). If we admit applying the reverse Fatou lemma, we obtain 
\[
\limsup_{n\to\infty}\int\sqrt{n}R(\mu_{n},\theta_{nh})\pi(\theta_{nh})\mathrm{d}h\le\pi(\theta_{0})\int\left\{ \limsup_{n\to\infty}\mathbb{E}_{P_{\theta_{nh}}^{n}}\sqrt{n}\left[W_{\mathcal{M}}^{*}(\theta_{nh})-W(\theta_{nh},\mu_{n}(Z^{n}))\right]\right\} \mathrm{d}h.
\]
The argument can be carried out analogously to (i). However, for the
second term of the RHS of (\ref{eq:decomp}),
\begin{align*}
 & \int\sqrt{n}\left(\max_{\mu\in\mathcal{M}}\int w(\theta_{0},x,t)\mathrm{d}\mu-\int w(\theta_{0},x,t)\mathrm{d}\mu_{n}(z)\right)\mathrm{d}P_{\theta_{nh}}^{n}(z)\\
 & =\int_{\{\mu_{n}\in A_{0}\}}\sqrt{n}\left(\max_{\mu\in\mathcal{M}}\int w(\theta_{0},x,t)\mathrm{d}\mu-\int w(\theta_{0},x,t)\mathrm{d}\mu_{n}(z)\right)\mathrm{d}P_{\theta_{nh}}^{n}(z)\\
 & \quad+\int_{\{\mu_{n}\notin A_{0}\}}\sqrt{n}\left(\max_{\mu\in\mathcal{M}}\int w(\theta_{0},x,t)\mathrm{d}\mu-\int w(\theta_{0},x,t)\mathrm{d}\mu_{n}(z)\right)\mathrm{d}P_{\theta_{nh}}^{n}(z)\\
 & =\int_{\{\mu_{n}\notin A_{0}\}}\sqrt{n}\left(\max_{\mu\in\mathcal{M}}\int w(\theta_{0},x,t)\mathrm{d}\mu-\int w(\theta_{0},x,t)\mathrm{d}\mu_{n}(z)\right)\mathrm{d}P_{\theta_{nh}}^{n}(z),
\end{align*}
where the second equality follows from $\max_{\mu\in\mathcal{M}}\int w(\theta_{0},x,t)\mathrm{d}\mu=\max_{\mu\in A_{0}}\int w(\theta_{0},x,t)\mathrm{d}\mu$.
Further,
\begin{align*}
 & \int_{\{\mu_{n}\notin A_{0}\}}\sqrt{n}\left(\max_{\mu\in\mathcal{M}}\int w(\theta_{0},x,t)\mathrm{d}\mu-\int w(\theta_{0},x,t)\mathrm{d}\mu_{n}(z)\right)\mathrm{d}P_{\theta_{nh}}^{n}(z)\\
 & \le\sqrt{n}P_{\theta_{nh}}^{n}\left(\mu_{n}\notin A_{0}\right)\times2\max_{(x,t)\in\mathcal{X}\times\mathcal{T}}\left|w(\theta_{0},x,t)\right|\to0,
\end{align*}
as $n\to\infty$ from the definition of $\{\mu_{n}\}\in\mathcal{D}$.
Thus, applying the portmanteau theorem yields
\begin{equation}
\limsup_{n\to\infty}\mathbb{E}_{P_{\theta_{nh}}^{n}}\sqrt{n}\left[W_{\mathcal{M}}^{*}(\theta_{nh})-W(\theta_{nh},\mu_{n}(Z^{n}))\right]\le\dot{W}_{\mathcal{M},0}^{*}[h]-\mathbb{E}_{P_{h}}\left[\int\dot{w}_{\theta_{0}}(x,t;h)\mathrm{d}\mu_{\infty}\right].\label{eq:limsup_integrand}
\end{equation}
Then we can conclude similarly as in (i).

Finally, we check the validity of applying the reverse Fatou lemma.
Let 
\[
g_{n}(h):=\mathbb{E}_{P_{\theta_{nh}}^{n}}\sqrt{n}\left[W_{\mathcal{M}}^{*}(\theta_{nh})-W(\theta_{nh},\mu_{n}(Z^{n}))\right].
\]
From (\ref{eq:limsup_integrand}), there exists $N$ such that for
all $n>N$, 
\[
g_{n}(h)\le\dot{W}_{\mathcal{M},0}^{*}[h]-\mathbb{E}_{P_{h}}\left[\int\dot{w}_{\theta_{0}}(x,t;h)\mathrm{d}\mu_{\infty}\right]+\varepsilon.
\]
Note that $g_{n}(h)$ is bounded for all $n\le N$. Then define
\[
g(h):=\max\left\{ g_{1}(h),\dots,g_{N}(h),\dot{W}_{\mathcal{M},0}^{*}[h]-\mathbb{E}_{P_{h}}\left[\int\dot{w}_{\theta_{0}}(x,t;h)\mathrm{d}\mu_{\infty}\right]+\varepsilon\right\} .
\]
Thus we have $g_{n}(h)\pi(\theta_{nh})\le\sup_{\theta\in\Theta}\pi(\theta)\times g(h)$
for all $n\in\mathbb{N}$, with $\sup_{\theta\in\Theta}\pi(\theta)<\infty$
from Assumption \ref{assu:prior} (i). 

(iii). Fix any sequence of rules $\{\mu_{n}^{\prime}\}\in\mathcal{D}$,
and let $\mu_{\infty}^{\prime}$ be the matched rule of $\mu_{n}^{\prime}$.
Combining (i) and (ii) yields 
\begin{align*}
\liminf_{n\to\infty}\int\sqrt{n}R(\mu_{n},\theta_{nh})\pi(\theta_{nh})\mathrm{d}h & \le\limsup_{n\to\infty}\int\sqrt{n}R(\mu_{n},\theta_{nh})\pi(\theta_{nh})\mathrm{d}h\\
 & \le\pi(\theta_{0})\int_{0}^{1}\int L_{\infty}(\mu_{\infty}(\Delta,u),\Delta)\mathrm{d}\Delta\mathrm{d}u\\
 & \le\pi(\theta_{0})\int_{0}^{1}\int L_{\infty}(\mu_{\infty}^{\prime}(\Delta,u),\Delta)\mathrm{d}\Delta\mathrm{d}u\\
 & \le\liminf_{n\to\infty}\int\sqrt{n}R(\mu_{n}^{\prime},\theta_{nh})\pi(\theta_{nh})\mathrm{d}h,
\end{align*}
which completes the proof. \end{proof}\selectlanguage{american}%

Let $\left(\mathbb{D},\left\lVert \cdot\right\rVert _{\mathbb{D}}\right)$
be the product metric space induced by $\left(\mathcal{M},d_{W}\right)$
and $\left([0,1].\left|\cdot\right|\right)$. Let $\ell^{\infty}(\mathbb{D}):=\{f:\mathbb{D}\to\mathbb{R}:\sup_{(\mu,\varepsilon)\in\mathbb{D}}|f(\mu,\varepsilon)|<\infty\}$.
Define 
\begin{align*}
(\mu,\varepsilon)\mapsto & \mathcal{Q}_{n}(\mu,\varepsilon;z):=\int\left[\int\sqrt{n}\left(w(\theta_{nh},x,t)-w(\theta_{0},x,t)\right)\mathrm{d}\mu\right]\pi_{n}(\theta_{nh}|z)\mathrm{d}h-\varepsilon H(\mu),\\
(\mu,\varepsilon)\mapsto & Q_{\infty}(\mu,\varepsilon;\Delta):=\int\left[\int\dot{w}_{\theta_{0}}(x,t;h)\mathrm{d}\mu\right]\mathrm{d}N(\Delta,I_{0}^{-1})(h)-\varepsilon H(\mu).
\end{align*}
In the same manner as the definition of $\mu_{n}^{B}(Z^{n})$, we
define $\mu_{\infty}^{*}(\Delta)$ as the limit of $\mu_{\infty,\varepsilon}^{*}(\Delta)$
where 
\[
\mu_{\infty,\varepsilon}^{*}(\Delta):=\arg\max_{\mu\in A_{0}}\int\int\dot{w}_{\theta_{0}}(x,t;s)\mathrm{d}N(\Delta,I_{0}^{-1})(s)\mathrm{d}\mu-\varepsilon H(\mu).
\]
It is easy to see that $\mu_{\infty}^{*}(\Delta)\in\arg\min_{\mu\in A_{0}}L_{\infty}(\mu,\Delta)$. 

\begin{lem}\label{lem:The-ex-post-Bayes} The Bayesian rule $\{\mu_{n}^{B}(Z^{n})\}$
satisfies $\sqrt{n}P_{\theta_{nh}}^{n}\left(\mu_{n}^{B}(Z^{n})\notin A_{0}\right)\to0$
and (ii) $\mu_{n}^{B}\stackrel{h}{\rightsquigarrow}\mu_{\infty}^{*}$
as $n\to\infty$. \end{lem}

\begin{proof} (i). We claim that for any $\theta_{0}\in\Theta$,
there are $\overline{n}$ and $\varepsilon_{n}^{\prime}(\overline{n})$
(which is at the order of $n^{\alpha+1}$ for some $\alpha\ge1$)
such that for all $n\ge\overline{n}$,
\[
P_{\theta_{nh}}^{n}\left(\mu_{n}^{B}(z)\notin A_{0}\right)\le P_{\theta_{nh}}^{n}\left(\pi_{n}\left(N_{1/n}(\theta_{0})^{c}|z\right)>2\varepsilon_{n}^{\prime}\right),
\]
where $N_{\varepsilon}(\theta_{0}):=\left\{ \theta:\rVert\theta-\theta_{0}\rVert<\varepsilon\right\} $
for $\varepsilon>0$, and $\pi_{n}\left(A|z\right):=\int_{A}\pi_{n}(\theta|z)\mathrm{d}\theta$.\footnote{This statement is an adaptation of \citet[Lemma 8]{christensen2023optimal}.
Their proof cannot directly apply to our setting because their arguments
could fail when the set of actions, $\mathcal{M}$ in our notation,
is not finite. } Then the conclusion follows since \citet[Lemmas 9--11]{christensen2023optimal}
imply that $\sqrt{n}P_{\theta_{nh}}^{n}\left(\pi_{n}\left(N_{1/n}(\theta_{0})^{c}|z\right)>2\varepsilon_{n}^{\prime}\right)\to0$
as $n\to\infty$.

It is sufficient to show that for any $z$, 
\begin{equation}
\pi_{n}\left(N_{1/n}(\theta_{0})^{c}|z\right)\le2\varepsilon_{n}^{\prime}\implies\mu_{n}^{B}(z)\in A_{0}.\label{eq:promise}
\end{equation}
It is trivial if $A_{0}=\mathcal{M}$, so suppose $A_{0}\subsetneq\mathcal{M}$.
By Lemma \ref{lem:conclusion_from_conditionK} below, there exists
$\overline{n}$ such that $n\ge\overline{n}$ implies 
\[
\min_{\mu\in A_{0}}\int_{N_{1/n}(\theta_{0})}W(\theta,\mu)\pi_{n}(\theta|z)\mathrm{d}\theta>\sup_{\nu\not\in A_{0}}\int_{N_{1/n}(\theta_{0})}W(\theta,\nu)\pi_{n}(\theta|z)\mathrm{d}\theta
\]
Then there exists $\alpha=\alpha(\overline{n})\ge1$ such that 
\begin{align*}
\min_{\mu\in A_{0}}\int_{N_{1/n}(\theta_{0})}W(\theta,\mu)\pi_{n}(\theta|z)\mathrm{d}\theta & >\sup_{\nu\not\in A_{0}}\int_{N_{1/n}(\theta_{0})}W(\theta,\nu)\pi_{n}(\theta|z)\mathrm{d}\theta+\frac{1}{\overline{n}^{\alpha}},
\end{align*}
which implies that for $n\ge\overline{n}$,
\begin{align*}
\min_{\mu\in A_{0}}\int_{N_{1/n}(\theta_{0})}W(\theta,\mu)\pi_{n}(\theta|z)\mathrm{d}\theta & >\sup_{\nu\not\in A_{0}}\int_{N_{1/n}(\theta_{0})}W(\theta,\nu)\pi_{n}(\theta|z)\mathrm{d}\theta+\frac{1}{n^{\alpha}}.
\end{align*}
Thus, we have 
\begin{align*}
\min_{\mu\in A_{0}}\left(\int W(\theta,\mu)\pi_{n}(\theta|z)\mathrm{d}\theta\right) & =\min_{\mu\in A_{0}}\left(\int_{N_{1/n}(\theta_{0})}W(\theta,\mu)\pi_{n}(\theta|z)\mathrm{d}\theta+\int_{N_{1/n}(\theta_{0})^{c}}W(\theta,\mu)\pi_{n}(\theta|z)\mathrm{d}\theta\right)\\
 & \ge\min_{\mu\in A_{0}}\int_{N_{1/n}(\theta_{0})}W(\theta,\mu)\pi_{n}(\theta|z)\mathrm{d}\theta+\min_{\mu\in A_{0}}\int_{N_{1/n}(\theta_{0})^{c}}W(\theta,\mu)\pi_{n}(\theta|z)\mathrm{d}\theta\\
 & \ge\min_{\mu\in A_{0}}\int_{N_{1/n}(\theta_{0})}W(\theta,\mu)\pi_{n}(\theta|z)\mathrm{d}\theta-\pi_{n}\left(N_{1/n}(\theta_{0})^{c}|z\right)M,
\end{align*}
where $\sup_{\theta,\mu}\left|W(\theta,\mu)\right|\le M<\infty$.
Also, 
\begin{align*}
\sup_{\nu\not\in A_{0}}\int W(\theta,\nu)\pi_{n}(\theta|z)\mathrm{d}\theta\le\sup_{\nu\not\in A_{0}}\int_{N_{1/n}(\theta_{0})}W(\theta,\nu)\pi_{n}(\theta|z)\mathrm{d}\theta+\pi_{n}\left(N_{1/n}(\theta_{0})^{c}|z\right)M.
\end{align*}
For the promise for (\ref{eq:promise}), we choose $\varepsilon_{n}^{\prime}>0$
that satisfies $\left(n^{\alpha+1}2M\right)^{-1}\le2\varepsilon_{n}^{\prime}<\left(n^{\alpha}2M\right)^{-1}$,
which leads to 
\[
\pi_{n}\left(N_{\varepsilon}(\theta_{0})^{c}|z\right)<\frac{1}{n^{\alpha}2M}.
\]
Then it follows that
\begin{align*}
 & \min_{\mu\in A_{0}}\int W(\theta,\mu)\pi_{n}(\theta|z)\mathrm{d}\theta-\sup_{\nu\not\in A_{0}}\int W(\theta,\nu)\pi_{n}(\theta|z)\mathrm{d}\theta\\
\ge & \min_{\mu\in A_{0}}\int_{N_{1/n}(\theta_{0})}W(\theta,\mu)\pi_{n}(\theta|z)\mathrm{d}\theta-\sup_{\nu\not\in A_{0}}\int_{N_{1/n}(\theta_{0})}W(\theta,\nu)\pi_{n}(\theta|z)\mathrm{d}\theta-2\pi_{n}\left(N_{1/n}(\theta_{0})^{c}|z\right)M\\
\ge & \frac{1}{n^{\alpha}}-2\pi_{n}\left(N_{1/n}(\theta_{0})^{c}|z\right)M>0,
\end{align*}
which implies
\[
\min_{\mu\in A_{0}}\int W(\theta,\mu)\pi_{n}(\theta|z)\mathrm{d}\theta>\sup_{\nu\not\in A_{0}}\int W(\theta,\nu)\pi_{n}(\theta|z)\mathrm{d}\theta.
\]
Thus we conclude $\mu_{n}^{B}(z)\in A_{0}$. 

(ii).  From the first statement, it follows that the asymptotic distribution
of $\mu_{n}^{B}(Z^{n})$ has the support only on $A_{0}$. Hence,
for sufficiently large $n$, $\mu_{n}^{B}$ equivalently solves 
\[
\mu_{n}^{B}(z)\in\arg\max_{\mu\in A_{0}}\mathcal{Q}_{n}(\mu,0;z).
\]
By the definition of $\mu_{n}^{B}$, 
\[
\mathcal{Q}_{n}(\mu_{n}^{B}(z),0;z)=\lim_{\varepsilon\downarrow0}\max_{\mu\in A_{0}}\mathcal{Q}_{n}(\mu,\varepsilon;z).
\]
Take any closed subset $G$ of $\mathcal{M}$. Note that this closedness
is in terms of $(\mathcal{M},d_{W})$. By the Portmanteau lemma, it
is sufficient to show that 
\[
\limsup_{n\to\infty}P_{\theta_{nh}}^{n}\left(\mu_{n}^{B}(z)\in G\right)\le P_{h}\left(\mu_{\infty}^{*}(\Delta)\in G\right)
\]
for the conclusion.

By Lemma \ref{lem:events_equivalence}, for each $n$, 
\[
\left\{ \mu_{n}^{B}(z)\in G\right\} =\left\{ \lim_{\varepsilon\downarrow0}\max_{\mu\in A_{0}\cap G}\mathcal{Q}_{n}(\mu,\varepsilon;z)=\lim_{\varepsilon\downarrow0}\max_{\mu\in A_{0}}\mathcal{Q}_{n}(\mu,\varepsilon;z)\right\} .
\]
By Lemmas \ref{lem:process_conv} and \ref{lem:epb_Qn}, it follows
that $\mathcal{Q}_{n}\stackrel{h}{\rightsquigarrow}Q_{\infty}$ in
$\mathcal{F}$ as $n\to\infty$, where 
\[
\mathcal{F}=\left\{ f\in\ell^{\infty}(\mathbb{D}):\lim_{\varepsilon\downarrow0}\max_{\mu\in A_{0}}f(\mu,\varepsilon)\text{ exists}\right\} .
\]
By Lemma \ref{lem:cont_operator}, the operator $f\mapsto\lim_{\varepsilon\downarrow0}\max_{\mu\in S}f(\mu,\varepsilon)$
is continuous for any closed $S\subset A_{0}$ at any $f\in\mathcal{F}$.
Applying the continuous mapping theorem yields 
\[
\lim_{\varepsilon\downarrow0}\max_{\mu\in S}\mathcal{Q}_{n}(\mu,\varepsilon;z)\stackrel{h}{\rightsquigarrow}\lim_{\varepsilon\downarrow0}\max_{\mu\in S}Q_{\infty}(\mu,\varepsilon;\Delta)\quad\text{as }n\to\infty.
\]
Then 
\begin{align*}
 & \limsup_{n\to\infty}P_{\theta_{nh}}^{n}\left(\mu_{n}^{B}(z)\in G\right)\\
= & \limsup_{n\to\infty}P_{\theta_{nh}}^{n}\left(\lim_{\varepsilon\downarrow0}\max_{\mu\in A_{0}\cap G}\mathcal{Q}_{n}(\mu,\varepsilon;z)=\lim_{\varepsilon\downarrow0}\max_{\mu\in A_{0}}\mathcal{Q}_{n}(\mu,\varepsilon;z)\right)\\
\le & P_{h}\left(\lim_{\varepsilon\downarrow0}\max_{\mu\in A_{0}\cap G}Q_{\infty}(\mu,\varepsilon;\Delta)=\lim_{\varepsilon\downarrow0}\max_{\mu\in A_{0}}Q_{\infty}(\mu,\varepsilon;\Delta)\right),
\end{align*}
where the inequality follows from the Portmanteau lemma. Since we
can derive the equivalence of the events 
\[
\left\{ \mu_{\infty}^{*}(\Delta)\in G\right\} =\left\{ \lim_{\varepsilon\downarrow0}\max_{\mu\in A_{0}\cap G}Q_{\infty}(\mu,\varepsilon;\Delta)=\lim_{\varepsilon\downarrow0}\max_{\mu\in A_{0}}Q_{\infty}(\mu,\varepsilon;\Delta)\right\} 
\]
in the same manner, applying the Portmanteau lemma again yields $\mu_{n}^{B}\stackrel{h}{\rightsquigarrow}\mu_{\infty}^{*}$
as $n\to\infty$.  \end{proof}

\begin{lem} \label{lem:conclusion_from_conditionK} There exists
$\overline{n}>0$ such that for all $n\ge\overline{n}$, 
\begin{equation}
\min_{\mu\in A_{0}}\int_{N_{1/n}(\theta_{0})}W(\theta,\mu)\pi_{n}(\theta|z)\mathrm{d}\theta>\sup_{\nu\not\in A_{0}}\int_{N_{1/n}(\theta_{0})}W(\theta,\nu)\pi_{n}(\theta|z)\mathrm{d}\theta.\label{eq:conclusion_from_conditionK}
\end{equation}

\end{lem}

\begin{proof} Let $g_{n}(\mu):=\int_{N_{1/n}(\theta_{0})}W(\theta,\mu)\pi_{n}(\theta)\mathrm{d}\theta$
and $V_{1/n}:=\int_{N_{1/n}(\theta_{0})}\pi_{n}(\theta)\mathrm{d}\theta$.
First, we claim that $V_{1/n}^{-1}g_{n}(\mu)$ converges to $W(\theta_{0},\mu)$
uniformly over $\mathcal{M}$ as $n\to\infty$; i.e., for all $\eta>0$,
there exists $n_{\eta}$ such that for all $\mu\in\mathcal{M}$,
\[
n\ge n_{\eta}\implies\left|V_{1/n}^{-1}g_{n}(\mu)-W(\theta_{0},\mu)\right|<\eta.
\]
To show this claim, we argue that (i) for each $\mu\in\mathcal{M}$,
$V_{1/n}^{-1}g_{n}(\mu)$ converges to $W(\theta_{0},\mu)$ in pointwise,
and (ii) $\left\{ g_{n}\right\} _{n\in\mathbb{N}}$ is equicontinuous;
i.e., for all $\eta>0$ and all $\mu\in\mathcal{M}$, there exists
$\delta_{(\eta,\mu)}>0$ such that for all $n\in\mathbb{N}$ and all
$\nu\in\mathcal{M}$, 
\[
d_{W}(\mu,\nu)<\delta_{(\eta,\mu)}\implies\left|g_{n}(\mu)-g_{n}(\nu)\right|<\eta.
\]
 Combining with the compactness of $\mathcal{M}$, the uniform convergence
follows from these two.

To see (i), note that 
\begin{equation}
\left|V_{1/n}^{-1}g_{n}(\mu)-W(\theta_{0},\mu)\right|\le V_{1/n}^{-1}\int_{N_{1/n}(\theta_{0})}\left|W(\theta,\mu)-W(\theta_{0},\mu)\right|\pi_{n}(\theta|z)\mathrm{d}\theta.\label{eq:pointwise}
\end{equation}
Fix $\eta>0$. Since the map $\theta\mapsto W(\theta,\mu)$ is continuous
at $\theta_{0}$, there exists $\delta>0$ such that 
\[
\theta\in N_{\delta}(\theta_{0})\implies\left|W(\theta,\mu)-W(\theta_{0},\mu)\right|<\eta.
\]
Then for all $n$ with $n^{-1}<\delta$, RHS of (\ref{eq:pointwise})
is bounded above by $\eta$. 

To see (ii), fix $\eta>0$ and $\mu\in\mathcal{M}.$ First, note that
$g_{n}(\mu)$ can be written as 
\[
g_{n}(\mu)=\mathbb{E}_{\mu}[\Psi_{n}(x,t)],\quad\Psi_{n}(x,t):=\int_{N_{1/n}(\theta_{0})}w(\theta,x,t)\pi_{n}(\theta|z)\mathrm{d}\theta,
\]
where $\mathbb{E}_{\mu}$ denotes the expectation with respect to
the coupling $\mu$. Note that $(x,t)\mapsto w(\theta,x,t)$ is uniformly
continuous and bounded since $\mathcal{X}\times\mathcal{T}$ is compact.
Then a function $(x,t)\mapsto w_{k}(\theta,x,t)$ defined by
\[
w_{k}(\theta,x,t):=\inf\left\{ w(\theta,x^{\prime},t^{\prime})+kd((x,t),(x^{\prime},t^{\prime})):(x^{\prime},t^{\prime})\in\mathcal{X}\times\mathcal{T}\right\} 
\]
is $k$-Lipschitz continuous, and converges uniformly to $w(\theta,x,t)$
from below as $k\to\infty$ (see e.g., \citet[Theorem 6.8]{heinonen2001lectures}).
Lemma \ref{lem:modulus} below further extends that the convergence
holds uniformly over $\Theta$; i.e., for all $\eta>0$, there is
a sufficiently large $K=K(\eta)$ such that 
\[
\sup_{\theta\in\Theta}\max_{(x,t)\in\mathcal{X}\times\mathcal{T}}\left|w(\theta,x,t)-w_{K}(\theta,x,t)\right|<\frac{\eta}{3}.
\]
 Given this $K$, define 
\[
\Psi_{n}^{K}(x,t):=\int_{N_{1/n}(\theta_{0})}w_{K}(\theta,x,t)\pi_{n}(\theta|z)\mathrm{d}\theta.
\]
Then $\Psi_{n}^{K}(x,t)$ is also Lipschitz continuous whose Lipschitz
constant is less than or equal to $K$.  Therefore,
\begin{align*}
\left|g_{n}(\mu)-g_{n}(\nu)\right|\le & \mathbb{E}_{\mu}\left|\Psi_{n}(x,t)-\Psi_{n}^{K}(x,t)\right|+\left|\mathbb{E}_{\mu}[\Psi_{n}^{K}(x,t)]-\mathbb{E}_{\nu}[\Psi_{n}^{K}(x,t)]\right|+\mathbb{E}_{\nu}\left|\Psi_{n}^{K}(x,t)-\Psi_{n}(x,t)\right|\\
< & \frac{2}{3}\eta V_{1/n}+Kd_{W}(\mu,\nu)\le\frac{2}{3}\eta+Kd_{W}(\mu,\nu),
\end{align*}
where the second inequality follows from the Kantorovich-Rubinstein
duality \citep[Theorem 5.10]{villani2009optimal}. Thus, we obtain
\[
d_{W}(\mu,\nu)<\frac{\eta}{3K}\implies\left|g_{n}(\mu)-g_{n}(\nu)\right|<\eta.
\]
Therefore $\left\{ g_{n}\right\} _{n\in\mathbb{N}}$ is an equicontinuous
family.

Finally, we show (\ref{eq:conclusion_from_conditionK}). By Assumption
\ref{assu:CMS25:lem:8}, there exists $\eta>0$ such that for all
$\mu\in A_{0}$,
\begin{equation}
W(\theta_{0},\mu)>\sup_{\nu\not\in A_{0}}W(\theta_{0},\nu)+\eta.\label{eq:def_of_A0}
\end{equation}
By the uniform convergence shown above, there exists $n_{\eta}$ such
that for all $\mu\in\mathcal{M}$, 
\[
n\ge n_{\eta}\implies\left|V_{1/n}^{-1}g_{n}(\mu)-W(\theta_{0},\mu)\right|<\frac{\eta}{3}.
\]
Then fix $n\ge n_{\eta}$, and let $\mu_{n}\in\arg\min_{\mu\in A_{0}}g_{n}(\mu)$.
For any $\nu\not\in A_{0}$, we obtain  
\begin{align*}
\min_{\mu\in A_{0}}g_{n}(\mu) & =g_{n}(\mu_{n})\\
 & >V_{1/n}\left[W(\theta_{0},\mu_{n})-\frac{\eta}{3}\right]\\
 & >V_{1/n}\left[W(\theta_{0},\nu)+\eta-\frac{\eta}{3}\right]\\
 & >g_{n}(\nu)+\frac{\eta}{3},
\end{align*}
where the second inequality follows from (\ref{eq:def_of_A0}). Thus
we obtain 
\[
\min_{\mu\in A_{0}}g_{n}(\mu)\ge\left\{ \sup_{\nu\not\in A_{0}}g_{n}(\nu)\right\} +\frac{\eta}{3}.
\]
Since $\eta>0$ does not depend on $\nu\not\in A_{0}$, we conclude
$n\ge n_{\eta}$ implies $\min_{\mu\in A_{0}}g_{n}(\mu)>\sup_{\nu\not\in A_{0}}g_{n}(\nu)$.
\end{proof}

\begin{lem} \label{lem:modulus} For all $\eta>0$, there is a sufficiently
large $K=K(\eta)$ such that 
\[
\sup_{\theta\in\Theta}\max_{(x,t)\in\mathcal{X}\times\mathcal{T}}\left|w(\theta,x,t)-w_{K}(\theta,x,t)\right|<\eta.
\]
\end{lem}

\begin{proof} To simplify the notation, let $\mathcal{Y}=\mathcal{X}\times\mathcal{T}$.
Let $D$ be the diameter of $\mathcal{Y}$. Since $y\mapsto w(\theta,y)$
is continuous uniformly over $\Theta$ (Assumption \ref{assu:w} (ii))
and $\mathcal{Y}$ is compact (Assumption \ref{assu:villani_thm6-9}),
we have 
\[
\Omega(\delta):=\sup_{\theta\in\Theta}\sup_{d(y,y^{\prime})<\delta}\left|w(\theta,y)-w(\theta,y^{\prime})\right|\to0\quad\text{as }\delta\downarrow0.
\]
Fix any $\theta\in\Theta$ and $y\in\mathcal{Y}.$ Then for all $y^{\prime}\in\mathcal{Y}$,
\[
w(\theta,y^{\prime})\ge w(\theta,y)-\left|w(\theta,y)-w(\theta,y^{\prime})\right|\ge w(\theta,y)-\Omega(d(y,y^{\prime})),
\]
where the last inequality follows from the definition of $\Omega$.
Adding $kd(y,y^{\prime})$ to both sides and taking the infimum with
respect to $y^{\prime}$ yields 
\[
w_{k}(\theta,y)\ge w(\theta,y)+\inf_{r\in[0,D]}\left\{ kr-\Omega(r)\right\} ,
\]
which implies 
\[
w(\theta,y)-w_{k}(\theta,y)\le\sup_{r\in[0,D]}\phi_{k}(r),\quad\phi_{k}(r):=\Omega(r)-kr.
\]
By the definition of $\Omega$, there exists $\delta>0$ such that
$\Omega(\delta)<\eta$. Note that $\Omega$ is non-decreasing function.
Then if $0\le r<\delta$, we have $\phi_{k}(r)\le\Omega(\delta)<\eta$.
If $r\ge\delta$, we have $\phi_{k}(r)\le\Omega(D)-k\delta$. Thus
$\sup_{r\in[0,D]}\phi_{k}(r)\le\left\{ \Omega(\delta),\Omega(D)-k\delta\right\} $.
Hence, for sufficiently large $K$, it follows $\sup_{r\in[0,D]}\phi_{K}(r)<\eta$.
This implies 
\[
\sup_{\theta\in\Theta}\max_{y\in\mathcal{Y}}\left\{ w(\theta,y)-w_{K}(\theta,y)\right\} \le\sup_{r\in[0,D]}\phi_{k}(r)<\eta.
\]
Note that it always holds $w_{k}(\theta,y)\le w(\theta,y)$ for each
$k$. Thus we conclude the proof.\end{proof}

\begin{lem}\label{lem:events_equivalence} For each $n$ and each
closed $G\subset\mathcal{M}$, the following equivalence of the events
holds:
\[
\left\{ \mu_{n}^{B}(z)\in G\right\} =\left\{ \lim_{\varepsilon\downarrow0}\max_{\mu\in A_{0}\cap G}\mathcal{Q}_{n}(\mu,\varepsilon;z)=\lim_{\varepsilon\downarrow0}\max_{\mu\in A_{0}}\mathcal{Q}_{n}(\mu,\varepsilon;z)\right\} .
\]
\end{lem}

\begin{proof}($\subset$). Take any $z\in\mathcal{Z}^{n}$ such that
$\left\{ \mu_{n}^{B}(z)\in G\right\} $. Note that 
\[
\lim_{\varepsilon\downarrow0}\max_{\mu\in A_{0}\cap G}\mathcal{Q}_{n}(\mu,\varepsilon;z)\le\lim_{\varepsilon\downarrow0}\max_{\mu\in A_{0}}\mathcal{Q}_{n}(\mu,\varepsilon;z)
\]
 is clear. Suppose by way of contradiction that
\[
\lim_{\varepsilon\downarrow0}\max_{\mu\in A_{0}\cap G}\mathcal{Q}_{n}(\mu,\varepsilon;z)<\lim_{\varepsilon\downarrow0}\max_{\mu\in A_{0}}\mathcal{Q}_{n}(\mu,\varepsilon;z).
\]
Then, there is small enough $\varepsilon_{1}>0$ such that 
\[
\max_{\mu\in A_{0}\cap G}\mathcal{Q}_{n}(\mu,\varepsilon_{1};z)<\max_{\mu\in A_{0}}\mathcal{Q}_{n}(\mu,\varepsilon_{1};z).
\]
We can take $\eta>0$ such that 
\[
\max_{\mu\in A_{0}\cap G}\mathcal{Q}_{n}(\mu,\varepsilon_{1};z)+2\eta<\max_{\mu\in A_{0}}\mathcal{Q}_{n}(\mu,\varepsilon_{1};z)
\]
Since $\max_{\mu\in A_{0}}\mathcal{Q}_{n}(\mu,\varepsilon;z)\to\max_{\mu\in A_{0}}\mathcal{Q}_{n}(\mu,0;z)$
as $\varepsilon\downarrow0$ from Lemmas \ref{lem:nutz5.2} and \ref{lem:nutz5.4},
there exists small enough $\varepsilon_{2}>0$ such that
\[
\left|\max_{\mu\in A_{0}}\mathcal{Q}_{n}(\mu,\varepsilon_{2};z)-\max_{\mu\in A_{0}}\mathcal{Q}_{n}(\mu,0;z)\right|<\eta.
\]
Also, because $\varepsilon H(\mu_{n}^{B})\to0$ as $\varepsilon\downarrow0$,
there exists small enough $\varepsilon_{3}>0$ such that $\varepsilon_{3}H(\mu_{n}^{B})<\eta$.
Let $\varepsilon=\min\left\{ \varepsilon_{1},\varepsilon_{2},\varepsilon_{3}\right\} $.
Then, 
\begin{align*}
\max_{\mu\in A_{0}}\mathcal{Q}_{n}(\mu,\varepsilon;z) & <\max_{\mu\in A_{0}}\mathcal{Q}_{n}(\mu,0;z)+\eta\\
 & =\int\int\sqrt{n}\left(w(\theta_{nh},x,t)-w(\theta_{0},x,t)\right)\mathrm{d}\mu_{n}^{B}(z)\mathrm{d}\pi(\theta_{nh}|z)+\eta\\
 & \le\int\int\sqrt{n}\left(w(\theta_{nh},x,t)-w(\theta_{0},x,t)\right)\mathrm{d}\mu_{n}^{B}(z)\mathrm{d}\pi(\theta_{nh}|z)-\varepsilon H(\mu_{n}^{B})+\varepsilon H(\mu_{n}^{B})+\eta\\
 & \le\max_{\mu\in A_{0}\cap G}\mathcal{Q}_{n}(\mu,\varepsilon;z)+\varepsilon H(\mu_{n}^{B})+\eta\\
 & <\max_{\mu\in A_{0}\cap G}\mathcal{Q}_{n}(\mu,\varepsilon;z)+2\eta<\max_{\mu\in A_{0}}\mathcal{Q}_{n}(\mu,\varepsilon;z),
\end{align*}
which is a contradiction. Hence it holds $\lim_{\varepsilon\downarrow0}\max_{\mu\in A_{0}\cap G}\mathcal{Q}_{n}(\mu,\varepsilon;z)=\lim_{\varepsilon\downarrow0}\max_{\mu\in A_{0}}\mathcal{Q}_{n}(\mu,\varepsilon;z)$.

($\supset$). For the other direction, take any $z\in\mathcal{Z}^{n}$
with 
\[
\lim_{\varepsilon\downarrow0}\max_{\mu\in A_{0}\cap G}\mathcal{Q}_{n}(\mu,\varepsilon;z)=\lim_{\varepsilon\downarrow0}\max_{\mu\in A_{0}}\mathcal{Q}_{n}(\mu,\varepsilon;z).
\]
Let $\varepsilon_{k}\downarrow0$ as $k\to\infty$. Recall
\[
\mu_{n,\varepsilon_{k}}^{B}(z)=\arg\max_{\mu\in A_{0}}\mathcal{Q}_{n}(\mu,\varepsilon_{k};z).
\]
By Proposition \ref{prop:nutz}, $\mu_{n,\varepsilon_{k}}^{B}(z)$
converges to $\mu_{n}^{B}(z)$ weakly as $k\to\infty$. Moreover,
by \citet[Theorem 6.9]{villani2009optimal}, we know that weak convergence
in $\mathcal{M}$ is equivalent to convergence in $(\mathcal{M},d_{W})$.
Hence, $\mu_{n,\varepsilon_{k}}^{B}(z)\to\mu_{n}^{B}(z)$ in the Wasserstein
distance $d_{W}$.

First, we argue that for each $K\in\mathbb{N}$ there is $k\ge K$
such that $\mu_{n,\varepsilon_{k}}^{B}(z)\in G$. By way of contradiction,
assume that there is $K$ such that for any $k\ge K$, $\mu_{n,\varepsilon_{k}}^{B}(z)\notin G$.
Now, $\mu_{n,\varepsilon_{k}}^{B}(z)\notin G$ implies that
\[
\max_{\mu\in A_{0}\cap G}\mathcal{Q}_{n}(\mu,\varepsilon_{k};z)<\mathcal{Q}_{n}(\mu_{n,\varepsilon_{k}}^{B},\varepsilon_{k};z).
\]
Take any $\eta>0$ such that
\[
\max_{\mu\in A_{0}\cap G}\mathcal{Q}_{n}(\mu,\varepsilon_{k};z)+\eta<\mathcal{Q}_{n}(\mu_{n,\varepsilon_{k}}^{B},\varepsilon_{k};z).
\]
Because 
\[
\lim_{k\to\infty}\left(\max_{\mu\in A_{0}}\mathcal{Q}_{n}(\mu,\varepsilon_{k};z)-\max_{\mu\in(A_{0}\cap G)}\mathcal{Q}_{n}(\mu,\varepsilon_{k};z)\right)=0,
\]
there exists $K^{\prime}$ such that if $k\ge K^{\prime}$ then 
\[
\left|\max_{\mu\in A_{0}}\mathcal{Q}_{n}(\mu,\varepsilon_{k};z)-\max_{\mu\in(A_{0}\cap G)}\mathcal{Q}_{n}(\mu,\varepsilon_{k};z)\right|=\max_{\mu\in A_{0}}\mathcal{Q}_{n}(\mu,\varepsilon_{k};z)-\max_{\mu\in(A_{0}\cap G)}\mathcal{Q}_{n}(\mu,\varepsilon_{k};z)<\eta/2.
\]
Therefore, for sufficiently large $k$, we have
\[
\max_{\mu\in A_{0}\cap G}\mathcal{Q}_{n}(\mu,\varepsilon_{k};z)+\eta<\mathcal{Q}_{n}(\mu_{n,\varepsilon_{k}}^{B},\varepsilon_{k};z)\le\max_{\mu\in A_{0}}\mathcal{Q}_{n}(\mu,\varepsilon_{k};z)<\max_{\mu\in(A_{0}\cap G)}\mathcal{Q}_{n}(\mu,\varepsilon_{k};z)-\eta/2,
\]
which leads to a contradiction. Therefore, for each $K$ there is
$k\ge K$ such that $\mu_{n,\varepsilon_{k}}^{B}(z)\in G$.

Now, create such a subsequence $\left\{ \mu_{n,\varepsilon_{k_{\ell}}}^{B}(z)\right\} _{\ell\in\mathbb{N}}$
with $\mu_{n,\varepsilon_{k_{\ell}}}^{B}(z)\in G$ for each $\ell$.
Note that any subsequence of convergent sequence in arbitrary metric
space converges to the same limit as the original sequence. Therefore,
$\mu_{n,\varepsilon_{k_{\ell}}}^{B}(z)\to\mu_{n}^{B}(z)$ in $d_{W}$
as $\ell\to\infty$. Since $G$ is closed, we conclude that $\mu_{n}^{B}(z)\in G$.
\end{proof}

Let
\begin{align*}
Q_{n}(\mu,\varepsilon;z) & =\int\left(\int\dot{w}_{\theta_{0}}(x,t;h)\mathrm{d}\mu\right)\pi_{n}(\theta_{nh}|z)\mathrm{d}h-\varepsilon H(\mu),\\
\tilde{Q}_{n}(\mu,\varepsilon;z) & =\int\left(\int\dot{w}_{\theta_{0}}(x,t;h)\mathrm{d}\mu\right)\mathrm{d}N(\Delta_{n}(z),I_{0}^{-1})(h)-\varepsilon H(\mu),
\end{align*}
where $\Delta_{n}(z)=\frac{1}{\sqrt{n}}\sum_{i=1}^{n}I_{0}^{-1}s(z_{i})\in\mathbb{R}^{k}$
with the score function $s$ at $\theta_{0}$ and $\Delta_{n}\overset{0}{\rightsquigarrow}G\sim N(0,I_{0}^{-1})$.
Define $\mathbb{D}\equiv\mathcal{M}\times[0,1]$. Let $\{Q_{n}(\mu,\varepsilon):(\mu,\varepsilon)\in\mathbb{D}\}$,
$\{\tilde{Q}_{n}(\mu,\varepsilon):(\mu,\varepsilon)\in\mathbb{D}\}$,
$\{Q_{\infty}(\mu,\varepsilon):(\mu,\varepsilon)\in\mathbb{D}\}$
be stochastic processes. Assume that they yield maps $Q_{n}:\mathcal{Z}^{n}\to\ell^{\infty}(\mathbb{D})$,
$\tilde{Q}_{n}:\mathcal{Z}^{n}\to\ell^{\infty}(\mathbb{D})$, and
$Q_{\infty}:\mathbb{R}^{k}\to\ell^{\infty}(\mathbb{D})$. We can do
this since the sample paths are continuous by Lemma \ref{lem:sample-path}.

Let
\[
\mathcal{F}=\left\{ f\in\ell^{\infty}(\mathbb{D}):\lim_{\epsilon\downarrow0}\max_{\mu\in A_{0}}f(\mu,\varepsilon)\text{ exists}\right\} ,
\]
where the sup-norm is equipped to $\mathcal{F}$. Note that $Q_{n}(z),\tilde{Q}_{n}(z),Q_{\infty}(\Delta)\in\mathcal{F}$
for any $z$ and $\Delta$. 

\begin{lem}$\mathcal{Q}_{n}\overset{h}{\rightsquigarrow}Q_{\infty}$
in $\mathcal{F}$ as $n\to\infty$.\label{lem:process_conv} \end{lem}

\begin{proof} Note that $\mathcal{Q}_{n}=Q_{n}+o_{P_{\theta_{nh}}^{n}}(1)$
as a process in $\mathcal{F}$ as $n\to\infty$ by Lemma \ref{lem:epb_Qn}.
Hence, it is sufficient to show that $Q_{n}\overset{h}{\rightsquigarrow}Q_{\infty}$
in $\mathcal{F}$ as $n\to\infty$. Let $C_{M}\subset\mathbb{R}^{k}$
be the closed ball of radius $M$ around $0$. Define stochastic processes
by
\begin{align*}
Q_{n,M}(\mu,\varepsilon;z) & =\int_{C_{M}}\left(\int\dot{w}_{\theta_{0}}(x,t;h)\mathrm{d}\mu(x,t)\right)\pi_{n}(\theta_{nh}|z)\mathrm{d}h-\varepsilon H(\mu),\\
\tilde{Q}_{n,M}(\mu,\varepsilon;z) & =\int_{C_{M}}\left(\int\dot{w}_{\theta_{0}}(x,t;h)\mathrm{d}\mu(x,t)\right)\mathrm{d}N(\Delta_{n}(z),I_{0}^{-1})(h)-\varepsilon H(\mu),\\
Q_{\infty,M}(\mu,\varepsilon;\Delta) & =\int_{C_{M}}\left(\int\dot{w}_{\theta_{0}}(x,t;h)\mathrm{d}\mu(x,t)\right)\mathrm{d}N(\Delta,I_{0}^{-1})(h)-\varepsilon H(\mu).
\end{align*}
First, we will show that $Q_{n,M}-\tilde{Q}_{n,M}\overset{h}{\rightarrow}0$
in $\mathcal{F}$ as $n\to\infty$ for any fixed $M$. Note that 
\begin{align*}
 & \left\lVert Q_{n,M}(\cdot;Z^{n})-\tilde{Q}_{n,M}(\cdot,Z^{n})\right\rVert _{\mathcal{F}}\\
 & \le\max_{h\in C_{M}}\max_{(x,t)\in\mathcal{X}\times\mathcal{T}}\left|\dot{w}_{\theta_{0}}(x,t;h)\right|\cdot\left\lVert \pi_{n}(\theta_{nh}|Z^{n})-N(\Delta_{n}(Z^{n}),I_{0}^{-1})\right\rVert _{\mathrm{TV}}
\end{align*}
where $\left\lVert \cdot\right\rVert _{\mathrm{TV}}$ is the total
variation norm. Because $C_{M}$ is bounded and $\dot{w}_{\theta_{0}}(x,t;h)$
is bounded for any $h$, we know that $\max_{h\in C_{M}}\max_{(x,t)\in\mathcal{X}\times\mathcal{T}}\left|\dot{w}_{\theta_{0}}(x,t;h)\right|<\infty$.
From the Bernstein--von Mises theorem, 
\[
\left\lVert \pi_{n}(\theta_{nh}|Z^{n})-N(\Delta_{n}(Z^{n}),I_{0}^{-1})\right\rVert _{\mathrm{TV}}\overset{0}{\rightarrow}0\quad\text{as }n\to\infty.
\]
From Le Cam's first lemma, it also converges to 0 along $P_{\theta_{nh}}^{n}$.
Therefore, $Q_{n,M}-\tilde{Q}_{n,M}\overset{h}{\rightarrow}0$ in
$\mathcal{F}$ as $n\to\infty$.

Next, we argue that $\tilde{Q}_{n,M}\overset{h}{\rightsquigarrow}Q_{M}$
in $\mathcal{F}$ as $n\to\infty$ for any fixed $M$. Define $\phi:\mathbb{R}^{k}\to\mathcal{F}$
by
\begin{align*}
\phi(\delta)(\mu,\varepsilon) & =\int_{C_{M}}\left(\int\dot{w}_{\theta_{0}}(x,t;h)\mathrm{d}\mu(x,t)\right)\mathrm{d}N(\delta,I_{0}^{-1})(h)-\varepsilon H(\mu)
\end{align*}
Since $\phi$ is continuous, and $\Delta_{n}\stackrel{h}{\rightsquigarrow}\Delta\sim N(h,I_{0}^{-1})$
as $n\to\infty$ by Le Cam's third lemma, the continuous mapping
theorem implies $\phi(\Delta_{n})\stackrel{h}{\rightsquigarrow}\phi(\Delta)$
as $n\to\infty$. Thus $\tilde{Q}_{n,M}\overset{h}{\rightsquigarrow}Q_{M}$
in $\mathcal{F}$ as $n\to\infty$.

Combining the above two findings, we obtain $Q_{n,M}\overset{h}{\rightsquigarrow}Q_{\infty,M}$
in $\mathcal{F}$ as $n\to\infty$ from the Slutsky theorem. We also
have that $Q_{\infty,M}-Q_{\infty}=o_{P_{h}^{\Delta}}(1)$ as $M\to\infty$
where $P_{h}^{\Delta}$ is the (marginal) law of $\Delta\sim N(h,I_{0}^{-1})$.
 Thus, there exists a sequence $M_{n}\to\infty$ such that $Q_{n,M_{n}}\overset{h}{\rightsquigarrow}Q_{\infty}$
in $\mathcal{F}$ as $n\to\infty$.

Finally, it remains to show that $Q_{n}-Q_{n,M_{n}}=o_{P_{\theta_{nh}}^{n}}(1)$
in $\mathcal{F}$ as $n\to\infty$, which leads to the conclusion,
$Q_{n}\overset{h}{\rightsquigarrow}Q_{\infty}$ in $\mathcal{F}$
as $n\to\infty$.  By Assumption \ref{assu:w-dot} (iii), 
\begin{align*}
\left|\int_{\mathbb{R}^{k}\setminus C_{M_{n}}}\left(\int\dot{w}_{\theta_{0}}(x,t;h)\mathrm{d}\mu(x,t)\right)\pi_{n}(\theta_{nh}|z)\mathrm{d}h\right| & \le\int_{\mathbb{R}^{k}\setminus C_{M_{n}}}K(h)\pi_{n}(\theta_{nh}|z)\mathrm{d}h.
\end{align*}
Then applying Proposition \ref{prop:subpoly-vdv} yields that RHS
is $o_{P_{\theta_{0}}^{n}}(1)$. Thus it is $o_{P_{\theta_{nh}}^{n}}(1)$
as well. Hence it follows that $Q_{n}(\mu,\varepsilon;Z^{n})-Q_{n,M_{n}}(\mu,\varepsilon;Z^{n})\stackrel{h}{\to}0$
for any $\left(\mu,\varepsilon\right)\in\mathbb{D}$. Then the continuity
of sample path implies $Q_{n}-Q_{n,M_{n}}=o_{P_{\theta_{nh}}^{n}}(1)$
in $\mathcal{F}$ as desired. \end{proof}

\begin{lem}\label{lem:epb_Qn}$\mathcal{Q}_{n}=Q_{n}+o_{P_{\theta_{nh}}^{n}}(1)$
in $\mathcal{F}$ as $n\to\infty$. \end{lem}

\begin{proof} Notice that for every $z\in\mathcal{Z}^{n}$, $\sup_{(\mu,\varepsilon)\in\mathbb{D}}\left|\mathcal{Q}_{n}(\mu,\varepsilon;z)-Q_{n}(\mu,\varepsilon;z)\right|$
is bounded above by
\[
\int\max_{(x,t)\in\mathcal{X}\times\mathcal{T}}\left|\tilde{w}_{n,0}(x,t;h)\right|\pi_{n}(\theta_{nh}|z)\mathrm{d}h,
\]
where $\tilde{w}_{n,0}(x,t;h):=\sqrt{n}\left(w(\theta_{nh},x,t)-w(\theta_{0},x,t)\right)-\dot{w}_{\theta_{0}}(x,t;h)$.
Let $C_{M_{n}}$ be a closed ball of radius $M_{n}$ around 0, where
$M_{n}$ is the divergent sequence specified in Lemma \ref{lem:hdd_growing}.
Then the previous display is further bounded by
\begin{equation}
\pi_{n}(\theta_{nh}|z)(C_{M_{n}})\times\max_{h\in C_{M_{n}}}\max_{(x,t)\in\mathcal{X}\times\mathcal{T}}\left|\tilde{w}_{n,0}(x,t;h)\right|+\int\boldsymbol{1}_{\mathbb{R}^{k}\setminus C_{M_{n}}}\left(h\right)\max_{(x,t)\in\mathcal{X}\times\mathcal{T}}\left|\tilde{w}_{n,0}(x,t;h)\right|\pi_{n}(\theta_{nh}|z)\mathrm{d}h\label{eq:lemA.5_decom}
\end{equation}
The first term of (\ref{eq:lemA.5_decom}) converges to zero as $n\to\infty$
by Lemma \ref{lem:hdd_growing}. For the second term of (\ref{eq:lemA.5_decom}),
note that from Assumption \ref{assu:w-dot} (i), 
\begin{align*}
\max_{(x,t)\in\mathcal{X}\times\mathcal{T}}\left|\sqrt{n}\left(w(\theta_{nh},x,t)-w(\theta_{0},x,t)\right)\right|\le\max_{(x,t)\in\mathcal{X}\times\mathcal{T}}\left|\dot{w}_{\theta_{0}}(x,t;h)\right|+o(1),
\end{align*}
From Assumption \ref{assu:w-dot} (iii), $\max_{(x,t)\in\mathcal{X}\times\mathcal{T}}\left|\dot{w}_{\theta_{0}}(x,t;h)\right|$
is bounded by $K(h)$ that grows at subpolynomially of order $p$.
This implies that $\max_{(x,t)\in\mathcal{X}\times\mathcal{T}}\left|\tilde{w}_{n,0}(x,t;h)\right|$
is also dominated by a function that grows subpolynomially of order
$p$ for sufficiently large $n$. Then applying Proposition \ref{prop:subpoly-vdv}
yields the conclusion. \end{proof}

\section{\label{sec:auxiliary-lemmas}Auxiliary lemmas for Theorem \ref{thm:optimality}}

The following extends the result on Hadamard directional differentiability
given by \citet[Proposition 1]{romisch2004delta}. Compared to his
setting, the objective map $\int w(\theta,x,t)\mathrm{d}\mu$ need
not to be linear in $\theta$. By leveraging the uniform continuity
from Assumptions \ref{assu:w} and \ref{assu:w-dot}, we obtain the
similar form of the directional derivative as his result.

\begin{lem}\label{lem:HDD} For a closed set $S\subset\mathcal{M}$,
define the map $W_{S}^{*}:\Theta\to\mathbb{R}$ by 
\[
W_{S}^{*}(\theta)=\max_{\mu\in S}\int w(\theta,x,t)\mathrm{d}\mu.
\]
 Then $W_{S}^{*}$ is Hadamard directionally differentiable with derivative
\[
\dot{W}_{S,0}^{*}[h]\equiv\lim_{\varepsilon\downarrow0}\sup_{\mu\in S^{\varepsilon}(\theta_{0})}\int\dot{w}_{\theta_{0}}(x,t;h)\mathrm{d}\mu
\]
where 
\[
S^{\varepsilon}(\theta)\equiv\left\{ \mu\in S:\int w(\theta,x,t)\mathrm{d}\mu+\varepsilon\ge\max_{\mu\in S}\int w(\theta,x,t)\mathrm{d}\mu\right\} \ne\emptyset
\]
for $\varepsilon>0$ and $\theta\in\Theta$. Moreover, if $S\subset A_{0}$
then $\dot{W}_{S,0}^{*}[h]=\max_{\mu\in S}\int\dot{w}_{\theta_{0}}(x,t;h)\mathrm{d}\mu$.\end{lem}

\begin{proof}The second statement follows from the first statement
and the fact that $S^{\varepsilon}(\theta_{0})=S$ for any $\varepsilon>0$.
Hereafter, we focus on the first statement.

Fix any closed $S\subset\mathcal{M}$ and any $\theta_{0}\in\Theta$.
Then, $\arg\max_{\mu\in S}W(\theta_{0},\mu)\ne\emptyset$ because
$w(\theta_{0},\cdot)$ is bounded continuous on $\mathcal{X}\times\mathcal{T}$
from Assumption \ref{assu:w} \citep[Theorem 4.1]{villani2009optimal}.
Let $\mu(\theta_{0})\in\arg\max_{\mu\in S}W(\theta_{0},\mu)$. Since
$\arg\max_{\mu\in S}W(\theta_{0},\mu)\subset S^{\varepsilon}(\theta_{0})$
for any $\varepsilon>0$, we can guarantee $S^{\varepsilon}(\theta_{0})\ne\emptyset$.
Also, because $w(\theta_{0},\cdot,\cdot)$ is bounded, $\max_{\mu\in S}\int w(\theta_{0},x,t)\mathrm{d}\mu<\infty$.

Let $r_{n}\downarrow0$, and $h_{n}\to h$. Define 
\[
\sigma_{n}=\frac{1}{r_{n}}\left(W_{S}^{*}(\theta_{0}+r_{n}h_{n})-W_{S}^{*}(\theta_{0})\right)
\]
and we want to show that $\sigma_{n}\to\dot{W}_{S,0}^{*}[h]$. But,
it is enough to show that for any subsequence of $\{\sigma_{n}\}$,
there exists a further subsequence that converges to $\dot{W}_{S,0}^{*}[h]$.
Take any subsequence and denote it by $\{\sigma_{n}\}$ for simplicity.

Define the maps $T_{n}:S\to\mathbb{R}$ and $T:S\to\mathbb{R}$ by
\begin{align*}
T_{n}(\mu) & =\int\frac{1}{r_{n}}\left(w(\theta_{0}+r_{n}h_{n},x,t)-w(\theta_{0},x,t)\right)\mathrm{d}\mu,\\
T(\mu) & =\int\dot{w}_{\theta_{0}}(x,t;h)\mathrm{d}\mu.
\end{align*}
First, to see that $T_{n}\to T$ uniformly, take any $\mu\in S$.
Note that 
\begin{align*}
\left|T_{n}(\mu)-T(\mu)\right| & \le\max_{(x,t)\in\mathcal{X}\times\mathcal{T}}\left|\frac{1}{r_{n}}\left(w(\theta_{0}+r_{n}h_{n},x,t)-w(\theta_{0},x,t)\right)-\dot{w}_{\theta_{0}}(x,t;h)\right|.
\end{align*}
We can make RHS arbitrary small as $n\to\infty$ without depending
on $\mu$ from Assumption \ref{assu:w-dot} (i). Hence, we conclude
that $T_{n}\to T$ uniformly.

Thus, for each $n$ there exists $K_{1}$ such that for each $k\ge K_{1}$
\begin{equation}
\left|\int\frac{1}{r_{n_{k}}}\left(w(\theta_{0}+r_{n_{k}}h_{n_{k}},x,t)-w(\theta_{0},x,t)\mathrm{d}\mu\right)-\int\dot{w}_{\theta_{0}}(x,t;h)\mathrm{d}\mu\right|<r_{n}^{2}\quad\forall\mu\in S.\label{eq:unif1}
\end{equation}
Moreover, from Assumption \ref{assu:w} (i), there exists $K_{2}$
such that for each $k\ge K_{2}$ 
\begin{equation}
\left|\int w(\theta_{0}+r_{n_{k}}h_{n_{k}},x,t)\mathrm{d}\mu-\int w(\theta_{0},x,t)\mathrm{d}\mu\right|<r_{n}^{2}/2\quad\forall\mu\in S.\label{eq:unif2}
\end{equation}
Let $k\ge\max\{K_{1},K_{2}\}$, and construct a further subsequence
$\{\sigma_{n_{k}}\}.$

For each $k$, take any $\mu_{n_{k}}\in S^{r_{n}^{2}r_{n_{k}}}(\theta_{0})$.
Then, from $\mu_{n_{k}}\in S$ and the definition of $S^{r_{n}^{2}r_{n_{k}}}(\theta_{0})$,
\begin{align*}
\frac{1}{r_{n_{k}}}\left(W_{S}^{*}(\theta_{0}+r_{n_{k}}h_{n_{k}})-W_{S}^{*}(\theta_{0})\right) & =\frac{1}{r_{n_{k}}}\left(\max_{\mu\in S}\int w(\theta_{0}+r_{n_{k}}h_{n_{k}},x,t)\mathrm{d}\mu-\max_{\mu\in S}\int w(\theta_{0},x,t)\mathrm{d}\mu\right)\\
 & \ge\frac{1}{r_{n_{k}}}\left(\int w(\theta_{0}+r_{n_{k}}h_{n_{k}},x,t)\mathrm{d}\mu_{n_{k}}-\max_{\mu\in S}\int w(\theta_{0},x,t)\mathrm{d}\mu\right)\\
 & \ge\frac{1}{r_{n_{k}}}\left(\int w(\theta_{0}+r_{n_{k}}h_{n_{k}},x,t)\mathrm{d}\mu_{n_{k}}-\int w(\theta_{0},x,t)\mathrm{d}\mu_{n_{k}}\right)-r_{n}^{2}.
\end{align*}
From \eqref{eq:unif1}, we have 
\[
\sigma_{n_{k}}>\int\dot{w}_{0}(x,t;h)\mathrm{d}\mu_{n_{k}}-2r_{n}^{2}.
\]
 Therefore, 
\[
\sigma_{n_{k}}\ge\sup_{\mu\in S^{r_{n}^{2}r_{n_{k}}}(\theta_{0})}\int\dot{w}_{0}(x,t;h)\mathrm{d}\mu-2r_{n}^{2}.
\]
Since $r_{n}^{2}\downarrow0$, it leads to
\[
\liminf_{k\to\infty}\sigma_{n_{k}}\ge\lim_{\varepsilon\downarrow0}\sup_{\mu\in S^{\varepsilon}(\theta_{0})}\int\dot{w}_{0}(x,t;h)\mathrm{d}\mu.
\]
Also, take any $\mu_{n_{k}}^{\prime}\in S^{r_{n}^{2}r_{n_{k}}}(\theta_{0}+r_{n_{k}}h_{n_{k}})$.
Then, 
\begin{align*}
\frac{1}{r_{n_{k}}}\left(W_{S}^{*}(\theta_{0}+r_{n_{k}}h_{n_{k}})-W_{S}^{*}(\theta_{0})\right) & =\frac{1}{r_{n_{k}}}\left(\max_{\mu\in S}\int w(\theta_{0}+r_{n_{k}}h_{n_{k}},x,t)\mathrm{d}\mu-\max_{\mu\in S}\int w(\theta_{0},x,t)\mathrm{d}\mu\right)\\
 & \le\frac{1}{r_{n_{k}}}\left(\int w(\theta_{0}+r_{n_{k}}h_{n_{k}},x,t)\mathrm{d}\mu_{n_{k}}-\int w(\theta_{0},x,t)\mathrm{d}\mu_{n_{k}}\right)+r_{n}^{2}.
\end{align*}
From \eqref{eq:unif1}, 
\[
\sigma_{n_{k}}<\int\dot{w}_{\theta_{0}}(x,t;h)\mathrm{d}\mu_{n_{k}}^{\prime}+2r_{n}^{2}.
\]
If we would have $S^{r_{n}^{2}r_{n_{k}}}(\theta_{0}+r_{n_{k}}h_{n_{k}})\subset S^{r_{n}^{2}r_{n_{k}}+r_{n}^{2}}(\theta_{0})$,
then we obtain 
\[
\sigma_{n_{k}}\le\sup\left\{ \int\dot{w}_{\theta_{0}}(x,t;h)\mathrm{d}\mu:\mu\in S^{r_{n}^{2}r_{n_{k}}+r_{n}^{2}}(\theta_{0})\right\} +2r_{n}^{2},
\]
which leads to 
\[
\limsup_{k\to\infty}\sigma_{n_{k}}\le\lim_{\varepsilon\downarrow0}\sup_{\mu\in S^{\varepsilon}(\theta_{0})}\int\dot{w}_{\theta_{0}}(x,t;h)\mathrm{d}\mu,
\]
thus $\sigma_{n_{k}}\to\dot{W}_{S,0}^{*}[h].$ Hence, it suffices
to show $S^{r_{n}^{2}r_{n_{k}}}(\theta_{0}+r_{n_{k}}h_{n_{k}})\subset S^{r_{n}^{2}r_{n_{k}}+r_{n}^{2}}(\theta_{0})$
for the conclusion. Take any $\nu\in S^{r_{n}^{2}r_{n_{k}}}(\theta_{0}+r_{n_{k}}h_{n_{k}})$,
then 
\begin{equation}
\int w(\theta_{0}+r_{n_{k}}h_{n_{k}})\mathrm{d}\nu+r_{n}^{2}r_{n_{k}}\ge\max_{\mu\in S}\int w(\theta_{0}+r_{n_{k}}h_{n_{k}})\mathrm{d}\mu.\label{eq:condition-nu}
\end{equation}
Let $\mu(\theta_{0})\in\arg\max_{\mu\in S}\int w(\theta_{0},x,t)\mathrm{d}\mu$,
then 
\begin{align*}
 & \int w(\theta_{0},x,t)\mathrm{d}\nu+r_{n}^{2}r_{n_{k}}+r_{n}^{2}\\
\ge & \int w(\theta_{0},x,t)\mathrm{d}\nu+r_{n}^{2}+\max_{\mu\in S}\int w(\theta_{0}+r_{n_{k}}h_{n_{k}},x,t)\mathrm{d}\mu-\int w(\theta_{0}+r_{n_{k}}h_{n_{k}},x,t)\mathrm{d}\nu\quad\because\eqref{eq:condition-nu}\\
\ge & \int w(\theta_{0},x,t)\mathrm{d}\nu+r_{n}^{2}+\int w(\theta_{0}+r_{n_{k}}h_{n_{k}},x,t)\mathrm{d}\mu(\theta_{0})-\int w(\theta_{0}+r_{n_{k}}h_{n_{k}},x,t)\mathrm{d}\nu\\
= & \int w(\theta_{0},x,t)\mathrm{d}\nu+r_{n}^{2}\\
 & +\int w(\theta_{0}+r_{n_{k}}h_{n_{k}},x,t)\mathrm{d}\mu(\theta_{0})-\int w(\theta_{0},x,t)\mathrm{d}\mu(\theta_{0})+\int w(\theta_{0},x,t)\mathrm{d}\mu(\theta_{0})-\int w(\theta_{0}+r_{n_{k}}h_{n_{k}},x,t)\mathrm{d}\nu\\
= & \int w(\theta_{0},x,t)\mathrm{d}\mu(\theta_{0})+r_{n}^{2}\\
 & +\int\left(w(\theta_{0},x,t)-w(\theta_{0}+r_{n_{k}}h_{n_{k}},x,t)\right)\mathrm{d}\nu+\int\left(w(\theta_{0}+r_{n_{k}}h_{n_{k}},x,t)-w(\theta_{0},x,t)\right)\mathrm{d}\mu(\theta_{0})\\
\ge & \int w(\theta_{0},x,t)\mathrm{d}\mu(\theta_{0})+r_{n}^{2}-r_{n}^{2}/2-r_{n}^{2}/2\quad\because\eqref{eq:unif2}\\
= & \int w(\theta_{0},x,t)\mathrm{d}\mu(\theta_{0}).
\end{align*}
Thus, $\nu\in S^{r_{n}^{2}r_{n_{k}}+r_{n}^{2}}(\theta_{0})$. \end{proof}

\begin{lem}\label{lem:sample-path} The sample paths of $Q_{n}$,
$\tilde{Q}_{n}$, and $Q_{\infty}$ are continuous and bounded in
$\mathcal{M}$.\end{lem}

\begin{proof}First, we will show that the sample path of 
\[
Q_{n}(\mu,\varepsilon;z)=\int\left(\int\dot{w}_{\theta_{0}}(x,t;h)\mathrm{d}\mu\right)\pi_{n}(\theta_{nh}|z)\mathrm{d}h-\varepsilon H(\mu)
\]
is continuous. Note that, then, it is bounded because $\mathcal{M}$
is compact. Fix any $z$, and take any $\left\{ (\mu_{k},\varepsilon_{k})\right\} _{k=1}^{\infty}\subset\mathbb{D}$
that converges to $(\mu,\varepsilon)$. Since $\mathbb{D}$ is a metric
space, overall convergence implies elementwise convergence. Thus,
$\mu_{k}\to\mu$ in the Wasserstein distance and $\varepsilon_{k}\to\varepsilon$.
We are done if 
\[
\left|\int\left\{ \int\dot{w}_{\theta_{0}}(x,t;h)\mathrm{d}\mu_{k}-\int\dot{w}_{\theta_{0}}(x,t;h)\mathrm{d}\mu\right\} \pi_{n}(\theta_{nh}|z)\mathrm{d}h-\varepsilon_{k}H(\mu_{k})+\varepsilon H(\mu)\right|\to0,
\]
as $k\to\infty$. 

By the triangle inequality, LHS is bounded above by
\begin{align*}
\left|\int\left\{ \int\dot{w}_{\theta_{0}}(x,t;h)\mathrm{d}\mu_{k}-\int\dot{w}_{\theta_{0}}(x,t;h)\mathrm{d}\mu\right\} \pi_{n}(\theta_{nh}|z)\mathrm{d}h\right|+\left|\varepsilon_{k}H(\mu_{k})-\varepsilon H(\mu)\right|.
\end{align*}
The second term converges to zero since $H$ is continuous. For the
first term, since $\mu_{k}\to\mu$ in the Wasserstein distance implies
$\mu_{k}\rightsquigarrow\mu$, we have 
\begin{align*}
\int\left(\int\dot{w}_{\theta_{0}}(x,t;h)\mathrm{d}\mu_{k}\right)\pi_{n}(\theta_{nh}|z)\mathrm{d}h & =\int\left(\int\dot{w}_{\theta_{0}}(x,t;h)\pi_{n}(\theta_{nh}|z)\mathrm{d}h\right)\mathrm{d}\mu_{k}\\
 & \to\int\left(\int\dot{w}_{\theta_{0}}(x,t;h)\pi_{n}(\theta_{nh}|z)\mathrm{d}h\right)\mathrm{d}\mu,
\end{align*}
where the convergence holds if the map $(x,t)\mapsto\int\dot{w}_{\theta_{0}}(x,t;h)\pi_{n}(\theta_{nh}|z)\mathrm{d}h$
is continuous and bounded. To see the continuity, let $\left\{ (x_{k},t_{k})\right\} $
be such that $(x_{k},t_{k})\to(x,t)$. Then 
\begin{align*}
 & \left|\int\dot{w}_{\theta_{0}}(x,t;h)\pi_{n}(\theta_{nh}|z)\mathrm{d}h-\int\dot{w}_{\theta_{0}}(x,t;h)\pi_{n}(\theta_{nh}|z)\mathrm{d}h\right|\\
\le & \int\left|\dot{w}_{\theta_{0}}(x,t;h)-\dot{w}_{\theta_{0}}(x,t;h)\right|\pi_{n}(\theta_{nh}|z)\mathrm{d}h\to0\quad\text{as }k\to\infty,
\end{align*}
where the convergence follows from the continuity of $\dot{w}_{\theta_{0}}(x,t;h)$
in $(x,t)$ and the dominated convergence theorem. The boundedness
of $(x,t)\mapsto\int\dot{w}_{\theta_{0}}(x,t;h)\pi_{n}(\theta_{nh}|z)\mathrm{d}h$
follows from the continuity of $\dot{w}_{\theta_{0}}(x,t;h)$ in $(x,t)$
(Assumption \ref{assu:w-dot} (ii)) and the compactness of $\mathcal{X}\times\mathcal{T}$
(Assumption \ref{assu:villani_thm6-9}).\textbf{ }

Similar arguments can be applied to $\tilde{Q}_{n}$ and $Q_{\infty}$.\end{proof}

\begin{lem}\label{lem:cont_operator} For any $S\subset A_{0}$.
the operator $M:\mathcal{F}\to\mathbb{R}$ where $M(f):=\lim_{\varepsilon\downarrow0}\max_{\mu\in S}f(\mu,\varepsilon)$
is continuous.\end{lem}

\begin{proof}Let $f_{k}\to f$ in $\mathcal{F}$ as $k\to\infty$.
Note that $M(f)$ and $M(f_{k})$ exist for each $k$ by the definition
of $\mathcal{F}$. Then 
\begin{align*}
\left|M(f)-M(f_{k})\right| & =\lim_{\varepsilon\downarrow0}\left|\max_{\mu\in\mathcal{M}}f(\mu,\varepsilon)-\max_{\mu\in\mathcal{M}}f_{k}(\mu,\varepsilon)\right|\\
 & \le\lim_{\varepsilon\downarrow0}\max_{\mu\in\mathcal{M}}\left|f(\mu,\varepsilon)-f_{k}(\mu,\varepsilon)\right|\\
 & \le\max_{\left(\mu,\varepsilon\right)\in\mathbb{D}}\left|f(\mu,\varepsilon)-f_{k}(\mu,\varepsilon)\right|\to0,
\end{align*}
as $k\to\infty$, where the equality follows because $\left|\cdot\right|$
is continuous, and the convergence follows because $f_{k}\to f$ in
$\mathcal{F}$ as $k\to\infty$. \end{proof}

The next result is an adaptation of \citet[Lemma 3]{christensen2023optimal},
where we need a modification to allow $\max$-operator within the
expression.

\begin{lem}\label{lem:hdd_growing} There is a sequence $\{M_{n}\}$
such that $M_{n}\uparrow\infty$, $M_{n}/\sqrt{n}\to0$, and 
\[
\sup_{\left\lVert h\right\rVert \le2M_{n}}\max_{(x,t)\in\mathcal{X}\times\mathcal{T}}\left|\sqrt{n}\left[w(\theta_{nh},x,t)-w(\theta_{0},x,t)\right]-\dot{w}_{\theta_{0}}(x,t;h)\right|\to0.
\]
\end{lem}

\begin{proof}From \citet[Lemmas 3.3 and 3.4]{shapiro1990concepts}
and Assumption \ref{assu:w-dot} (i), we know that for any compact
$S\subset\mathbb{R}^{k}$ and $\varepsilon>0$, there is $N$ such
that $\sup_{h\in S}g_{n}(h)<\epsilon$ for any $n\ge N$ where $g_{n}(h)=\max_{(x,t)\in\mathcal{X}\times\mathcal{T}}\left|\sqrt{n}\left[w(\theta_{nh},x,t)-w(\theta_{0},x,t)\right]-\dot{w}_{\theta_{0}}(x,t;h)\right|$.
Define 
\[
\psi_{n}=\sup_{\left\lVert h\right\rVert \le2\log(1+n)}g_{n}(h).
\]
We are done if $\psi_{n}\to0$ because $\log(1+n)\uparrow\infty$
and $n^{-1/2}\log(n+1)\to0$. To show $\psi_{n}\to0$, it is enough
to show that for any subsequence $\psi_{n}$ (abusing notation) there
is a further subsequence converging to $0$. First, consider $\sup_{\left\lVert h\right\rVert \le2\log(1+1)}g_{n}(h)$.
We know that there is $N(1)$ such that $\sup_{\left\lVert h\right\rVert \le2\log(1+1)}g_{n}(h)<1/\log(1+1)$
for any $n\ge N(1)$. Second, for $\sup_{\left\lVert h\right\rVert \le2\log(1+2)}g_{n}(h)$,
there is $N(2)$ such that $\sup_{\left\lVert h\right\rVert \le2\log(1+2)}g_{n}(h)<1/\log(1+2)$
for any $n\ge N(2)$. Proceed with $N(1)<N(2)<N(3)<\cdots$ WLOG.
Then, the map $N:\mathbb{N}\to\mathbb{N}$ satisfies $i<j\implies N(i)<N(j)$.
Hence, $\psi_{N(n)}$ is a subsequence of $\psi_{n}$. To show $\psi_{N(n)}=\sup_{\left\lVert h\right\rVert \le2\log(1+n)}g_{N(n)}(h)\to0$,
take any $\varepsilon>0$. Since $1/\log(1+n)\downarrow0$, there
is $\bar{n}$ such that $1/\log(1+n)<\epsilon$ for any $n\ge\bar{n}$.
Therefore, $\psi_{N(n)}<1/\log(1+n)<\epsilon$ for any $n\ge\bar{n}$.
\end{proof}

\section{\label{sec:penalty}Proof of Proposition \ref{prop:penalty}}

By \citet[Corollary 6.11]{villani2009optimal}, $\mu\mapsto d_{W}(\mu,\nu)$
is continuous on $\mathcal{M}$.\footnote{More explicitly, if $\mu_{k}$ converges to $\mu$ weakly in $\mathcal{M}$,
then $d_{W}(\mu_{k},\nu)\to d_{W}(\mu,\nu)$.} Convexity of $\mu\mapsto d_{W}(\mu,\nu)$ is shown below. Because
the map $\mu\mapsto d_{W}(\mu,\nu)$ is convex and nonnegative and
the map $\mathbb{R}_{+}\ni x\mapsto x^{2}$ is increasing and strictly
convex, the composite map $\mu\mapsto\left(d_{W}\left(\mu,\nu\right)\right)^{2}$
is strictly convex and nonnegative. Also, it is bounded by the compactness
of $\mathcal{M}$ and the continuity.

To see the convexity of $\mu\mapsto d_{W}(\mu,\nu)$, fix any $\mu_{1},\mu_{2}\in\mathcal{M}$
and $\alpha\in(0,1)$. To simplify the notation, let $\mathcal{Y}=\mathcal{X}\times\mathcal{T}$.
Let
\begin{align*}
\gamma_{1}^{*} & \in\arg\inf_{\gamma\in\Gamma(\mu_{1},\nu)}\int d(y,y^{\prime})\gamma(\mathrm{d}y,\mathrm{d}y^{\prime}),\\
\gamma_{2}^{*} & \in\arg\inf_{\gamma\in\Gamma(\mu_{2},\nu)}\int d(y,y^{\prime})\gamma(\mathrm{d}y,\mathrm{d}y^{\prime}).
\end{align*}
First, we show that $\alpha\gamma_{1}^{*}+(1-\alpha)\gamma_{2}^{*}\in\Gamma(\alpha\mu_{1}+(1-\alpha)\mu_{2},\nu)$.
Take any $A,B\in\mathcal{Y}$. Then, 
\[
\left(\alpha\gamma_{1}^{*}+(1-\alpha)\gamma_{2}^{*}\right)(\mathcal{Y}\times B)=\alpha\gamma_{1}^{*}(\mathcal{Y}\times B)+(1-\alpha)\gamma_{2}^{*}(\mathcal{Y}\times B)=\nu(B).
\]
Also,
\begin{align*}
 & \left(\alpha\gamma_{1}^{*}+(1-\alpha)\gamma_{2}^{*}\right)(A\times\mathcal{Y})\\
 & =\alpha\gamma_{1}^{*}(A\times\mathcal{Y})+(1-\alpha)\gamma_{2}^{*}(A\times\mathcal{Y})=\alpha\mu_{1}(A)+(1-\alpha)\mu_{2}(A)=\left(\alpha\mu_{1}+(1-\alpha)\mu_{2}\right)(A).
\end{align*}
Hence,
\begin{align*}
r(\alpha\mu_{1}+(1-\alpha)\mu_{2}) & \le\int d(y,y^{\prime})\mathrm{d}(\alpha\gamma_{1}^{*}+(1-\alpha)\gamma_{2}^{*})\\
 & =\alpha\int d(y,y^{\prime})\mathrm{d\gamma_{1}^{*}}+(1-\alpha)\int d(y,y^{\prime})\mathrm{d}\gamma_{2}^{*}=\alpha r(\mu_{1})+(1-\alpha)r(\mu_{2}).
\end{align*}

\section{\label{sec:Proof_Lemma_Nutz}Proof of Proposition \ref{prop:nutz}}

We provide a proof under a general framework using continuous and
bounded cost function $c:\mathcal{X}\times\mathcal{T}\to\mathbb{R}$.
Define 
\begin{align*}
\mathcal{C}_{\varepsilon} & :=\inf_{\mu\in\mathcal{M}}\int c\mathrm{d}\mu+\varepsilon H(\mu).\qquad(\varepsilon\text{EOT})\\
\mathcal{C}_{0} & :=\inf_{\mu\in\mathcal{M}}\int c\mathrm{d}\mu.\qquad(\text{OT}).
\end{align*}
Let $\mathcal{M}_{opt}=\arg\min_{\mu\in\mathcal{M}}\int c\mathrm{d}\mu$.
It should be noted that $(\mathcal{M},d_{W})$ is a metric space which
is convex and compact. 

Our proof of Proposition \ref{prop:nutz} proceeds as follows. Lemma
\ref{lem:nutz5.5} gives the conclusion under the assumption $\lim_{\varepsilon\to0}\mathcal{C}_{\varepsilon}=\mathcal{C}_{0}$.
Lemmas \ref{lem:nutz1.9}--\ref{lem:nutz1.17} are needed to prove
Lemma \ref{lem:nutz5.5}. Finally, Lemmas \ref{lem:nutz5.2} and \ref{lem:nutz5.4}
show $\lim_{\varepsilon\to0}\mathcal{C}_{\varepsilon}=\mathcal{C}_{0}$. 

\begin{lem} \label{lem:nutz1.9} Let $\mu_{n}\in\mathcal{M}$. Suppose
that $\lim_{n}H\left(\mu_{n}\right)=:a\in\mathbb{R}$ exists and that
\[
\limsup_{m,n\to\infty}H\left(\mu_{m,n}\right)\ge a
\]
 for $\mu_{m,n}:=\left(\mu_{m}+\mu_{n}\right)/2$. Then $\left\{ \mu_{n}\right\} $
converges weakly. \end{lem}

\begin{proof} Let $D<\infty$ be the diameter of $\mathcal{X}\times\mathcal{T}$.
By \citet[Theorem 6.15]{villani2009optimal}, $d_{W}\left(\mu,\nu\right)$
is bounded by $D\left\lVert \mu-\nu\right\rVert _{\mathrm{TV}}$.
By following the same arguments of \citet[Lemma 1.9]{nutzIntroductionEntropicOptimal2022},
we obtain
\[
\lim_{m,n\to\infty}\left\lVert \mu_{m}-\mu_{n}\right\rVert _{\mathrm{TV}}=0.
\]
Thus it follows that $\lim_{m,n\to\infty}d_{W}\left(\mu_{m},\mu_{n}\right)=0$.\end{proof}

The next result is an adaptation of \citet[Theorem 1.10]{nutzIntroductionEntropicOptimal2022}.
We use weak convergence as the mode of convergence, whereas the original
proof uses convergence in total variation.

\begin{lem} \label{lem:nutz1.10} Let $\mathcal{Q}\subset\mathcal{M}$
be a convex and closed subset. There exists a unique $\mu_{*}\in\mathcal{Q}$
such that 
\[
H(\mu_{*})=\inf_{\mu\in\mathcal{Q}}H(\mu)\in[0,\infty).
\]
\end{lem}

\begin{proof} Let $\mu_{n}\in\mathcal{Q}$ be such that $H(\mu_{n})\to\inf_{\mu^{\prime}\in\mathcal{Q}}H(\mu^{\prime})$.
By convexity of $\mathcal{Q}$, we have $\mu_{m,n}:=\left(\mu_{m}+\mu_{n}\right)/2\in\mathcal{Q}$
and hence $H(\mu_{m,n})\ge\inf_{\mu\in\mathcal{Q}}H(\mu)$ for all
$m,n$. Lemma \ref{lem:nutz1.9} shows that $\left\{ \mu_{n}\right\} $
converges weakly to some $\mu_{*}$. By the continuity of $\mu\mapsto H(\mu)$,
$\mu_{*}$ is a minimizer of $\inf_{\mu^{\prime}\in\mathcal{Q}}H(\mu^{\prime})$.
Uniqueness follows from the strict convexity of $H$. \end{proof}

The next result is an adaptation of \citet[Proposition 1.17]{nutzIntroductionEntropicOptimal2022}.

\begin{lem} \label{lem:nutz1.17} Consider a decreasing sequence
of sets $\mathcal{Q}_{n}\subset\mathcal{M}$ that are convex and closed,
and let $\mathcal{Q}:=\cap_{n}\mathcal{Q}_{n}$. Let $\mu_{n}=\arg\min_{\mu\in\mathcal{Q}_{n}}H(\mu)$
be the minimizer of $\mathcal{Q}_{n}$. Then 
\[
\mu_{n}\to\mu_{*}\text{ weakly,}\quad\text{and}\quad H(\mu_{n})\to H(\mu_{*}),
\]
where $\mu_{*}=\arg\min_{\mu^{\prime}\in\mathcal{Q}}H(\mu^{\prime})$.
\end{lem}

\begin{proof} Note that the inclusion $\mathcal{Q}_{n}\supset\mathcal{Q}_{n+1}\supset\mathcal{Q}$
implies that $H(\mu_{n})$ is increasing and $H(\mu_{n})\le\inf_{\mu^{\prime}\in\mathcal{Q}}H(\mu^{\prime})$.
Since any increasing and bounded-above sequence is convergent, we
have $\lim H(\mu_{n})\le\inf_{\mu^{\prime}\in\mathcal{Q}}H(\mu^{\prime})<\infty$.
For $m\ge n$, we have $\mu_{m,n}:=\left(\mu_{m}+\mu_{n}\right)/2\in\mathcal{Q}_{n}$
by convexity. Then $H(\mu_{m,n})\ge H(\mu_{n})$. Thus $\limsup_{m,n\to\infty}H(\mu_{m,n})\ge\lim H(\mu_{n})$.
Since $\lim H(\mu_{n})<\infty$, Lemma \ref{lem:nutz1.9} implies
that $\mu_{n}$ converges weakly to some limit $\mu$. By the continuity
of $H$ on $\mathcal{M}$, 
\[
H(\mu)=\lim_{n}H(\mu_{n})\le\inf_{\mu^{\prime}\in\mathcal{Q}}H(\mu^{\prime}).
\]
Thus we obtain $\mu\in\arg\min_{\mathcal{\mu^{\prime}\in Q}}H(\mu^{\prime})$.
By the uniqueness of the minimizer shown in Lemma \ref{lem:nutz1.10},
we have $\mu=\mu_{*}$; i.e., $\mu_{n}$ converges weakly to $\mu_{*}$.
\end{proof}

The next result is an adaptation of \citet[Theorem 5.5]{nutzIntroductionEntropicOptimal2022}.

\begin{lem} \label{lem:nutz5.5} Suppose that $\lim_{\varepsilon\to0}\mathcal{C}_{\varepsilon}=\mathcal{C}_{0}$.
Let $\mu_{\varepsilon}$ be the optimizer of ($\varepsilon$EOT).
Then, 
\[
\mu_{\varepsilon}\to\mu_{*}\text{ weakly as }\varepsilon\downarrow0,\quad\text{and}\quad H(\mu_{\varepsilon})\to H(\mu_{*}),
\]
where $\mu_{*}=\arg\min_{\mu\in\mathcal{M}_{opt}}H(\mu)$.\end{lem}

\begin{proof} The additive form of ($\varepsilon$EOT) and the optimality
of the couplings imply that 
\[
H(\mu_{\varepsilon})\le H(\mu_{\varepsilon^{\prime}})\quad\text{and}\quad\int c\mathrm{d}\mu_{\varepsilon}\ge\int c\mathrm{d}\mu_{\varepsilon^{\prime}}\quad\text{for }\varepsilon\ge\varepsilon^{\prime}>0.
\]
Denote $\mathcal{Q}:=\mathcal{M}_{opt}$ and 
\[
\mathcal{Q}_{\varepsilon}:=\left\{ \mu\in\mathcal{M}:\int c\mathrm{d}\mu\le\int c\mathrm{d}\mu_{\varepsilon}\right\} .
\]
Note that $\mathcal{Q}_{\varepsilon}$ is a closed convex set, and
$\mu_{\varepsilon}=\arg\min_{\mu\in\mathcal{Q}_{\varepsilon}}H(\mu)$.\footnote{Suppose, by contradiction, that there exists $\mu\in\mathcal{Q}_{\varepsilon}$
such that $H(\mu)<H(\mu_{\varepsilon})$. Then 
\[
\int c\mathrm{d}\mu+\varepsilon H(\mu)<\int c\mathrm{d}\mu_{\varepsilon}+\varepsilon H(\mu_{\varepsilon}),
\]
which contradicts with the optimality of $\mu_{\varepsilon}$.} Then $\int c\mathrm{d}\mu_{\varepsilon}\ge\int c\mathrm{d}\mu_{\varepsilon^{\prime}}$
implies that $\mathcal{Q}_{\varepsilon}\supset\mathcal{Q}_{\varepsilon^{\prime}}$
for $\varepsilon\ge\varepsilon^{\prime}$. Next, we claim that $\mathcal{Q}=\cap_{\varepsilon}\mathcal{Q}_{\varepsilon}$.
It is easy to see $\mathcal{Q}\subset\cap_{\varepsilon}\mathcal{Q}_{\varepsilon}$.
For the other direction, take any $\mu\in\cap_{\varepsilon}\mathcal{Q}_{\varepsilon}$.
Then we have $\mu\in\mathcal{Q}$ because
\[
\int c\mathrm{d}\mu\le\int c\mathrm{d}\mu_{\varepsilon}\le\mathcal{C}_{\varepsilon}\to\mathcal{C}_{0}.
\]
Then applying Lemma \ref{lem:nutz1.17} completes the proof.\end{proof}

Thus, it remains to show that $\lim_{\varepsilon\to0}\mathcal{C}_{\varepsilon}=\mathcal{C}_{0}$.
The next result is an adaptation of \citet[Lemma 5.2]{nutzIntroductionEntropicOptimal2022}.

\begin{lem} \label{lem:nutz5.2} Suppose that given $\eta>0$, there
exists $\mu^{\eta}\in\mathcal{M}$ with $\int c\mathrm{d}\mu^{\eta}\le\mathcal{C}_{0}+\eta$
and $H(\mu^{\eta})<\infty$. Then $\lim_{\varepsilon\to0}\mathcal{C}_{\varepsilon}=\mathcal{C}_{0}$.
\end{lem}

\begin{proof} Given $\eta>0$, we have 
\[
\mathcal{C}_{\varepsilon}\le\int c\mathrm{d}\mu^{\eta}+\varepsilon H(\mu^{\eta})\le\mathcal{C}_{0}+\eta+\varepsilon H(\mu^{\eta}).
\]
Thus $\lim_{\varepsilon\to0}\mathcal{C}_{\varepsilon}\le\mathcal{C}_{0}+\eta$.
Since $\eta>0$ is arbitrary, we are done. \end{proof}

The next result is an adaptation of \citet[Lemma 5.4]{nutzIntroductionEntropicOptimal2022}. 

\begin{lem} \label{lem:nutz5.4} Let $c$ be continuous and bounded.
Then $\lim_{\varepsilon\to0}\mathcal{C}_{\varepsilon}=\mathcal{C}_{0}$.
\end{lem}

\begin{proof} Let $\eta>0$ and $\mu\in\mathcal{M}$ an optimal transport
for (OT). By \citet[Lemma 5.3]{nutzIntroductionEntropicOptimal2022},
there exists $\mu^{\eta}\in\mathcal{M}$ such that 
\[
\left|\int c\mathrm{d}\mu^{\eta}-\int c\mathrm{d}\mu\right|\le\eta.
\]
Note that $H(\mu^{\eta})<\infty$. Then applying Lemma \ref{lem:nutz5.2}
yields the conclusion. \end{proof}

Then by Lemma \ref{lem:nutz5.5}, the conclusion of Proposition \ref{prop:nutz}
follows.

\section{\label{sec:Semipara}Optimality in semiparametric models}

We generalize the setup presented in the main texts to allow more
flexible sampling distributions for observable data. Our setup here
basically follows \citet[Section 5]{christensen2023optimal}. Assume
that data $Z^{n}=(Z_{1},\dots,Z_{n})$ are i.i.d. and $Z_{i}$ follows
the distribution $P_{\theta,\eta}$ indexed by $\theta\in\Theta\subset\mathbb{R}^{k}$
and $\eta\in\mathcal{H}$, where $\eta$ is a possibly infinite-dimensional
nuisance parameter. For instance, in a GMM model, $\mathcal{H}$ is
the set of marginal distributions $\eta$ of $Z_{i}$ where for each
$\eta\in\mathcal{H}$, there exists $\theta\in\Theta$ such that $\eta$
satisfies the moment restriction $\int g(\theta,z)\mathrm{d}\eta=0$,
given some known vector function $g$. 

It is said that $\mathcal{P}=\{P_{\theta,\eta}:\theta\in\Theta,\eta\in\mathcal{H}\}$
has the \emph{least favorable submodels} at $(\theta,\eta)$ if there
exist an open neighborhood $\Theta_{\theta,\eta}$ of $\theta$ and
a map $\Theta_{\theta,\eta}\ni t\mapsto\eta_{t}\in\mathcal{H}$ such
that the parametric submodel $\{P_{t,\eta_{t}}:t\in\Theta_{\theta,\eta}\}$
has the density function $p_{t,\eta_{t}}$ with respect to a common
dominating measure $\nu$ and satisfies the DQM condition 
\[
\int\left[\sqrt{p_{\theta+h,\eta_{\theta}}}-\sqrt{p_{\theta,\eta_{\theta}}}-\frac{1}{2}h^{\top}\dot{\ell}_{\theta,\eta}p_{\theta,\eta_{\theta}}\right]^{2}\mathrm{d}\nu=o(\left\lVert h\right\rVert ^{2}),\quad\text{as }h\to0,
\]
where $\dot{\ell}_{\theta,\eta}:\mathcal{Z}^{n}\to\mathbb{R}^{k}$
is the efficient score function for $\theta$. Thus, the parametric
submodel $\{P_{t,\eta_{t}}:t\in\Theta_{\theta,\eta}\}$ achieves the
semiparametric efficiency bound by the inverse of $I_{\theta,\eta}:=\int\dot{\ell}_{\theta,\eta}\dot{\ell}_{\theta,\eta}^{\top}\mathrm{d}P_{\theta,\eta_{\theta}}$.
For each $(\theta,\eta)$, the least favorable submodels need not
to be unique. Picking one of them gives no loss of generality because
they all behave in the same manner asymptotically.

Following the parametric model, we assume that the planner's utility
function $w$ only depends on $\theta$, and not on the nuisance parameter
$\eta$.

\subsection{Decision theoretic framework and rules}

Fix $(\theta_{0},\eta_{0})\in\Theta\times\mathcal{H}$. Consider a
least favorable submodel $\{P_{\beta(t)}:t\in\Theta_{\theta_{0},\eta_{0}}\}$,
where $\beta(t)=(t,\eta_{t})$. Under the reparametrization $t=\theta_{0}+h/\sqrt{n}=\theta_{nh}$,
$Z^{n}$ follows the distribution $P_{\beta(\theta_{nh})}^{n}$. We
denote $\stackrel{h}{\rightsquigarrow}$ by the weak convergence along
the path $P_{\beta(\theta_{nh})}^{n}$, $\stackrel{h}{\to}$ by the
convergence in probability along $P_{\beta(\theta_{nh})}^{n}$, and
$\stackrel{0}{\to}$ by the convergence in probability along $P_{\theta_{0},\eta_{0}}^{n}$. 

We define the class of the sequences of rules by 
\[
\mathcal{D}:=\left\{ \left\{ \mu_{n}\right\} :\mu_{n}(Z^{n})\stackrel{h}{\rightsquigarrow}Q_{\theta_{0},h}\text{ and }\sqrt{n}P_{\beta(\theta_{nh})}^{n}(\mu_{n}\in A_{0})\to0\text{ as }n\to\infty\quad\forall h\in\mathbb{R}^{k},\forall\theta_{0}\in\Theta\right\} .
\]

The optimality criterion in this semiparametric model is based on
the least favorable submodels. The \emph{risk }associated with the
map $Z^{n}\mapsto\mu(Z^{n})\in\mathcal{M}$ at $(\theta,\eta)\in\Theta\times\mathcal{H}$
is given by 
\[
R(\mu,(\theta,\eta)):=\mathbb{E}_{P_{\theta,\eta}^{n}}\left[W_{\mathcal{M}}^{*}(\theta)-W(\theta,\mu(Z^{n}))\right],
\]
where the expectation is taken with respect to the sampling distribution
$P_{\theta,\eta}^{n}$ of $Z^{n}$. Let $\pi$ be any prior density
function on $\Theta$ that is continuous and positive at $\theta_{0}$.
Then, a sequence of rules $\left\{ \mu_{n}^{*}\right\} \in\mathcal{D}$
is said to be \emph{(semiparametrically) average optimal} if $\left\{ \mu_{n}^{*}\right\} $
attains the infimum of the asymptotic risk function: 
\begin{equation}
\inf_{\left\{ \mu_{n}\right\} \in\mathcal{D}}\liminf_{n\to\infty}\int\sqrt{n}R(\mu_{n},\beta(\theta_{nh}))\pi(\theta_{nh})\mathrm{d}h.
\end{equation}

\subsection{Quasi-Bayesian implementation of the Bayesian rules}

We replace a posterior function specified in the parametric setup
by a quasi-posterior. Let $\hat{\theta}_{n}$ be the (semiparametrically)
efficient estimator of $\theta$, and $\hat{I}_{n}^{-1}$ be a consistent
estimator of the asymptotic covariance $I_{\theta,\eta}^{-1}$. We
combine the limited-information quasi-likelihood $N(\hat{\theta}_{n},(n\hat{I}_{n})^{-1})$
for $\theta$ and a prior $\pi$ on $\Theta$ to obtain the quasi-posterior
\[
\pi_{n}(\theta|Z^{n})\propto\exp\left(-\frac{1}{2}(\theta-\hat{\theta}_{n})^{\top}(n\hat{I}_{n})(\theta-\hat{\theta}_{n})\right)\pi(\theta).
\]
We compute the Bayesian rules using $\pi_{n}(\theta|Z^{n})$; i.e.,
\[
\mu_{n}^{B}(z)\in\mathcal{M}_{opt}(z):=\arg\max_{\mu\in\mathcal{M}}\int\sqrt{n}W(\theta,\mu)\mathrm{d}\pi_{n}(\theta|z).
\]
Following the parametric model, we construct a unique $\{\mu_{n}^{B}(z)\}$
where $\mu_{n}^{B}(z)$ minimizes the penalty function $H$ over $\mathcal{M}_{opt}(z)$.

\subsection{Optimality results}

As an analog for Assumption \ref{assu:cms25:assu2} in the parametric
models, we impose the following assumptions.

\begin{assu}\label{assu:cms25:assu3} 

(i) $\Theta$ is open. 

(ii) $\mathcal{P}$ has a least favorable submodel at each $(\theta_{0},\eta_{0})\in\Theta\times\mathcal{H}$. 

(iii) $I_{\theta_{0},\eta_{0}}$ is finite and nonsingular at each
$(\theta_{0},\eta_{0})\in\Theta\times\mathcal{H}$. 

(iv) For each $(\theta_{0},\eta_{0})\in\Theta\times\mathcal{H}$ and
each $h\in\mathbb{R}^{k}$, (iv-a) $\sqrt{n}P_{\beta(\theta_{nh})}^{n}(\lVert\hat{\theta}_{n}-\theta_{0}\rVert>\varepsilon)\to0$
for each $\varepsilon>0$ as $n\to\infty$, and (iv-b) there exists
$c\in(0,1)$ such that $\sqrt{n}P_{\beta(\theta_{nh})}^{n}(c\le\hat{\lambda}_{\min},\hat{\lambda}_{\max}\le c^{-1})\to0$
as $n\to\infty$, 

(v) For each $(\theta_{0},\eta_{0})\in\Theta\times\mathcal{H}$, (v-a)
$\sqrt{n}(\hat{\theta}_{n}-\theta_{0})\stackrel{h}{\rightsquigarrow}Z$
with $Z\sim N(h,I_{\theta_{0},\eta_{0}}^{-1})$ for all $h\in\mathbb{R}^{k}$
as $n\to\infty$, and (v-b) $\hat{I}_{n}\stackrel{0}{\to}I_{\theta_{0},\eta_{0}}$
as $n\to\infty$. \end{assu}

\begin{thm}\label{thm:optimality_semipara} Under Assumptions \ref{assu:cms25:assu3}
and \ref{assu:villani_thm6-9}--\ref{assu:CMS25:lem:8}, $\{\mu_{n}^{B}\}\in\mathcal{D}$
is average optimal. \end{thm}

\begin{proof} Once we fix the parameter $(\theta_{0},\eta_{0})\in\Theta\times\mathcal{H}$,
the least favorable submodel $\{P_{\beta(t)}:t\in\Theta_{\theta_{0},\eta_{0}}\}$
becomes a parametric model. Hence, only slight modifications from
the proof of Theorem \ref{thm:optimality} are needed. Specifically,
Lemmas \ref{lem:lower_bound} and \ref{lem:events_equivalence} follow
in the same manner. For Lemma \ref{lem:The-ex-post-Bayes}, we need
a modification to show $\sqrt{n}P_{\beta(\theta_{nh})}^{n}(\mu_{n}^{B}(Z^{n})\not\in A_{0})\to0$
for all $h\in\mathbb{R}^{k}$, which is given in Lemma \ref{lem:CMS25_lemma12}
below. For Lemmas \ref{lem:process_conv} and \ref{lem:epb_Qn}, we
need to use the quasi-posterior counterparts of the Bernstein-von
Mises theorem given by \citet[Lemma 5]{christensen2023optimal} and
Proposition \ref{prop:subpoly-vdv} given by \citet[Lemma A.5]{xu2024jmp}.
Auxiliary lemmas given in Appendix \ref{sec:auxiliary-lemmas} do
not need modifications. \end{proof}

\begin{lem}\label{lem:CMS25_lemma12} The Bayesian rule $\{\mu_{n}^{B}(Z^{n})\}$
satisfies $\sqrt{n}P_{\beta(\theta_{nh})}^{n}(\mu_{n}^{B}(Z^{n})\not\in A_{0})\to0$
as $n\to\infty$. \end{lem}

\begin{proof} From Lemma \ref{lem:The-ex-post-Bayes} (i), for any
$\theta_{0}\in\Theta$, there are $\overline{n}$ and $\varepsilon_{n}^{\prime}(\overline{n})$
(which is at the order of $n^{\alpha+1}$ for some $\alpha\ge1$)
such that for all $n\ge\overline{n}$,
\[
P_{\theta_{nh}}^{n}\left(\mu_{n}^{B}(z)\notin A_{0}\right)\le P_{\theta_{nh}}^{n}\left(\pi_{n}\left(N_{1/n}(\theta_{0})^{c}|z\right)>2\varepsilon_{n}^{\prime}\right)
\]
Under Assumption \ref{assu:cms25:assu3} (iii) and (iv), \citet[Lemma 12]{christensen2023optimal}
implies that 
\[
\sqrt{n}P_{\beta(\theta_{nh})}\left(\pi_{n}\left(N_{1/n}(\theta_{0})^{c}\right)>2\varepsilon_{n}^{\prime}\right)\to0
\]
 as $n\to\infty$. \end{proof}

\printbibliography

\end{document}